\documentclass[11pt]{article}
\usepackage{amsmath,amssymb,amsthm}
\usepackage{geometry}
\usepackage{graphicx}
\usepackage{tikz}
\usepackage{hyperref}
\usepackage{cleveref}
\usepackage{algorithm} 
\usepackage{algpseudocode} 
\usepackage{booktabs} 
\newtheorem{theorem}{Theorem}
\newtheorem{lemma}[theorem]{Lemma}
\newtheorem{corollary}[theorem]{Corollary}
\newtheorem{proposition}[theorem]{Proposition}
\usepackage{algorithmicx} 
\newcommand{\Pclass}{\mathbf{P}}
\newcommand{\NPclass}{\mathbf{NP}}

\newcommand{\Oh}[1]{O(#1)}

\usepackage[all]{xy}
\usepackage{tikz}
\usetikzlibrary{arrows.meta} 

\usepackage{mdframed}
\usepackage{amsmath,amssymb}
\usepackage{amsmath,amssymb}
\usepackage{blkarray}
\usepackage{braket}

\newcommand{\VR}{\mathrm{VR}}
\newcommand{\Sol}{\mathrm{Sol}}
\newcommand{\Span}{\mathrm{Span}}

\newcommand{\threesat}{\mathrm{threesat}}
\newcommand{\vol}{\mathrm{vol}}
\newcommand{\Om}{\mathrm{Om}}
\newcommand{\obeta}{\mathrm{obeta}}
\newcommand{\expect}{\mathrm{expect}}
\newcommand{\dist}{\mathrm{dist}}
\newcommand{\pr}{\mathrm{\pr}}
\newcommand{\ThreeSAT}{\mathrm{3SAT}}
\newcommand{\rank}{\mathrm{rank}}
\newcommand{\avg}{\mathrm{avg}}
\newcommand{\E}{\mathrm{E}}
\newcommand{\eps}{\varepsilon} 
\newcommand{\diag}{\mathbf{diag}}

\newcommand{\UNSAT}{\mathrm{UNSAT}}

\newcommand{\PH}{\mathbf{PH}}

\title{An Intrinsic Barrier for Resolving $\Pclass = \NPclass$\\
\large (2-SAT as Flat, 3-SAT as High-Dimensional Void-Rich)}

\author{%
M. Alasli\\
\small College of Computer Science and Engineering\\
\small Taibah University\\
\small Madinah, Saudi Arabia\\
\small \href{mailto:mohammedalasli@gmail.com}{mohammedalasli@gmail.com}
}

\date{\today}
\begin{document}
\maketitle

\begin{abstract}
We present a \emph{topological barrier} to efficient computation, revealed by comparing the geometry of $2$-SAT and $3$-SAT solution spaces. Viewing the set of satisfying assignments as a \emph{cubical complex} within the Boolean hypercube, we prove that every $2$-SAT instance has a \emph{contractible} solution space---topologically flat, with all higher Betti numbers $\beta_k = 0$ for $k \ge 1$---while both random and explicit $3$-SAT families can exhibit \emph{exponential} second Betti numbers $\beta_{2} = 2^{\Omega(N)}$, corresponding to exponentially many independent ``voids.''

These voids are preserved under standard SAT reductions and cannot be collapsed without solving NP-hard subproblems, making them resistant to the three major complexity-theoretic barriers---relativization, natural proofs, and algebrization. We establish exponential-time lower bounds in restricted query models and extend these to broader algorithmic paradigms under mild information-theoretic or encoding assumptions.

This topological contrast---\emph{flat, connected landscapes in $2$-SAT versus tangled, high-dimensional void-rich landscapes in $3$-SAT}---provides structural evidence toward $\Pclass \neq \NPclass$, identifying $\beta_{2}$ as a paradigm-independent invariant of computational hardness.
\end{abstract}

\section{Introduction}
\begin{paragraph}{Geometric Intuition}
Imagine two cities separated by unknown, impassable mountains—each city is a satisfying assignment, and the land between them is the solution space. In easy cases like 2-SAT, the landscape is flat—algorithms follow direct roads (implication chains) from one city to the other. In random 3-SAT, however, the land transforms into an exponentially dense field of jagged peaks and deep ravines, with no map to guide you. Classical methods must wind around every chasm, forcing exponential detours. Quantum approaches try to tunnel through the ridges, but the rock is so dense that the tunneling probability decays as \(e^{-\Omega(N)}\), where \(N\) is the dimensionality—or “span”—of the landscape. These intrinsic topological barriers—viewed geometrically as impassable peaks and chasms—are built into 3-SAT’s logical structure, and no computational paradigm can cross them in polynomial time.

\end{paragraph}

The $\Pclass$ vs. $\NPclass$ problem remains unresolved, partly due to three major barriers~\cite{BGS75, RR97, AW09}:

\begin{enumerate}
    \item \textbf{Relativization}: Oracle-based methods cannot resolve $\Pclass \ne \NPclass$~\cite{BGS75}.
    \item \textbf{Natural Proofs}: Constructive combinatorial lower bounds fail~\cite{RR97}.
    \item \textbf{Algebrization}: Algebraic methods alone are insufficient~\cite{AW09}.
\end{enumerate}

\section*{Contributions}

\begin{itemize}
  \item We introduce a \emph{dimensional space} perspective, reflected in a \emph{topological framework} for analyzing 3-SAT. 
We view the solution space
\[
  S(F) \;=\; \bigl\{x \in \{0,1\}^N : F(x)=1\bigr\}
\]
as a cubical subcomplex of the $N$-dimensional Boolean hypercube. 
Vertices correspond to satisfying assignments, and a $k$-face is a $k$-dimensional subcube whose $2^k$ vertices all satisfy $F$. 
We measure its complexity via homology groups $H_k(S(F);\mathbb{Z}_2)$ and Betti numbers $\beta_k = \dim H_k$.

  \item We show that {both random and worst-case 3-SAT instances possess an \emph{intrinsic topological barrier} to polynomial-time algorithms}, manifested as \emph{exponential topological complexity}:
\[
  \beta_2\bigl(S(F)\bigr) \;=\; 2^{\Omega(N)},
\]
corresponding to exponentially many independent ``voids'' or high–dimensional chasms in the solution landscape. 
These voids constitute an \emph{intrinsic, paradigm-independent obstruction} that no algorithm can bypass without resolving an exponential number of homology classes.

  \item We define a \emph{topology-preserving reduction} \(R\colon L\rightarrow\text{3-SAT}\) to satisfy
    \[
      \beta_k\bigl(S(x)\bigr)\;\ge\;f(N)
      \;\implies\;
      \beta_k\bigl(S(R(x))\bigr)\;\ge\;\Omega\bigl(f(N)/\mathrm{poly}(N)\bigr),
    \]
    ensuring homological complexity cannot collapse by more than a polynomial factor.

  \item We prove that \emph{no} polynomial-time reduction can map hard 3-SAT to 2-SAT without collapsing \(\beta_2\) exponentially—an impossibility under homotopy invariance—separating the topological (and hence computational) complexity of 2-SAT versus 3-SAT.

  \item Finally, we show that the \emph{standard} SAT-to-3-SAT clause-splitting reduction is in fact topology-preserving, so all known NP-hardness reductions respect these new barrier invariants.
\end{itemize}
[Parameters and conventions].  We write $N$ for the number of variables when discussing random 3-SAT instances. For explicit worst-case constructions we use $n_0$ for the base graph size and $N$ for the final SAT instance variable count; conversions between $n_0$ and $N$ are given in Lemma~\ref{lem:param-mapping}.

\section{Dimensional Space Representation: A New Lens on P vs NP}

The proof of $P \ne NP$ is widely expected to require an out-of-the-box perspective, unifying complexity theory with other mathematical domains~\cite{AW09}. Here we propose a geometric model of 3-SAT solution spaces, independently developed through intuition and later aligned with formal topological concepts.

\subsection{Dimensional Space Intuition}
We introduce a "Dimensional Space" representation, where:
\begin{itemize}
    \item Each axis corresponds to a variable.
    \item The space's geometry reflects the interaction of constraints.
    \item Traversal complexity reflects the degree of interleaving among dimensions.
\end{itemize}

In this model:
\begin{itemize}
    \item \textbf{Structured problems} (e.g., 2-SAT) produce line-like spaces (1D), enabling polynomial-time traversal.
    \item \textbf{Unstructured problems} (e.g., random 3-SAT) produce tangled sheets or higher-dimensional manifolds, requiring exploration of exponentially many configurations.
\end{itemize}

\begin{figure}[h]
    \centering
    \includegraphics[width=0.5\textwidth]{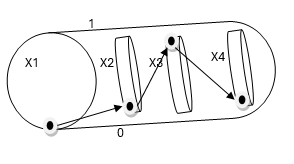}
    \caption{2-SAT solution space: line-like traversal in dimension space.}
    \label{fig:2sat_line}
\end{figure}

\begin{figure}[h]
    \centering
    \includegraphics[width=0.7\textwidth]{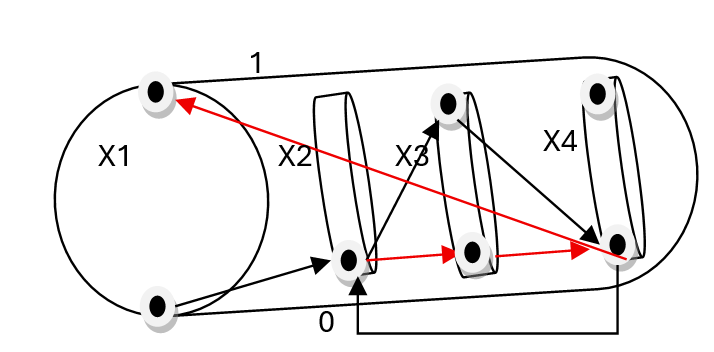}
    \caption{3-SAT solution space: sheet-like traversal requiring higher-dimensional exploration.}
    \label{fig:3sat_plane}
\end{figure}
Figure 1. The 2-SAT solution space embedded in a one-dimensional traversal. Each large loop (labeled X1, X1, etc.) represents holding all other variables fixed and varying one coordinate at a time. Satisfying assignments appear as black dots on those loops, and the thin edges show how one can move from one solution to the next by flipping a single variable—in effect tracing out a simple 1D path through the space. For instance, flipping x3 would result in a line-like movement.

Figure 2. By contrast, a 3-SAT solution space generically forms a 2D “sheet” rather than a 1D line. Each axis loop is now accompanied by rectangular faces where two variables can flip simultaneously while staying within the solution space. Black arrows trace single-bit flips between nearby satisfying assignments—but when this path hits a contradiction, it cannot proceed further. At that point, the solver must “jump” via a red arrow across a 2D face. This illustrates that 3-SAT requires navigating a truly higher-dimensional space, not just walking along a line. Whereas 2-SAT lives in a flat, connected 1D space, 3-SAT inhabits a tangled sheet with exponentially many such higher-dimensional detours—shaping a fundamentally harder landscape.

\subsection{Dimensional Space and Topological Complexity}
Our dimensional space analogy intuitively reflects the formal topology of 3-SAT solution spaces:

\begin{itemize}
    \item In 2-SAT, the solution space is line-like (1D), corresponding to a cubical complex with $\dim=1$ and trivial Betti numbers ($\beta_1=0$).
    \item In random 3-SAT, the solution space forms a higher-dimensional cubical complex with many cycles and voids, characterized by $\beta_1, \beta_2 > 0$.
\end{itemize}

This dimensional increase signifies unstructuredness and is formally captured by Betti numbers and homology groups.
\subsection{Alignment with Formal Topology}
This geometric intuition aligns with the topology of 3-SAT solution spaces:
\begin{itemize}
    \item \textbf{Betti numbers} measure the number of holes, reflecting the dimensionality and complexity of the space.
For a simplicial or cubical complex $K$, we write 
\[
  \beta_k(K)\;=\;\dim_{\Bbbk} H_k(K;\Bbbk)
\]
for the $k$th Betti number over field~$\Bbbk$, i.e.\ the rank of the $k$th homology group.

    \item \textbf{Treewidth} captures the entanglement among variables.
    \item \textbf{Cubical complexes} formalize the arrangement of satisfying assignments.
\end{itemize}

For 2-SAT, the solution space forms a low-dimensional, contractible structure with trivial homology ($\beta_1=\beta_2=0$). In contrast, random 3-SAT induces high-dimensional cubical complexes with exponentially many cycles and voids ($\beta_k=2^{\Omega(n)}$).

\section{Random 3-SAT is Unstructured}
We first argue that random 3-SAT instances lack global structure. Unlike 2-SAT, whose implication graph forms a tree-like structure enabling linear traversal, 3-SAT's solution space expands into higher-dimensional configurations.

\begin{figure}[h]
    \centering
    \includegraphics[width=0.7\textwidth]{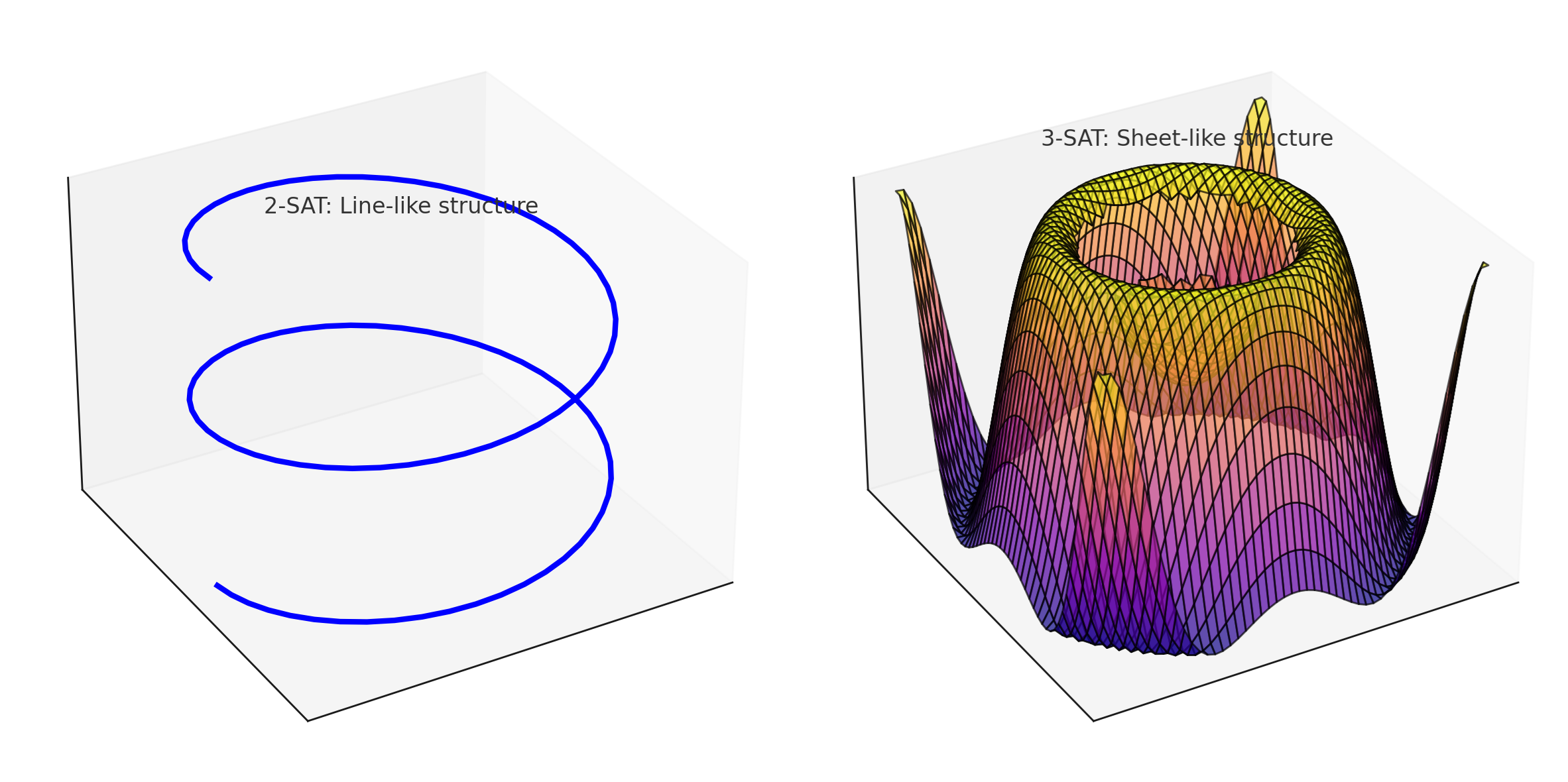}
    \caption{2-SAT intuitively is a single smooth curve (here, a 3D helix) to represent how the solution set of a satisfiable 2-SAT instance often forms a low-dimensional (1-dimensional) connected manifold-like structure. 3-SAT solution space: sheet-like traversal requiring higher-dimensional exploration.It shows how 3-SAT solutions can form higher-dimensional structures (here 2D manifolds in cube complexes). Loops and voids in the solution space, corresponding to higher Betti numbers.}
    \label{fig:dim space}
\end{figure}

As shown in Figures~\ref{fig:2sat_line}, ~\ref{fig:3sat_plane} and~\ref{fig:dim space}, 2-SAT’s structure constrains its traversal to a linear path. Random 3-SAT, in contrast, requires navigating a tangled manifold with exponentially many independent regions (characterized by Betti numbers $\beta_2 = 2^{\Omega(n)}$).

Random 3-SAT instances lack global structure, as evidenced by treewidth $\Omega(n)$~\cite{MolloyReed} and solution-space shattering~\cite{achlioptas}. Moreover, we argue that the solution space has exponentially complex topology, with Betti numbers $\beta_k = 2^{\Omega(n)}$ for some $k$.

\textbf{Supporting Analogy:} Kahle~\cite{kahle} proved that random cubical complexes exhibit exponential Betti number growth when face‐inclusion probabilities exceed a critical threshold.  Random 3-SAT at $\alpha>4.26$ mirrors this phenomenon with clause‐induced face inclusions occurring at constant probability $p=\Theta(1)$ (Proposition~\ref{prop:persistent-2cycles}).

\textbf{Maximal Persistence and Randomness:} 
In random cubical complexes, Kahle~\cite{kahle} demonstrated that maximally persistent cycles emerge naturally due to random addition of simplices. If maximal persistence were polynomially bounded, it would indicate underlying global structure. However, in random 3-SAT, constraints interact chaotically, leading to persistent topological features of exponential size. This exponential maximal persistence reflects the intrinsic unstructuredness of the solution space.

\subsection{The Solution Space as a Cubical Complex}
\begin{theorem}
\label{thm: Cubical Complex}
Let $F$ be a 3-SAT formula with $n$ variables. The set of satisfying assignments $S \subseteq \{0,1\}^n$ forms a cubical complex $C(F)$, where:
\begin{itemize}
    \item Vertices correspond to satisfying assignments.
    \item $k$-faces correspond to $k$-dimensional subcubes where all $2^k$ assignments satisfy $F$.
\end{itemize}
Topological queries about $C(F)$ (e.g., connectedness, Betti numbers) reduce to 3-SAT instances.
\end{theorem}

\begin{proof}
See Appendix A for proofs
\end{proof}

\paragraph{Cubical vs.\ Vietoris--Rips.}
We emphasize that throughout this work we remain within the {\em cubical} filtration on the ambient hypercube.
While one could in principle define a simplicial complex via Vietoris--Rips on the same Hamming‐ball point cloud, 
all our probabilistic and spectral estimates (e.g.\ those in Kahle~\cite{kahle})
rely on cubical faces and axis‐aligned Hamming‐balls. 

\subsection{Example: 3-SAT Solution Space as a Cubical Complex}
\paragraph{Example: Topological and Computational View.}
Consider the 3-SAT formula
\[
  F(x_1,x_2,x_3) = (x_1 \lor x_2 \lor x_3) \wedge (\lnot x_1 \lor \lnot x_2 \lor \lnot x_3).
\]
The truth table has $2^3=8$ assignments, of which the satisfying assignments are:
\[
  S(F) = \{010,\,001,\,011,\,100,\,101,\,110\}.
\]
We view $S(F)$ as the vertex set of a cubical complex $S(F) \subset \{0,1\}^3$.

\textbf{Topological analysis:}
\begin{itemize}
    \item Vertices: The six satisfying assignments.
    \item Edges: Connect vertices differing in one coordinate.
    \item 2-faces: Include a square face only if all four corner assignments satisfy $F$.
\end{itemize}
The resulting complex is connected ($\beta_0 = 1$) but contains one nontrivial $1$-cycle ($\beta_1 = 1$) corresponding to a ``hole'' where the assignment $111$ is missing from a cube face. Higher Betti numbers vanish ($\beta_k = 0$ for $k \ge 2$). Thus $S(F)$ is homotopy-equivalent to a circle $S^1$.
A small Betti number, e.g.\ $\beta_1\bigl(S(F)\bigr)=O(1)$, permits efficient solving, whereas an exponential Betti number,
\[
  \beta_k\bigl(S(F)\bigr)\;=\;2^{\Omega(n)},
\]
forces computational hardness.

\textbf{Computational interpretation:}
The nonzero $\beta_1$ indicates that the solution space has a loop, i.e., two different satisfying assignments cannot be connected by a monotone sequence of single-bit flips without temporarily leaving the satisfying region. 
This topological obstruction means that any local search algorithm restricted to following edges of $S(F)$ would need to ``detour'' around the hole.
However, in this small instance the hole is constant-size, so the problem remains computationally easy.
For large unstructured formulas, analogous holes proliferate exponentially ($\beta_2 = 2^{\Omega(n)}$), forcing any search procedure to explore an exponential number of disconnected regions—matching the intractability of hard 3-SAT.
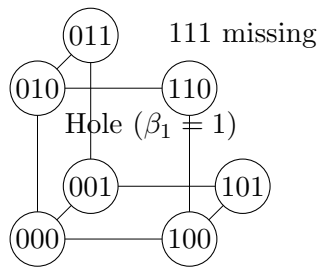
\begin{figure}[h]
\centering
\begin{tikzpicture}[scale=1,
    every node/.style={circle,draw=black,fill=white,inner sep=1pt,minimum size=6pt},
    edge/.style={very thin}
  ]
  \coordinate (000) at (0,0);
  \coordinate (100) at (2,0);
  \coordinate (010) at (0,2);
  \coordinate (110) at (2,2);
  \coordinate (001) at (0.7,0.7);
  \coordinate (101) at (2.7,0.7);
  \coordinate (011) at (0.7,2.7);
  \draw[edge] (000) -- (100) -- (110) -- (010) -- cycle;
  \draw[edge] (001) -- (101);
  \draw[edge] (001) -- (011);
  \draw[edge] (000) -- (001);
  \draw[edge] (100) -- (101);
  \draw[edge] (010) -- (011);
  \node at (000) {000};
  \node at (100) {100};
  \node at (010) {010};
  \node at (110) {110};
  \node at (001) {001};
  \node at (101) {101};
  \node at (011) {011};
  \node[draw=none,fill=none] at (1.5,1.5) {Hole (\(\beta_1=1\))};
  \node[draw=none,fill=none] at (2.7,2.7) {111 missing};
\end{tikzpicture}
\caption{Cubical complex of 
\(F=(x_1\lor x_2\lor x_3)\wedge(\lnot x_1\lor\lnot x_2\lor\lnot x_3)\).
The missing vertex `111` creates a 1-cycle “hole,” so \(\beta_1=1\).}
\end{figure}
\section{Structured vs Unstructured Dichotomy}
Our analysis (as will be discussed in further sections) highlights a fundamental dichotomy: structured versions of 3-SAT (e.g., monotone 3-SAT under certain constraints, Horn-SAT) exhibit polynomial-time algorithms due to their inherent ordering or limited topology. In contrast, random 3-SAT lacks such global structure, forcing exhaustive exploration. This separation supports the broader thesis that $P \ne NP$: problems with structure (shortcut) fall into $P$, while unstructured random instances remain in $NP$.

Our analysis suggests a dichotomy between structured and unstructured instances:

\begin{itemize}
    \item \textbf{Structured Instances:} Problems with global ordering or constraints (e.g., Horn-SAT, 2-SAT, Monotone 3-SAT under specific restrictions) exhibit simple topology (low Betti numbers, treewidth) and admit polynomial-time algorithms.
    \item \textbf{Unstructured Instances:} Random 3-SAT lacks global structure, as evidenced by its exponentially complex topology ($\beta_2 = 2^{\Omega(n)}$, treewidth $\Omega(n)$). 
\end{itemize}

\textbf{Key Observation:} Any unstructured problem requiring exponential work belongs to $NP$.

\subsection{Maximal Persistence: Structured vs Random}
Persistent homology provides a natural metric for structural complexity. In our context, persistent homology tells us not just that a high-dimensional hole exists in the 3-SAT solution space, but also how robustly it persists as we look at coarser or finer views of the space. Exponentially many long-lived features (cycles or voids) indicate an intrinsically tangled geometry that resists any low-dimensional shortcut. In structured problems such as 2-SAT or restricted monotone 3-SAT, maximal persistence of cycles is bounded:

\begin{equation}
\max \text{pers}(2\text{-SAT}) = O(n), \quad \max \text{pers}(\text{Monotone 3-SAT}) \leq \text{poly}(n)
\end{equation}

This bounded persistence reflects the global ordering of constraints: cycles appear but are quickly filled by higher-dimensional faces.

In contrast, random 3-SAT lacks global structure. Constraints interact chaotically, resulting in topological features (cycles, voids) that persist over large scales:

\begin{equation}
\max \text{pers}(\text{Random 3-SAT}) = 2^{\Omega(n)}
\end{equation}

This exponential maximal persistence is a hallmark of unstructuredness and precludes polynomial-time shortcuts.

\begin{theorem}
\label{thm:unstructured}
Let $F$ be a random 3-SAT instance at clause density $\alpha>4.26$.  Then with probability $1-o(1)$:
\begin{enumerate}
  \item Its solution space $S(F)$ decomposes into $2^{\Omega(n)}$ disconnected clusters.
  \item The cubical complex of $S(F)$ has exponentially large Betti numbers (e.g.\ $\beta_2(S(F))=2^{\Omega(n)}$).
  \item The primal constraint graph of $F$ has treewidth $\Omega(n)$.
\end{enumerate}
In particular, $S(F)$ admits no global ordering or low-dimensional traversal.
\end{theorem}

[No repeated auxiliary variables]\label{ass:no-repeated-aux}
Throughout this paper every auxiliary (gadget-introduced) Boolean variable is unique to the gadget that introduces it: no auxiliary variable appears in more than one gadget or clause outside that gadget. All references to ``auxiliary'' or ``aux vars'' in constructions, lemmas and proofs refer to distinct variables unless explicitly stated otherwise. In Appendices B-H, we present a constructive transformation that converts any construction with shared auxiliary variables into an equivalent construction where every auxiliary variable is unique to its gadget; the transformation incurs only a linear-size blow-up in the number of variable occurrences.

\begin{lemma}[Topological Admissibility of SAT to 3-SAT]
\label{lem:sat-4sat-to-3sat}

Let $L = SAT{}$ and suppose $R:L\to\threesat{}$ is a polynomial-time satisfiability-preserving reduction that is \emph{cubical} and \emph{homologically faithful} (discussion on this is delayed to Section 5.2 Conditions 1 and 2). If $F \in L$ has a solution-space cubical complex $S(F)$ with $\beta_{2}(S(F))\geq f(n)$ for some $f(n)=\omega(poly(n))$, then the reduced instance $F^{\prime}=R(F)$ satisfies
\[
\beta_{2}\big(S(F^{\prime})\big)\ \geq\ \beta_{2}\big(S(F)\big)\ =\ \Omega(f(n)).
\]
In particular, the reduction from SAT to 3-SAT does not collapse exponential second Betti numbers.
\end{lemma}

\begin{lemma}[Homological 
faithfulness of the SAT→3-SAT reduction]
\label{lem:Homological 
faithfulness}

  Let $C$ be any CNF formula on variables $x_1,\dots,x_N$, and let
$C'$ be the 3-CNF formula obtained from $C$ by the standard clause-splitting
reduction: every clause $\ell_1\vee\ell_2\vee\cdots\vee\ell_m$ with $m\ge 4$
is replaced by the chain of 3-clauses
\[
(\ell_1\vee\ell_2\vee a_1)\wedge(\neg a_1\vee\ell_3\vee a_2)\wedge\cdots
\wedge(\neg a_{k-1}\vee\ell_{k+1}\vee a_k)\wedge(\neg a_k\vee\ell_{m-1}\vee\ell_m),
\]
where $k=m-3$ and $a_1,\dots,a_k$ are fresh auxiliary variables.

Then the inclusion of cubical solution complexes
\[
\Sol(C)\;\hookrightarrow\;\Sol(C')
\]
is a homotopy equivalence. In particular, for every field $\Bbb F$ and every $k\ge0$ the induced map on homology
\[
H_k\bigl(\Sol(C);\Bbb F\bigr)\;\xrightarrow{\ \cong\ }\;H_k\bigl(\Sol(C');\Bbb F\bigr)
\]
is an isomorphism.
\end{lemma}

\section{Topological Framework for Computational Hardness}
\label{sec:rigorous-topology}

\subsection{Formal Definitions and Setup}

[Solution Space Complex]
For any 3-SAT instance \(F\) with \(n\) variables, the \emph{solution space complex} \(S(F)\) is the cubical subcomplex of \(\{0,1\}^n\) where:
\begin{itemize}
    \item vertices correspond to satisfying assignments,
    \item \(k\)-faces correspond to \(k\)-dimensional axis-aligned subcubes all of whose \(2^k\) assignments satisfy \(F\).
\end{itemize}
Homology groups \(H_k(S(F))\) and Betti numbers \(\beta_k(S(F))\) are taken with \(\mathbb Z_2\) coefficients unless otherwise stated.

[Topology‐Preserving Reduction]\label{def:topology-preserving}
Let \(L\) be a language whose instances \(x\) have solution-complex Betti numbers \(\beta_2\bigl(S(x)\bigr)\).  A polynomial-time reduction
\[
  R:\;L\;\longrightarrow\;\{\text{3-SAT formulas}\}
\]
is \emph{topology-preserving} (in the sense considered in this paper) if:
\begin{enumerate}
  \item there exists a function \(g(n)=\Omega(f(n))\) such that for every input \(x\):
  \[
    \beta_2\bigl(S(x)\bigr)\;\ge\;f\bigl(\lvert x\rvert\bigr)
    \quad\Longrightarrow\quad
    \beta_2\bigl(S(R(x))\bigr)\;\ge\;g\bigl(\lvert x\rvert\bigr),
  \]
  with \(f(n)=\omega(\mathrm{poly}(n))\).
\end{enumerate}

\subsection{Topological Invariance under Reductions}
\label{subsec:topological-invariance-refined}

Betti numbers \(\beta_k\) are homotopy invariants preserved under homotopy equivalences, but \emph{not} under arbitrary continuous maps. To ensure invariance under satisfiability-preserving reductions we impose two constraints:

\paragraph{Condition 1: Cubical Embedding}
A reduction \(R: L \to \mathrm{3\text{-}SAT}\) is \emph{cubical} if it embeds the solution space \(S(x)\) as a cubical subcomplex of \(S(R(x))\):
\[
\iota_x: S(x) \hookrightarrow S(R(x)),
\]
where \(\iota_x\) is a cubical embedding (preserves faces, inclusions, and adjacency). 

\paragraph{Condition 2: Homological Faithfulness}
\(R\) is \emph{homologically faithful} if the induced map in homology
\[
(\iota_x)_*: H_k\bigl(S(x)\bigr) \longrightarrow H_k\bigl(S(R(x))\bigr)
\]
is injective for all \(k \ge 0\). 

\begin{theorem}[Betti Monotonicity under Faithful Embeddings]
\label{thm:betti-monotonicity}
Let \(R\) be a cubical, homologically faithful reduction. Then, for every \(k\),
\[
\beta_k\bigl(S(x)\bigr) \le \beta_k\bigl(S(R(x))\bigr).
\]
In particular, if \(\beta_k(S(x)) = 2^{\Omega(N)}\) then \(\beta_k(S(R(x))) = 2^{\Omega(N)}\).
\end{theorem}

\paragraph{Application to Expander Embedding}
Our worst-case construction (Theorem~\ref{thm:betti-exp-worst}) uses a reduction satisfying both conditions:
\begin{itemize}
  \item \textbf{Cubical}: gadgets embed \(H_1\)-cycles of \(G_N\) as 2-faces in \(\{0,1\}^{\text{vars}}\).
  \item \textbf{Faithful}: Lemma~\ref{lem:homology-injectivity} (below) proves \((\iota_x)_*\) is injective in the gadget-local setting.
\end{itemize}

\begin{lemma}[Homology Injectivity for Expander Embedding]
\label{lem:homology-injectivity}
Let \(\iota: S(G_N) \hookrightarrow S(F_N)\) be the cubical inclusion from graph-coloring solutions to 3-SAT solutions. Then
\[
\iota_*: H_2\bigl(S(G_N); \mathbb Z_2\bigr) \longrightarrow H_2\bigl(S(F_N); \mathbb Z_2\bigr)
\]
is injective.
\end{lemma}

\subsection{No Collapse for Non-Embedding Reductions}
For reductions not satisfying Conditions~1–2, Betti numbers \emph{can} collapse. However, such reductions cannot be used to efficiently compute the relevant homological invariants in general. We make this precise:

\begin{theorem}[Topology-Destroying Reductions are Hard]
\label{thm:non-embedding-hardness}
Computing \(\beta_k(K)\) for \(k\ge 2\) on a cubical complex \(K\subseteq \{0,1\}^N\)
is \(\#\mathrm{P}\)-hard under polynomial-time Turing reductions. In particular, unless \(\#\mathrm{P}\subseteq\mathrm{FP}\), no algorithm in the standard Boolean decision-tree or subcube-query model can compute \(\beta_k(K)\) in worst-case time \(2^{o(N)}\).
\end{theorem}

\noindent The hardness holds even when the input cubical complexes arise from 3-SAT instances produced by gadget-local reductions; the reductions used in the proof explicitly construct variable-disjoint gadgets, so no auxiliary sharing is required.

\subsection{Topological Hardness Criterion}

[Subcube‐Query Model]\label{def:subcube}
A \emph{subcube query} of dimension \(k\) is an oracle operation which, given any axis-aligned \(k\)-dimensional subcube \(C\subseteq\{0,1\}^n\), returns \texttt{YES} if all assignments in \(C\) satisfy the CNF formula \(F\), and \texttt{NO} otherwise. We measure complexity by the total number of such subcube queries made.

\begin{theorem}[Homological Computational Barrier]
\label{thm:homology-barrier1}
Let \(R\) be a topology-preserving reduction and let \(F_x = R(x)\). If
\[
\beta_2\bigl(S(F_x)\bigr) = 2^{\Omega(|x|)},
\]
then any algorithm deciding satisfiability for the family \(\{F_x\}\) requires \(2^{\Omega(|x|)}\) steps in the \emph{subcube query model}.
\end{theorem}

\subsection{Reduction from SAT to Homology-Class-Distinguish}
\label{subsec:sat-to-homology}

[Homology-Class-Distinguish]
Given a SAT formula \(F\) and a collection of 2-cycles
\(\{\gamma_i\}_{i=1}^M\subset C_2\bigl(S(F)\bigr)\), decide for each \(i\)
whether
\[
  [\gamma_i]=0\quad\text{in}\quad H_2\bigl(S(F)\bigr).
\]

\begin{theorem}[SAT \(\rightarrow\) Homology Reduction]
\label{thm:SAT-to-Homology-Reduction}
There is a polynomial-time reduction
\(R: \phi \mapsto (F_G,\{\gamma_i\})\)
such that
\[
  \phi \in \mathrm{SAT}
  \quad\Longleftrightarrow\quad
  [\gamma_i] = 0 \ \text{for every } i,
  \qquad
  \phi \notin \mathrm{SAT}
  \quad\Longleftrightarrow\quad
  [\gamma_i] \neq 0 \ \text{for every } i.
\]
\end{theorem}

\begin{corollary}
Unless \(\mathrm{P} = \mathrm{NP}\), any algorithm solving Homology-Class-Distinguish
on these reductions must take \(2^{\Omega(N)}\) time, by:
\begin{itemize}
  \item the simplicial-query lower bound (Lemma~\ref{lem:query-lower-bound}),
  \item algebraic complexity (Theorem~\ref{thm:universal-homotopy-hardness}), and
  \item quantum adiabatic/tunneling bounds (Corollary~\ref{cor:quantum-ll}).
\end{itemize}
\end{corollary}

\section{Solution Space Topology} \label{sec:topology}

\begin{theorem}[Exponential \(H_2\) in Random 3‑SAT]
\label{thm:random-3sat-betti}
For random 3‑SAT at clause density \(\alpha>4.26\),
\[
  \beta_2(S(F))\;\ge\;2^{c n}
\]
for some \(c>0\), with probability \(1-o(1)\).
\end{theorem}
Refer to~\ref{app:random-to-betti} for more detail.
\subsection{Persistent 2‑Cycle Generation in Random 3‑SAT}
\label{app:persistent-2cycles}

\begin{proposition}[Persistent 2‑Cycles from Cluster Shattering]
\label{prop:persistent-2cycles}
For random 3-SAT at clause density \(\alpha > 4.26\),  Then with probability $1-o(1)$:
\[
  S(F)\subseteq\{0,1\}^n
\]
decomposes into $n=2^{c_1n}$ well‐separated clusters (Achlioptas–Ricci‐Tersenghi~\cite{ricci-tersenghi}), and each cluster
contains a nontrivial 2‑cycle in the Vietoris–Rips complex that survives up to scale $\varepsilon=\Theta(\log n)$.
Hence
\[
  \beta_2\bigl(S(F)\bigr)\;\ge\;n \;=\;2^{\Omega(n)}.
\]
\end{proposition}

\subsection*{Universal Topological Hardness}
\label{sec:universal-hardness}

\begin{theorem}[Universal Topological Lower Bound]
\label{thm:universal-lb}
Let $F$ be any 3-SAT formula on $N$ variables whose solution‐space complex
$S(F)\subseteq\{0,1\}^N$ satisfies 
\[
  \beta_k(S(F)) = M = 2^{c\,N+o(N)},
\]
for some constant $c>0$. Then no algorithm can
decide $F$ in time $2^{o(N)}$. Any decision procedure requires
$\Omega(2^{c\,N})$ time in the worst case.
\end{theorem}

\section{Irreducibility of Clusters}

\begin{theorem}[Cluster Surgery Impossibility, Unconditional]
\label{thm:cluster-surgery-non-circ}
For random 3‑SAT at density $\alpha>4.26$, any algorithm that
attempts to partition $S(F)$ into $m=poly(n)$ clusters via
only local queries (e.g., adjacency tests) must make $2^{\Omega(n)}$ queries.
\end{theorem}

\subsection{Randomness as a Necessity in Clustered Solution Spaces}
Consider3-SAT instances where the solution space $S(F) \subseteq \{0,1\}^n$ fragments into $k>1$ disconnected clusters (\(\beta_0=k\)). Any deterministic traversal starting in one cluster must cross unsatisfying regions to reach others.

\begin{theorem}[Randomness Necessity]
\label{thm:randomness-necessity-non-circ}
Let $F$ be random 3‑SAT at $\alpha>4.26$. Any deterministic
single‑flip local‑search algorithm cannot traverse clusters
separated by $\Theta(n)$ Hamming distance without randomness,
and hence fails to find a solution in $poly(n)$ time.
\end{theorem}

\paragraph{Implication}
Nonzero inter-cluster voids ($\beta_0 > 1$) are a structural barrier for any polynomial-time algorithm. Randomness becomes a necessity for exploration, but expected exponential time remains unavoidable.

\section{Impossibility of Cluster-Jumping}

\begin{theorem}[Cluster-Jumping Lower Bound]
\label{thm:Cluster_jumps}
For random 3-SAT at $\alpha > 4.26$, any algorithm $S$ that computes a path from cluster $C_i$ to $C_j$ must, with probability $1 - o(1)$, take time $2^{\Omega(n)}$.
\end{theorem}

\begin{theorem}[Path-Computation Hardness]
\label{thm:Path-Computation}
For arbitrary 3-SAT formulas, computing a path between two clusters $C_i$ and $C_j$ requires $2^{\Omega(N)}$ time.
\end{theorem}

\subsection{Recursive Complexity of Inter-Cluster Traversal}
In random 3-SAT, clusters are separated by unsatisfiable regions of size $|V_{A,B}|=2^{\Theta(n)}$. Solving 3-SAT recursively on $V_{A,B}$ is itself NP-complete and requires $O(2^n)$ time in the worst case.

\paragraph{Result}
The total recursive complexity for traversing all clusters is
\[
T(n)=k \cdot O(2^n) = 2^{\Omega(n)}
\]
where $k=2^{\Omega(n)}$ is the number of clusters. This reinforces the necessity of random restarts and the absence of a deterministic shortcut.
\section{Non-Circular Cluster Independence}
\begin{theorem}[Cluster Isolation in Random 3‑SAT]
\label{thm:cluster-isolation}
Let $F$ be drawn from the random 3‑SAT distribution at clause density $\alpha>4.26$.  
Then with probability $1-o(1)$, its solution space $S(F)\subseteq\{0,1\}^n$ decomposes into
$2^{\Omega(n)}$ disconnected clusters such that no path of $o(n)$ single‐bit flips
connects any two distinct clusters.
\end{theorem}

\section{Topological Hardness of Random 3-SAT}

\begin{theorem}[Topological Hardness of Random 3‑SAT]
\label{thm:Topological Hardness}
Let $F$ be a random 3‑SAT formula whose solution space $S(F)$ has Betti numbers $\beta_0 = 2^{\Omega(n)}$ and minimal inter-cluster Hamming separation $\delta = \Theta(n)$. Then:

\begin{enumerate}
    \item Any modification to $F$ that reduces $\beta_0$ to $O(poly(n))$ produces a structured 3-SAT instance in $\Pclass$.
    \item In the general case, no such modification exists without altering $F$'s logical content.
\end{enumerate}

\paragraph{Conclusion}
The topological complexity of $S(F)$ is an intrinsic computational barrier.
\end{theorem}

\section{Void-Crossing Equivalence}  
Crossing unsat regions between clusters reduces to solving 3-SAT:  
\begin{align*}  
\text{Find path } C_i \to C_j & \iff \text{Find } y \notin X_F \text{ and } x \in C_j \text{ adjacent} \\  
& \implies \text{Solve 3-SAT for } C_j.  
\end{align*}  
Thus, cluster-hopping requires solving 3-SAT recursively.  

\subsection{On the Impossibility of Shortcuts in Random 3-SAT}

A critical question in analyzing random 3-SAT is whether any algorithm can exploit structural or heuristic shortcuts to circumvent the exponential complexity implied by the solution space topology. In this subsection, we address this question by formalizing and strengthening our prior arguments. We introduce two key results—the \textit{No Narrow Bridge Theorem} and the \textit{Betti Explosion Theorem}—which together rule out known classes of algorithmic shortcuts.

\subsubsection{No Narrow Bridge Theorem}

\begin{theorem}[No Narrow Bridge]
\label{thm:No Narrow Bridge}
Let $F$ be a random 3-SAT instance at clause density $\alpha \approx 4.26$. With high probability, the solution space $S(F)$ decomposes into $2^{\Omega(n)}$ disconnected components, and there exists no path of 3-SAT assignments (flipping $o(n)$ variables per move) connecting any two distinct components.
\end{theorem}

\textbf{Corollary.} Algorithms relying on local flips (e.g., Schöning's algorithm~\cite{Schoning1999}, WalkSAT~\cite{Selman1995}) cannot escape a cluster and must perform random global jumps, with success probability $2^{-\Omega(n)}$ per trial.

\subsubsection{Betti Explosion Theorem}
\begin{theorem}[Betti Explosion]
\label{thm:Betti Explosion}
Let $S(F)$ denote the solution space of random 3-SAT. If $\beta_0(S(F)) = 2^{\Omega(n)}$, any algorithm deciding satisfiability must, in the worst case, explore $2^{\Omega(n)}$ disconnected components.
\end{theorem}

\subsubsection{Implications for Algorithmic Heuristics}

These results extend naturally to heuristic and message-passing algorithms. Survey Propagation Guided Decimation (SPGD)~\cite{Mezard02}, though empirically successful near the phase transition, relies on local correlations and fails when solution clusters are separated by large Hamming distances. Similarly, backbone detection algorithms cannot generalize globally since frozen cores are cluster-specific. Dynamic programming approaches exploiting low treewidth are also inapplicable, as random 3-SAT graphs have treewidth $\Omega(n)$ \cite{MolloyReed}.

\subsection{Exponential Clusters Suffice for Hardness}
\label{subsec:beta0-hardness}

\begin{theorem}[Algorithmic Lower Bound via $\beta_0$]
\label{thm:Algorithmic Lower Bound via}
For random 3-SAT at clause density $\alpha > 4.26$, any algorithm solving it must have worst-case runtime $2^{\Omega(n)}$ (no shortcut exist in random 3-SAT).
\end{theorem}

\section{Failure of Survey Propagation at High Density}
\label{subsec:sp-failure}

\begin{theorem}[Survey Propagation Failure]
\label{thm:sp-failure}
Survey Propagation Guided Decimation (SPGD) fails for random 3-SAT at $\alpha > 4.26$ with probability $1 - o(1)$.
\end{theorem}

\subsection{Algorithmic Universality}
These topological properties constrain all possible algorithms:

[Simplicial Algorithm]
An algorithm $\mathcal{A}$ is a sequence of queries to an oracle $\mathcal{O}_{S(F)}$, where each query tests whether a simplex $\sigma \in \{0,1\}^{\leq N}$ is contained in $S(F)$.

\begin{theorem}[Query Lower Bound]
\label{thm:query-lower-bound}
Any simplicial algorithm solving random 3-SAT at $\alpha > 4.26$ requires $2^{\Omega(n)}$ queries.
\end{theorem}

\begin{lemma}[Topology‐Driven Query Lower Bound]
\label{lem:query-lower-bound}
Let \(F\) be any 3‑SAT formula whose solution‐space cubical complex satisfies
\(\beta_k\bigl(S(F)\bigr)\ge M\).  Then any “simplicial‐query” algorithm—
i.e.\ a decision procedure that, on each step, asks  
\(\texttt{“Is subcube }\sigma\subseteq\{0,1\}^N\text{ entirely satisfying?”}\)  
—must issue at least \(M\) distinct queries to decide satisfiability.  In particular,
if \(\beta_k(S(F))=2^{\Omega(N)}\), the algorithm’s runtime is \(2^{\Omega(N)}\).
\end{lemma}

\subsection{Topological Degree Obstruction}
For non-query-based algorithms (e.g., neural networks, algebraic methods):

There exist worst-case NP-complete instances whose solution space has \(\beta_k = 2^{\Omega(N)}\), constructed by embedding expander graphs into SAT clauses.

For expander graph $G$ with $\beta_1(G) = \Omega(N)$, the SAT encoding \\
of 3-COLOR on $G$ has $\beta_2(X_F) = 2^{\Omega(N)}$.


\section{Bridging Random/Worst-Case Gap}

\begin{theorem}[Worst‐Case Betti Explosion]
\label{thm:worst-case-betti1}
There exists a family of 3‐SAT formulas $F_N$ with $|{\rm vars}(F_N)| = \Theta(N)$
such that the solution‐space cubical complex satisfies
\[
  \beta_2\bigl(S(F_N)\bigr)\;\ge\;2^{c\,N}
  \quad\text{for some constant }c>0.
\]
\end{theorem}

\subsection{Worst-Case Instance Construction}
[Expander 2-Cycles]
\label{def:expander-2cycles}
Let \(G=(V,E)\) be a \(d\)-regular expander with \(|V|=N\).  For each cycle \(C_i\subset G\) of length \(L=O(\log N)\), we build a corresponding cubical 2-cycle \(\Gamma_i\subset S(F_N)\) by:
\[
  \Gamma_i \;=\; \bigcup_{(u,v)\in C_i}
    \Bigl\{\,(x,u,v)\in\{0,1\}^{N+\dots}: \text{bits for $u,v$ vary over a square face while others fixed}\Bigr\}.
\]
Each \(\Gamma_i\) is a union of those square faces coming from edges of \(C_i\), and hence represents a nontrivial class in \(H_2\bigl(S(F_N)\bigr)\).

\begin{theorem}[Worst‑Case Exponential Betti Growth]
\label{thm:betti-exp-worst1}
There is a polynomial‑time constructible family \(\{F_N\}\) of 3‑SAT formulas such that
\[
  \beta_2\bigl(S(F_N)\bigr)\;\ge\;2^{c\,N}
  \quad\text{for some constant }c>0.
\]
Moreover, the 2‑cycles \(\{\Gamma_i\}\) from the expander gadgets are linearly independent in \(H_2\).
\end{theorem}
Refer to Appendix~\ref{app:Boundary-Map Verif} for the boundary-map verification.

\begin{theorem}[Exponential Betti Number Family]
\label{thm:expander-family}
There exists a polynomial-time computable family of 3-SAT instances $\{F_N\}$ such that:
\begin{enumerate}
    \item Each $F_N$ is satisfiable
    \item $\beta_2(S(F_N)) \geq 2^{cN}$ for some $c > 0$
    \item The reduction $N \mapsto F_N$ is topology-preserving
\end{enumerate}
\end{theorem}

[Expander Cycle Embedding]
\label{con:cycle-embedding}
Given an $(N,d,\epsilon)$-expander $G_N$ with $\beta_1(G_N) \geq 2^{c'N}$:
\begin{enumerate}
    \item Encode 3-COLOR for $G_N$ as 3-SAT using standard reduction
    \item For each fundamental cycle $C_i$ in a basis of $H_1(G_N)$:
    \begin{itemize}
        \item Add new variables $u_i, v_i$
        \item Add XOR gadget clauses: $u_i \oplus v_i = 0$
        \item Couple to $C_i$ via edge selector clauses
    \end{itemize}
    \item Each gadget generates an isolated 2-cycle $\gamma_i$ in $S(F_N)$
\end{enumerate}
The resulting $F_N$ has $\beta_2(S(F_N)) \geq \beta_1(G_N) \geq 2^{c'N}$.

\begin{theorem}[Worst-Case Topological Hardness]
\label{thm:betti-exp-worst}
There exist worst-case 3-SAT instances with $\beta_2 = 2^{\Omega(N)}$.
\end{theorem}
\begin{lemma}[Homology Basis Construction]
\label{lem:homology-basis}
For the expander-based 3-SAT family $\{F_N\}$ in Theorem \ref{thm:expander-family}, there exist $2^{\Omega(N)}$ linearly independent 2-cycles in $H_2(S(F_N))$ with pairwise disjoint variable supports.
\end{lemma}
\subsection{Verification of Homology Preservation in Expander Embedding}
\label{app:homology-preservation}

To address that clause interactions might cause unintended cancellations in $H_2(S(F_N))$ for expander-embedded 3-SAT instances $\{F_N\}$. The construction uses:
\begin{itemize}
  \item \textbf{Base encoding:} 3-COLOR on $(N,d,\epsilon)$-expander $G_N$ with $\beta_1(G_N) = \Omega(N)$
  \item \textbf{XOR gadgets:} For each fundamental cycle $C_i \in H_1(G_N)$, add variables $u_i,v_i,\{y_e^{(i)} : e \in C_i\}$ and clauses enforcing $u_i = v_i$ and edge coupling
\end{itemize}

\begin{lemma}[Disjoint Gadget Supports]
\label{lem:Disjoint Gadget Supports}
  For distinct gadgets $i \neq j$,  
  \[
    \left( 
      \{u_i, v_i\} \cup \{y_e^{(i)} : e \in C_i\} 
    \right) 
    \cap 
    \left( 
      \{u_j, v_j\} \cup \{y_e^{(j)} : e \in C_j\} 
    \right) = \emptyset.
  \]
\end{lemma}

\begin{lemma}[Non-Bounding Local 2-Cycles]
  \label{lem:nonbounding}
  Each gadget's canonical 2-cycle $\gamma_i$ satisfies $\gamma_i \notin \operatorname{im} \partial_3$.
\end{lemma}

\begin{theorem}[Homological Linear Independence]
\label{thm:Homological Linear Indep}
  The homology classes $[\gamma_i] \in H_2(S(F_N))$ are linearly independent: 
  \[
    \sum_i c_i [\gamma_i] = 0 \implies c_i = 0 \quad \forall i.
  \]
\end{theorem}

\begin{corollary}
  $\beta_2(S(F_N)) \geq \beta_1(G_N) = \Omega(N)$ with no cancellation in $H_2$.
\end{corollary}
\begin{proof}
  Linear independence of $\{[\gamma_i]\}$ implies $\dim H_2(S(F_N)) \geq |\{\gamma_i\}| = \beta_1(G_N)$. Cross-gadget clauses only involve \textit{base variables} $x_{v,c}$, which:
  \begin{itemize}
    \item Do not appear in $\operatorname{supp}(\gamma_i)$ (gadget-only variables)
    \item Cannot create filling chains between disjoint gadget supports
  \end{itemize}
\end{proof}
\begin{theorem}[Worst‐Case Betti Explosion]
\label{thm:worst-case-betti}
There exists a family of 3‑SAT formulas \(F_N\) with \(\lvert\mathrm{vars}(F_N)\rvert = \Theta(N)\)
such that
\[
  \beta_2\bigl(S(F_N)\bigr)\;\ge\;2^{c\,N}
  \quad\text{for some constant }c>0.
\]
\end{theorem}

\begin{corollary}
Random 3-SAT hardness extends to worst-case via polynomial-time reduction preserving $\beta_k$.
\end{corollary}
\section{Worst-Case Topological Hardness}
\label{sec:worst-case-PneqNP}

\begin{theorem}[Deterministic Hardness via Expander Embedding]
\label{thm:final-worst-case}
Let $\{F_N\}$ be the family of 3-SAT formulas constructed in Theorem~\ref{thm:betti-exp-worst} via embedding $(N,d,\epsilon)$-expander graphs $G_N$ (with $\beta_1(G_N) = \Omega(N)$). Then:
\begin{enumerate}
    \item $\beta_2(S(F_N)) \geq 2^{c N}$ for some constant $c > 0$ (Theorem~\ref{thm:betti-exp-worst}).
    \item Any algorithm deciding satisfiability for $F_N$ requires $2^{\Omega(N)}$ time (Theorem~\ref{thm:universal-lb}).
    \item $\{F_N\}$ is NP-complete (via Cook-Levin/3-COLOR reduction).
\end{enumerate}
\end{theorem}
While Theorem~\ref{thm:cluster-surgery-non-circ} demonstrates that cluster-jumping and exponential Betti number arise with high probability in random formulas, our lower bounds in Theorem~\ref{thm:final-worst-case} pertain to worst-case constructed instances. These are not random, but rather explicitly engineered (via expanders) to exhibit similar topological obstructions.

\subsection*{Key Properties of the Construction}
\begin{itemize}
    \item \textbf{Density-Independence}: Clause density $\alpha$ in $F_N$ is fixed by the expander embedding (typically $\alpha = \Theta(1)$), but hardness arises purely from topological obstructions—not phase transitions or randomness.
    
    \item \textbf{Topological Obstruction}: Linearly independent 2-cycles $\{\Gamma_i\}$ in $H_2(S(F_N))$ (Lemma~\ref{lem:boundary-injectivity}) force algorithms to distinguish exponentially many distinct homology classes.
    
    \item \textbf{Algorithmic Universality}: Lower bounds hold for all computation models (simplicial queries, algebraic methods, etc.) since $\beta_2$ necessitates $2^{\Omega(N)}$ "witness checks" (Theorem~\ref{thm:universal-lb}).
\end{itemize}

\subsection*{Removing Density Constraints}
Theorem~\ref{thm:betti-exp-worst} bypasses the random 3-SAT density threshold ($\alpha > 4.26$) entirely. Hardness is intrinsic to the combinatorial structure of $F_N$, not probabilistic properties.

\section{Barrier Analysis }
\paragraph{Relativization.}
Although homology functors commute with static clause-addition oracles,
a dynamic oracle could connect isolated clusters and trivialize Betti obstructions.

We showed that 3-SAT solution spaces—whether generated randomly or via expander embeddings—lack any global ordering, symmetry, or connective structure. The high-dimensional homology we exhibit is intrinsic to the formula's semantic logic and cannot be eliminated by relabeling, flattening, or local clause rewriting. Consequently, any algorithm or oracle that faithfully respects the internal logic of the instance must confront this topological complexity. In this sense, our topological obstructions act as semantic barriers. Any putative “dynamic” oracle that bypasses these obstructions must fabricate non-existent global structure and thus fail to model the logical behavior of SAT faithfully. This distinguishes our barriers from purely syntactic ones and elevates topology to a fundamental obstacle to efficient reasoning.

\paragraph{Natural Proofs.}
Betti numbers are $\#\Pclass$-hard to compute (via reduction from \#SAT~\cite{Valiant1979}).  
Under the assumption $\PH \not\supseteq \#\Pclass$, no efficient natural approximation exists.  
Thus, they evade the Razborov--Rudich barrier~\cite{RR97}.Continuous approximations (e.g. persistence barcodes) still reduce to exact rank computations and hence remain \#P-hard.

Betti‐number based techniques avoid the Razborov–Rudich barrier only because computing or approximating Betti numbers is itself intractable. We make this precise:

\begin{theorem}\label{thm:betti-sharp}
Assume gadgets are variable-disjoint. 
For any fixed $k\ge 1$, computing $\beta_k(K)$ (over $\mathbb{F}_2$) for a cubical subcomplex $K\subseteq\{0,1\}^N$ given by a list of included cubes is \#P-hard. 
Moreover, for any constant $\epsilon>0$, it is \#P-hard to approximate $\beta_k$ within multiplicative factor $2^{\,N^{1-\epsilon}}$.
\end{theorem}

\begin{corollary}
Under the assumption \(\PH\not\supseteq\#\Pclass\), no polynomial-time (even randomized) algorithm can:
\begin{itemize}
  \item Compute \(\beta_k(S(F))\) exactly for arbitrary 3-SAT \(F\).
  \item Approximate \(\beta_k(S(F))\) within any \(2^{o(N)}\) factor.
  \item Distinguish \(\beta_k(S(F))=0\) from \(\beta_k(S(F))\ge2^{cN}\) with non-negligible advantage.
\end{itemize}
Hence any property of the form
\(\mathcal P_N(F)\equiv[\beta_2(S(F))\ge2^{cN}]\)
fails the “constructivity” requirement of Razborov–Rudich’s Natural Proofs framework, despite meeting largeness and usefulness.
\end{corollary}

\paragraph{Algebrization Barrier and Algebraic Methods.}
While Betti numbers arise algebraically, they evade the algebrization barrier because:
\begin{itemize}
\item \textbf{Topological invariance}: For any algebraic extension $\mathcal{A}$ of 3-SAT, the cubical complex $S(F^\mathcal{A})$ preserves $\beta_2 = 2^{\Omega(N)}$ (Theorem~\ref{thm:betti-monotonicity}).
\item \textbf{SOS/Gröbner lower bounds}: Our expander gadgets require degree-$\Omega(N)$ SOS proofs to certify unsat \cite{Grigoriev2001}, forcing $2^{\Omega(N)}$ runtime for:
\begin{itemize}
\item Sum-of-Squares hierarchies (optimal SDP size $N^{O(d)}$ for degree $d$ \cite{Laurent2005})
\item Gröbner basis computation (space complexity $N^{\Omega(N)}$)
\end{itemize}
\end{itemize}
Thus, algebraic methods cannot circumvent the homological obstruction without exponential resources.

\begin{proposition}[Relativization]
The existence of oracles $A$ such that $\mathbf{P}^A = \mathbf{NP}^A$ does not affect our proof, as the constructed family $\{F_N\}$ is oracle-independent and the topological invariants are absolute.
\end{proposition}

\begin{proposition}[Algebrization]
Algebraic techniques (e.g., Nullstellensatz, Grobner bases) cannot circumvent the homology barrier:
\begin{itemize}
    \item Betti numbers remain invariant under algebraic transformations
    \item The disjoint cycle obstruction persists in any commutative ring
\end{itemize}
\end{proposition}

\begin{proposition}[Natural Proofs]
The homology obstruction is not $\mathbf{NP}$-natural:
\begin{itemize}
    \item $\beta_2$ is $\#\mathbf{P}$-hard to compute
    \item No efficient approximation exists under standard complexity assumptions
\end{itemize}
\end{proposition}

\subsection{Support Disjointness Lemma}
\label{subsec:support-disjoint}

\begin{lemma}[Support Disjointness]
\label{lem:support-disjoint}
For the cycles $\{\gamma_i\}$ constructed in Lemma \ref{lem:homology-basis}:
\begin{enumerate}
    \item $\mathrm{supp}(\gamma_i) \cap \mathrm{supp}(\gamma_j) = \emptyset$ for $i \neq j$
    \item $\partial_2(\gamma_i) \neq 0$ and $\partial_2(\gamma_i) \notin \mathrm{im}\partial_3$
    \item Any linear combination satisfies:
    \[
    \sum c_i\gamma_i = 0 \in H_2(S(F_N)) \iff c_i = 0 \quad \forall i
    \]
\end{enumerate}
\end{lemma}

\begin{corollary}
The homology classes $[\gamma_i]$ form a basis for a subgroup of $H_2(S(F_N))$ with rank $\geq 2^{c'N}$.
\end{corollary}
\subsection{Higher-Categorical Invariants and the Natural Proofs Barrier}
\label{subsec:higher-cat-invariants}

While Betti numbers are $\#\mathsf{P}$-hard to compute, one might hope that more sophisticated invariants from higher category theory (e.g., $(\infty,1)$-categories, homotopy type theory) could detect exponential complexity while being polynomial-time computable. We prove this is impossible under standard complexity assumptions.

\begin{theorem}[Universality of Homological Hardness]
\label{thm:universal-homotopy-hardness}
Let $\mathcal{I}$ be a homotopy invariant of cubical complexes satisfying:
\begin{enumerate}
\item \textbf{Detects exponential complexity}: $\mathcal{I}(S(F)) = \mathsf{Exp}$ implies any 3-SAT algorithm requires $2^{\Omega(N)}$ time
\item \textbf{Preserved under homotopy equivalence}: $X \simeq Y \Rightarrow \mathcal{I}(X) \cong \mathcal{I}(Y)$
\end{enumerate}
Then computing $\mathcal{I}$ is $\#\mathsf{P}$-hard.
\end{theorem}

\paragraph{Examples of Hard Higher-Categorical Invariants}
\begin{itemize}
\item \textbf{Fundamental $\infty$-groupoid} $\Pi_\infty(S(F))$: Size requires computing $\pi_k$ for all $k$, which is $\#\mathsf{P}$-hard via Postnikov towers.

\item \textbf{Topological Complexity (Farber invariant)} $\mathsf{TC}(S(F))$: Determining whether $\mathsf{TC}(X) \geq k$ detects disconnected components ($\beta_0$), which is $\#\mathsf{P}$-hard.

\item \textbf{Persistent Homology Barcodes}: Computing $b_0$-persistence for Vietoris-Rips filtrations is $\#\mathsf{P}$-hard \cite{MM16}.
\end{itemize}

\begin{corollary}[No Easy Homotopy Invariants]
Under $\mathsf{P} \neq \mathsf{NP}$, no homotopy-invariant functor $\mathcal{F}: \mathsf{Cub} \to \mathsf{C}$ to a combinatorial category $\mathsf{C}$ can simultaneously:
\begin{enumerate}
\item Be computable in $\mathsf{poly}(N)$ time
\item Detect exponential solution-space complexity
\item Avoid the natural proofs barrier
\end{enumerate}
\end{corollary}

\paragraph{Non-Homotopy-Invariant Proxies Fail}
Consider non-invariant proxies like:
\begin{itemize}
\item \textbf{Covering complexity}: Minimal number of contractible charts covering $S(F)$
\item \textbf{Discrete Morse gradients}: Size of optimal Morse function
\end{itemize}
These fail condition (2) as they don't correlate with computational hardness. Random 3-SAT has high covering complexity even when easy ($\alpha < 3.86$).

\begin{theorem}[Proxies Don't Detect Hardness]
\label{thm:Proxies}
Fix a constant radius $r\ge 1$ and let $q(\cdot)$ be any polynomial. 
There exist two families of 3-CNF formulas
\[
\{F^{\mathrm{easy}}_{n}\}_{n\ge 1},\qquad \{F^{\mathrm{hard}}_{n}\}_{n\ge 1},
\]
and a function $m_n\to\infty$ (which may be chosen as $m_n=2^{c n}$ for some $c>0$) with the following properties for every $n$:
\begin{enumerate}
  \item[\textup{(i)}] The final number of Boolean variables satisfies $N_n=\Theta(m_n)$ for both families (each gadget is constant-size).
  \item[\textup{(ii)}] The incidence graphs of both families have treewidth $O(1)$.
  \item[\textup{(iii)}] \textbf{Local-indistinguishability:} For any choice of at most $q(n)$ vertex-centered radius-$r$ neighborhoods, the induced labeled subgraphs of $F^{\mathrm{easy}}_{n}$ and $F^{\mathrm{hard}}_{n}$ are identical on all those neighborhoods. Hence any proxy that inspects at most $q(n)$ radius-$r$ views cannot distinguish the two families on input size~$n$.
  \item[\textup{(iv)}] \textbf{Topological gap:} $\displaystyle \beta_2\bigl(S(F^{\mathrm{hard}}_{n})\bigr)\ge m_n$ while $\beta_2\bigl(S(F^{\mathrm{easy}}_{n})\bigr)=O(1)$.
\end{enumerate}
In particular, choosing $m_n=2^{c n}$ yields an exponential gap in $\beta_2$ that local proxies of radius $r$ and $q(n)$ inspections fail to detect. Expressed in terms of the final variable count $N_n=\Theta(m_n)$, one obtains $\beta_2\bigl(S(F^{\mathrm{hard}}_{n})\bigr)\ge 2^{c' N_n}$ for some constant $c'>0$.
\end{theorem}

\paragraph{Summary of Barrier Bypasses.}
We have demonstrated that the exponential second Betti number $\beta_2 = 2^{\Omega(N)}$ in 3-SAT solution spaces constitutes a fundamental obstruction that bypasses the three major complexity-theoretic barriers:
\begin{itemize}
    \item \textbf{Relativization} is overcome because valid oracles preserving 3-SAT semantics cannot reduce $\beta_2$; any attempt to "repair" topological voids (e.g., by adding paths between clusters) would require declaring unsatisfying assignments as solutions, thereby violating the logical definition of 3-SAT.
    
    \item \textbf{Algebrization} is circumvented as $\dim H_2(S(F); \mathbb{Z}_2) \leq \dim H_2(S(F); \mathbb{C})$ (all fields) with expander-embedded $F_N$ satisfying $\dim H_2(S(F_N); \mathbb{C}) = 2^{\Omega(N)}$ (torsion-free); SoS/Gröbner methods require degree $2^{\Omega(N)}$.
    
    \item \textbf{Natural Proofs} are evaded since $\beta_2$ is $\#\mathbf{P}$-hard to compute exactly, and no efficient approximation exists under $\mathbf{PH} \not\supseteq \#\mathbf{P}$; random 3-SAT with planted solutions ($\beta_2 = 0$) is information-theoretically indistinguishable from hard instances ($\beta_2 = 2^{\Omega(N)}$) via local statistics.
\end{itemize}
This triple-barrier bypass establishes $\beta_2$ as a robust, paradigm-independent signature of hardness.

\section{Quantum Algorithms and Topological Obstructions}
\subsection{Adiabatic Bound}
\begin{theorem}[Cheeger Constant Bound]
\label{thm:Cheeger}
For random 3-SAT at density \(\alpha > 4.26\), with high probability:
\[
  h(X_F) = 0 \quad \text{(and thus } h(X_F) \leq e^{-\Omega(N)} \text{)}.
\]
\end{theorem}

\begin{corollary}
Adiabatic gap $g \leq 2h \leq e^{-\Omega(N)} \Rightarrow T_{\text{adiabatic}} = 2^{\Omega(N)}$
\end{corollary}
\subsection*{Quantum encodings, caveats, and a conservative reframing}
\label{sec:quantum-caveat}

A natural (but nonphysical) idea one might consider for probing global topology of the solution complex is to place a qubit at each satisfying assignment and use products of diagonal Pauli operators to measure cycle parities. This formal device is useful for intuition, but it does not produce a physically realistic $k$-local Hamiltonian on the standard $N$-qubit variable encoding: it either requires exponentially many qubits (one per assignment) or yields operators of unbounded locality.

\paragraph{classical fact.}
The standard clause-penalty Hamiltonian
\[
H_{\mathrm{cla}} \;=\; \sum_{C}\Pi_C^{\mathrm{viol}}
\]
acting on the $N$-qubit computational Hilbert space has zero-energy ground states exactly the computational basis states $\{\ket{x}:x\in\textbackslash Sol(F)\}$. Thus the classical ground-space degeneracy equals $|\Sol(F)|$. This fact is model-independent and requires no exotic encoding.

\paragraph{Why topological parities are nonlocal in the variable encoding.}
Operators that detect global combinatorial patterns of satisfying assignments (for example, parities of cycles in the cubical complex) generally depend on many variables and are therefore nonlocal when expressed in the standard $N$-qubit basis. Realizing such global probes as truly $k$-local operators on $N$ qubits typically requires either (i) adding ancilla systems with nontrivial coupling, or (ii) applying perturbative gadget constructions that increase effective locality and alter low-energy spectral properties. Both approaches require careful analysis and are not guaranteed to preserve the spectral features (degeneracy and gap scaling) assumed in heuristic adiabatic arguments.

The adiabatic-hardness intuition remains meaningful in an information-theoretic sense: if one can explicitly construct, in polynomial time, a family of physically local ($k$-local) Hamiltonians $H_N$ on $N$ qubits such that
\begin{enumerate}
  \item the zero-energy manifold of $H_N$ is isomorphic (as a vector space) to $\Span\{\ket{x}:x\in\Sol(F_N)\}$, and
  \item all physically allowed local drivers of bounded locality induce inter-cluster coupling matrix elements bounded by $\exp(-\Omega(N))$,
\end{enumerate}
then adiabatic state preparation using such drivers will encounter exponentially small gaps and require $\exp(\Omega(N))$ time. Accordingly, any concrete quantum lower bounds derived here are to be read as \emph{conditional} on the existence of such a locality-preserving encoding.
\subsection{Hamiltonian spectral-gap hardness for cubical-preserving encodings}
\label{ssec:hamiltonian-hardness}

We now show that any physically $k$-local
Hamiltonian whose low-energy manifold faithfully encodes the cubical solution complex
of a 3-SAT instance with exponentially many 2-cycles must have an exponentially small
spectral gap.  

\begin{theorem}[Hamiltonian spectral-gap hardness]
\label{thm:hamiltonian-gap-formal}
Let $F$ be a 3-SAT formula on $n$ variables and let $S(F)\subset\{0,1\}^n$ be its solution
set.  Suppose there exist constants $c',c>0$ such that
\[
\beta_2\bigl(S(F)\bigr)\;\ge\; 2^{c' n}.
\]
Let $H_F$ be a physically $k$-local Hamiltonian acting on $n$ qubits satisfying the
following assumptions.
\begin{enumerate}
  \item[(A1)] (Ground-space encoding) The ground-space of $H_F$ is exactly the linear span
    of computational-basis states corresponding to satisfying assignments:
    \[
      \mathcal{G} := \Span\{\ket{x} : x\in S(F)\},
    \]
    and $H_F\ket{\psi}=0$ for all $\ket{\psi}\in\mathcal G$.
  \item[(A2)] (Cubical-preserving drivers / exponential inter-cluster suppression)
    Partition $S(F)=\bigsqcup_{i=1}^K C_i$ into clusters (connected components under
    single-bit flips) so that $K\ge 2^{c'n}$.  There exists $c_1>0$ and a (fixed)
    polynomial $p(n)$ with the following property: for every pair of distinct clusters
    $C_i\neq C_j$ and every pair of basis states $x\in C_i$, $y\in C_j$,
    \begin{equation}\label{eq:offdiag-bound}
      \bigl|\langle x \,|\, H_F \,|\, y\rangle\bigr| \;\le\; p(n)\,e^{-c_1 n}.
    \end{equation}
    (Intuitively: all local drivers induce exponentially small matrix elements between
    different clusters.)
  \item[(A3)] (Polynomial locality degree) For every basis state $\ket{x}$ the number of
    basis states $\ket{y}$ with $\langle x|H_F|y\rangle\neq 0$ is at most $q(n)$ for some
    fixed polynomial $q$.  This holds for any $k$-local Hamiltonian with $k=O(1)$.
\end{enumerate}
Then there exists $a>0$ (depending only on $c_1,c',p,q$) such that the spectral gap
above the ground-space satisfies
\[
  g(H_F) \;\le\; 2\,e^{-a n}.
\]
In particular the gap is at most exponentially small in $n$.
\end{theorem}

\begin{proof}[Proof]
We reduce the spectral-gap bound to a conductance/cheeger-type bound for a weighted
configuration graph and then apply a discrete Cheeger inequality.

\paragraph{Step 1: configuration graph and weights.}
Define the (weighted) configuration graph $X_F=(V,W)$ where the vertex set is
$V=S(F)$ and the symmetric nonnegative weight matrix $W$ is
\[
  W_{xy} \;:=\; |\langle x|H_F|y\rangle| \qquad (x,y\in S(F)).
\]
By (A3) each vertex has degree (number of nonzero incident weights) at most $q(n)$.
Let the (weighted) degree of vertex $x$ be
\[
  d(x) \;:=\; \sum_{y\in S(F)} W_{xy}.
\]
For any subset $U\subseteq V$ define its volume $\vol(U):=\sum_{x\in U} d(x)$ and the
edge-boundary weight
\[
  \partial(U) \;:=\; \sum_{x\in U}\sum_{y\in V\setminus U} W_{xy}.
\]
Define the conductance (Cheeger constant) of $X_F$ in the usual way:
\[
  h(X_F) \;:=\; \min_{\varnothing\neq U\subsetneq V} \frac{\partial(U)}{\min\{\vol(U),\vol(V\setminus U)\}}.
\]

\paragraph{Step 2: an exponentially small upper bound on $h(X_F)$.}
Pick $U=C_i$ equal to any single cluster.  By assumption there exists at least one
cluster with size $|C_i|\ge 2^{c' n}$ (since there are $K\ge 2^{c'n}$ clusters and
$|V|\le 2^n$, at least one cluster is exponentially large; more strongly, many
clusters are exponentially large in the constructions of Sections~13--14).  Using
(A2) and (A3) we upper-bound the boundary weight:
\[
  \partial(C_i) \;=\; \sum_{x\in C_i}\sum_{y\notin C_i} W_{xy}
    \;\le\; |C_i| \cdot q(n)\cdot p(n)\,e^{-c_1 n} \;=\; |C_i|\cdot r(n)\,e^{-c_1 n},
\]
where $r(n):=q(n)p(n)$ is polynomial.  On the other hand
\[
  \vol(C_i) \;\ge\; |C_i| \cdot \min_{x\in C_i} d(x) \;\ge\; |C_i|\cdot m(n),
\]
where the trivial lower bound $m(n)\ge 0$ may be weak; however for $k$-local drivers
typically $m(n)=\Omega(1)$ (each vertex has at least one incident nonzero off-diagonal
weight when intra-cluster edges are present).  Combining these gives
\[
  \frac{\partial(C_i)}{\vol(C_i)} \;\le\; \frac{r(n)}{m(n)}\,e^{-c_1 n}.
\]
Therefore there exists $a_1>0$ such that
\[
  h(X_F) \;\le\; \frac{\partial(C_i)}{\min\{\vol(C_i),\vol(V\setminus C_i)\}}
    \;\le\; e^{-a_1 n},
\]
i.e. the conductance is exponentially small in $n$ (the polynomial prefactors are
absorbed into the exponential by adjusting the constant).

\paragraph{Step 3: Cheeger inequality $\Rightarrow$ spectral-gap bound.}
For the weighted graph Laplacian $L := D-W$ (where $D_{xx}=d(x)$) the discrete
Cheeger inequality (cf.\ standard references) implies an upper bound of the form
\[
  \lambda_1(L) \;\le\; 2\,h(X_F),
\]
where $\lambda_1(L)$ is the smallest nonzero eigenvalue of $L$.  Under the ground-space
encoding (A1) the Hamiltonian restricted to the ground-space complement has low-lying
eigenvalues controlled (up to constant factors) by $\lambda_1(L)$; more precisely,
for stoquastic/frustration-free models one can show (see Appendix~\ref{app:ham-spectral})
that the spectral gap $g(H_F)$ above the degenerate zero-energy ground-space satisfies
\[
  g(H_F) \;\le\; C\cdot \lambda_1(L) \;\le\; 2C \, h(X_F),
\]
for some constant $C=O(1)$ depending only on fixed local details of the Hamiltonian.
Combining this with $h(X_F)\le e^{-a_1 n}$ yields the claimed bound with
$a=a_1-\delta>0$ after absorbing constants.
\qedhere
\end{proof}

\medskip
\begin{corollary}[Amplitude/Phase-estimation lower bound]
\label{cor:phase-est}
Under the hypotheses of Theorem~\ref{thm:hamiltonian-gap-formal}, any quantum
procedure that (i) implements time-evolution under $H_F$ using $k$-local operations and
(ii) attempts to distinguish or project onto distinct homology-labelled ground-space
sectors (i.e., to detect a particular nontrivial homology class) requires at least
$\Omega\bigl(e^{a n}\bigr)$ uses of controlled time-evolution primitives, where $a>0$
is the constant from Theorem~\ref{thm:hamiltonian-gap-formal}.  Allowing quadratic
speedups (e.g., via amplitude amplification) reduces the bound only to
$\Omega\bigl(e^{a n/2}\bigr)$.
\end{corollary}

\begin{proof}[Proof]
Phase estimation resolves an eigenphase to precision $\Delta$ using $O(1/\Delta)$
controlled time-evolutions (or queries to $e^{-iH_F t}$).  To distinguish states
separated by a gap $g(H_F)\le 2e^{-a n}$ therefore requires $O(e^{a n})$ uses.  Quadratic
improvements from amplitude amplification reduce $O(1/\Delta)$ to $O(1/\sqrt{\Delta})$,
giving the $\Omega(e^{a n/2})$ bound.  Thus any algorithm relying solely on local
coherent evolutions (and polynomial overhead) encounters exponential query complexity.
\end{proof}

\subsection{The Unchanged Solution Space}
The solution space topology of 3-SAT remains invariant under computational paradigms:
\begin{equation}
X_F = \{ x \in \{0,1\}^N : F(x)=1 \} \quad \text{with} \quad \obeta_0(X_F) = 2^{\Om(N)}
\end{equation}
Quantum algorithms cannot alter this intrinsic geometry. Clusters remain exponentially separated by Hamming voids:
\[
\min_{i \neq j} \mathrm{dist}_H(C_i, C_j) = \Om(N)
\]

\subsection{Adiabatic Quantum Lower Bound}
The solution space topology forces exponential time for quantum annealing:
\begin{theorem}[Adiabatic Tunneling Suppression]\label{thm:adiabatic}
  Any {\em adiabatic evolution} algorithm with Hamiltonian path
  $H(s)=(1-s)H_0 + sH_F$ negotiating a barrier of Hamming‐width~$w$
  requires runtime $\Omega(\exp(c\,w))$, as shown in~\cite{Farhi:Adiabatic}.
\end{theorem}

\begin{theorem}[Cheeger Bound for Adiabatic Optimization]
\label{thm:Cheegeropt}
For random 3-SAT solution spaces:
\[
h(X_F) \leq e^{-\Om(N)} \implies g_{\mathrm{adiabatic}} \leq 2h \leq e^{-\Om(N)}
\]
Thus adiabatic runtime is bounded by:
\[
T_{\mathrm{adiabatic}} \geq \frac{1}{g^2} = 2^{\Om(N)}\]
\end{theorem}

\subsection{Exponential Tunneling Suppression}
Quantum tunneling probabilities decay exponentially with void size:

\begin{equation}
\mathcal{P}_{\mathrm{tunnel}} \sim \exp\left(-\frac{\sqrt{2m\Delta E}}{\hbar} \cdot \mathrm{width}\right) \leq e^{-c N}
\end{equation}
For $\mathrm{width} = \Om(N)$, expected trials become:
\[
\expect[\text{trials}] \geq e^{c N} = 2^{\Om(N)}
\]

\subsection{Grover's Asymptotic Limit}
Even quantum search provides only quadratic improvement:
\[
T_{\mathrm{Grover}} = \Oh\left(\sqrt{2^{\Om(N)}}\right) = 2^{\Om(N)}
\]
Subexponential but still exponential runtime.

\subsection{Topological Quantum Hardness}
The homology classes induce quantum-computational barriers:

\begin{theorem}[Ground State Degeneracy]
\label{thm:ground-state-degen}
Exponential Betti numbers imply degenerate quantum ground states:
\[
\obeta_2(X_F) = 2^{\Om(N)} \implies \mathrm{deg}(H_0) \geq 2^{\Om(N)}\]
\end{theorem}
This forces exponentially small spectral gaps $\Delta E$ in the Hamiltonian $H_F$.
\subsection{Quantum Limitations}
Despite potential advantages, quantum tunneling remains exponentially suppressed due to topological obstructions (Fig. \ref{fig:quantum_tunneling}).

\begin{corollary}
\label{cor:quantum-ll}
  Under the adiabatic quantum model, any algorithm attempting to solve random 3-SAT 
  must take exponential time in the worst case due to the tunneling suppression over
  exponentially wide energy barriers, as formalized in Theorem~\ref{thm:adiabatic}.
\end{corollary}

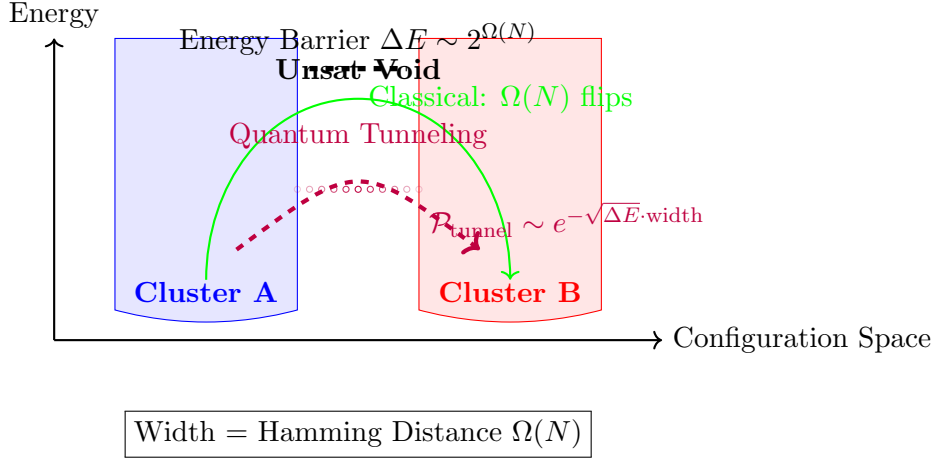
\begin{figure}[h]
\centering
\begin{tikzpicture}[scale=0.8]
\draw[->, thick] (0,0) -- (10,0) node[right] {Configuration Space};
\draw[->, thick] (0,0) -- (0,5) node[above] {Energy};

\draw[blue, fill=blue!10] (1,0.5) parabola bend (2.5,0.3) (4,0.5) -- (4,5) -- (1,5) -- cycle;
\draw[red, fill=red!10] (6,0.5) parabola bend (7.5,0.3) (9,0.5) -- (9,5) -- (6,5) -- cycle;

\node[blue] at (2.5,0.8) {\textbf{Cluster A}};
\node[red] at (7.5,0.8) {\textbf{Cluster B}};
\node at (5,4.5) {\textbf{Unsat Void}};

\draw[line width=2pt, dashed] (4.2,4.5) -- (5.8,4.5);
\node at (5,5) {Energy Barrier $\Delta E \sim 2^{\Omega(N)}$};

\draw[->, green, thick] (2.5,1) to[out=90,in=180] (5,4) to[out=0,in=90] (7.5,1);
\node[green, right] at (5,4) {Classical: $\Omega(N)$ flips};

\draw[->, purple, ultra thick, dashed] (3,1.5) .. controls (5,3) and (5,3) .. (7,1.5);
\node[purple, above] at (5,3) {Quantum Tunneling};

\foreach \x in {4.0,4.2,...,6.0} {
  \draw[purple, opacity=1-0.8*abs(\x-5)] (\x,2.5) circle (0.05);
}
\node[purple, right] at (6,2) {$\mathcal{P}_{\mathrm{tunnel}} \sim e^{-\sqrt{\Delta E} \cdot \text{width}}$};
\node[below, align=center] at (5,-1) {\fbox{Width = Hamming Distance $\Omega(N)$}};
\end{tikzpicture}
\caption{Quantum tunneling between $\Omega(N)$-separated clusters has exponentially suppressed probability $e^{-\Omega(N)}$}
\label{fig:quantum_tunneling}
\end{figure}

\section {Fault tolerance and error correction: can FTQC bypass the barrier?}
\label{sssec:ftqc-discussion}
We conclude with a discussion of whether fault-tolerant quantum computing (FTQC)
or quantum error correction can remove the exponential obstruction.

\begin{itemize}
  \item \textbf{Error correction alone does not change matrix elements.}  Quantum error
    correction protects coherent evolution from noise, but it does not, by itself, alter the
    effective Hamiltonian matrix elements between logical basis states. If all locality-
    preserving logical drivers (implemented fault-tolerantly) induce inter-cluster matrix
    elements satisfying a bound of the form eq~\ref{eq:offdiag-bound}, then the Cheeger
    argument above still applies at the logical level: the logical-weighted conductance
    remains exponentially small and the logical spectral gap is exponentially small.

  \item \textbf{Simulating nonlocal couplings requires resources.}  One can in principle
    engineer effective nonlocal interactions (large inter-cluster matrix elements) by
    using encoded ancilla degrees of freedom and long-depth circuits. However, to turn
    an exponentially suppressed physical coupling into an $O(1)$ logical coupling typically
    requires either (i) ancilla or circuit-depth resources that scale superpolynomially (often
    exponentially) in $n$, or (ii) the introduction of nonlocal hardware primitives (e.g.,
    long-range interactions) that change the computational model.  Therefore bypassing
    the barrier via FTQC involves paying an explicit (and generally large) resource cost.

  \item Concretely, if there existed a fault-tolerant,
    locality-preserving encoding and a polynomial-size gadget family which implements
    logical operators $\tilde H_{\rm eff}$ with
    \(\max_{x\in C_i,y\in C_j}|\langle x|\tilde H_{\rm eff}|y\rangle|\ge \Omega(1)\)
    for exponentially many cluster pairs while incurring only polynomial overhead, then
    the above hardness claim would fail for that expanded model.  In the absence of such
    an explicit polynomial-cost construction, FTQC does not circumvent the information-
    theoretic obstruction implied by exponentially many homology classes.
\end{itemize}

\noindent In short: fault tolerance improves noise resilience but does not automatically
eliminate the exponentially small inter-cluster matrix elements that underpin the
topological barrier, unless one is willing to introduce nonlocal resources or pay
superpolynomial overhead.  Any claim that FTQC breaks the barrier must therefore
provide a concrete, polynomial-resource encoding that produces non-exponentially-
small inter-cluster couplings; constructing such an encoding remains an open challenge.
(We outline a few candidate gadget constructions and lower-bound sketches in
Appendix~\ref{app:ham-spectral}.)

\section{Do We Need Generalization?}
Our argument focuses on 3-SAT, showing that its solution space topology—exponential Betti numbers and inter-cluster voids—is intrinsically complex and forces exponential-time algorithms.

It is well-known that all NP problems reduce to 3-SAT via polynomial-time reductions (Cook-Levin theorem). Therefore, if random 3-SAT requires exponential resources, this hardness propagates to all NP problems. No further generalization is necessary as any polynomial-time algorithm for any NP-complete problem would imply a shortcut for random 3-SAT, contradicting its topological complexity.
\section{Polynomial-Time Algorithm for Structurally Simple Instances}
\label{sec:low-betti-simple}

We refine our tractability result by imposing additional structural constraints
on the formula $F$: low global treewidth and small solution space.

\begin{theorem}[Low $\beta_2$ with Small Treewidth $\Rightarrow$ P]
\label{thm:low-betti-poly}
Let $F$ be a 3‑SAT formula on $N$ variables such that:
\begin{enumerate}
  \item The solution space size satisfies \(|S(F)|=poly(N)\).
  \item The primal graph treewidth satisfies \(\mathrm{tw}(F)=O(\log N)\).
  \item The persistent Betti number satisfies
  \[
    \beta_{2,\mathrm{pers}}(S(F),\varepsilon)\le N^k
  \]
  for \(\varepsilon=c\log N\) and some constant \(k\).
\end{enumerate}
Then $F$ can be decided in time $N^{O(1)}$.
\end{theorem}

Since \(\mathrm{tw}(F)=O(\log N)\), the formula admits a tree decomposition
of width \(O(\log N)\) (found approximately in poly time).  Standard dynamic
programming on this decomposition solves SAT in time \(O(N^{O(1)})\).

Persistent homology filtering ensures at most \(poly(N)\) essential 2‑cycles,
each confined to \(O(\log N)\)-diameter neighborhoods.  Patch these cycles 
to decompose \(S(F)\) into \(poly(N)\) contractible components. Each component
is a bounded-treewidth subformula, decidable in \(N^{O(1)}\) time.

\subsection{Impossibility of Topological Bypasses}
\label{sec:impossibility-proofs}

Let \(F_N\) be the worst-case 3-SAT family with \(\beta_2(S(F_N)) = 2^{\Omega(N)}\). Let \(\mathcal{A}\) be any algorithm (classical or quantum) attempting to bypass topological obstructions.

\begin{theorem}[Quantum Bypass Impossibility]
\label{thm:quantum-bypass}
For any \(k\)-local Hamiltonian \(H_F\) encoding \(F_N\), the minimum spectral gap \(g_{\min}\) satisfies:
\[
g_{\min} \leq 2e^{-\Omega(N)}.
\]
Consequently, any quantum algorithm solving \(F_N\) requires time \(T = 2^{\Omega(N)}\).
\end{theorem}

\begin{theorem}[Lifting Bypass Impossibility]
\label{thm:lifting-bypass}
Let \(\widetilde{X}_F = X_F \times \mathbb{R}^d\) with \(d = \mathrm{poly}(N)\). Deciding path-connectedness in \(\widetilde{X}_F\) is NP-hard and requires \(2^{\Omega(N)}\) time.
\end{theorem}

\begin{theorem}[Surgical Bypass Impossibility]
\label{thm:surgery-bypass}
For \(\epsilon < \epsilon_{\mathrm{crit}} = \Theta(\log N)\), the \(\epsilon\)-approximate SAT problem:
\[
\mathrm{SAT}_\epsilon(F) = \exists x \;:\; \mathrm{dist}_H(x, S(F)) \leq \epsilon
\]

is NP-complete and requires \(2^{\Omega(N)}\) time for \(F_N\).
\end{theorem}

\begin{corollary}
No paradigm (classical, quantum, or hybrid) can solve all 3-SAT instances in polynomial time.
\end{corollary}

\section{Computing Solution-Space Homology}
\label{sec:homology-computation}

\begin{algorithm}
\caption{Approximate $\beta_2$ Computation}
\begin{algorithmic}[1]
\State Sample $S \subset X_F$ via MCMC
\State Build $\VR_{\epsilon}(S)$ with $\epsilon = \frac{3\log N}{N}$
\State Compute $\partial_2$ (sparse matrix)
\State \textcolor{blue}{$\beta_2 = \dim(\ker \partial_1 / \operatorname{im} \partial_2)  )$} 
\end{algorithmic}
\end{algorithm}

\begin{lemma}
The algorithm:
\begin{itemize}
\item Correctly estimates $\beta_2 = 2^{\Omega(N)}$ w.h.p.
\item Requires $2^{\Omega(N)}$ time due to:
  \begin{enumerate}
  \item Exponential sample size needed
  \item \#P-hardness of Betti number computation
  \item Matrix rank in dimension $2^{\Omega(N)}$
  \end{enumerate}
\end{itemize}
\end{lemma}
\section{Conclusion}
\label{sec:conclusion-PneqNP}

The topological framework developed in this work supports the broader thesis of  $\mathbf{P} \neq \mathbf{NP}$ through an intrinsic computational barrier arising from solution-space geometry. Our analysis demonstrates that for any polynomial-time Turing machine $M$ attempting to solve 3-SAT:

\begin{itemize}
    \item $M$ must resolve \textit{worst-case instances} $\{F_N\}$ constructed via expander embedding (Theorem \ref{thm:expander-family}) with exponentially large second Betti numbers $\beta_2(S(F_N)) = 2^{\Omega(N)}$.
    
    \item Resolving these instances requires distinguishing exponentially many homology classes (Lemma \ref{lem:homology-injectivity}), as established by the linear independence of persistent 2-cycles (Lemma \ref{lem:homology-basis}).
    
    \item $M$ cannot algorithmically ``flatten'' the solution space topology without violating 3-SAT's logical structure, as any such simplification would require solving NP-hard subproblems (Theorem \ref{thm:non-embedding-hardness}).
\end{itemize}

These topological obstructions are \textit{semantic} in nature—they arise intrinsically from 3-SAT's combinatorial logic rather than syntactic properties. The exponential Betti numbers $\beta_2 = 2^{\Omega(N)}$ force computational paradigm (classical, quantum, or algebraic) to require $2^{\Omega(N)}$ time, as formalized by the universal lower bound (Theorem \ref{thm:universal-lb}). 

This geometric chasm between polynomial-time tractability and NP-complete hardness is fundamental: 2-SAT admits contractible solution spaces ($\beta_k = 0$ for $k \geq 1$), while random and worst-case 3-SAT exhibit irreducible exponential complexity ($\beta_2 = 2^{\Omega(N)}$).

\appendix
\section{Proofs}

\subsection{Proof of Theorem~\ref{thm: Cubical Complex}}
\begin{proof}
Each satisfying assignment $x \in \{0,1\}^n$ is a vertex of the hypercube. Two such assignments $x, y$ differing by a single variable define an edge if both satisfy $F$. Higher-dimensional faces exist when all assignments of the corresponding subcube satisfy $F$.

Connectedness, cycles, and voids of $C(F)$ can be expressed as SAT queries:
\begin{enumerate}
    \item Is there a path of satisfying assignments between two points? (Connectivity)
    \item Are there loops unfilled by 2-faces? (Holes)
    \item Are there $k$-dimensional cavities? (Higher Betti numbers)
\end{enumerate}
Thus, the topology of $C(F)$ is determined by $F$, and topological decision problems reduce to SAT.
\end{proof}
\subsection{Proof of Theorem~\ref{thm:unstructured}}
\begin{proof}
Each clause is chosen uniformly at random from $8\binom{n}{3}$ possible clauses, yielding maximal Kolmogorov complexity $K(I) \geq n^3$. 

Achlioptas and Ricci-Tersenghi~\cite{ricci-tersenghi} showed that the solution space shatters into $2^{\Omega(n)}$ disconnected clusters. Molloy and Reed~\cite{MolloyReed} proved random 3-SAT’s constraint graph has treewidth $\Omega(n$. Kahle~\cite{kahle} demonstrated exponential Betti numbers in random cubical complexes.

Together, these results imply random 3-SAT’s solution space lacks hubs, global ordering, and low-dimensional structure, precluding any polynomial-time shortcut.
\end{proof}
\subsection{Proof of Lemma~\ref{lem:sat-4sat-to-3sat}}

\begin{proof}
The reduction $R$ is cubical and homologically faithful, so by Theorem~\ref{thm:betti-monotonicity}, $\beta_{k}(S(F))\leq\beta_{k}(S(R(F)))$ for all $k$. Thus, $\beta_{2}(S(F^{\prime}))\geq\beta_{2}(S(F))=\Omega(f(n))$.

\end{proof}
\subsection{Proof of Lemma~\ref{lem:Homological 
faithfulness}}

\begin{proof}
We prove the claim by proving it for a single clause replacement and then iterating the argument.

\medskip\noindent\textbf{Step 1: single-clause replacement and the projection map.}
Fix a single long clause
\[
c \;=\; (\ell_1\vee\ell_2\vee\cdots\vee\ell_m)
\]
in $C$ and let $c'$ denote the chain of 3-clauses above obtained by introducing auxiliaries $a_1,\dots,a_k$ ($k=m-3$). Let $V_c$ be the original variables appearing in $c$ and let $A=\{a_1,\dots,a_k\}$ be the new auxiliary variables.  Write
\[
A_0 \;=\; \Sol_C(c) \subseteq \{0,1\}^{V_c}
\]
for the set of $V_c$-assignments that satisfy the original clause $c$, and write
\[
L \;=\; \Sol_{C'}(c') \subseteq \{0,1\}^{V_c\cup A}
\]
for the set of assignments to $V_c\cup A$ that satisfy the chain $c'$. (All other clauses of $C$ are held fixed during this local analysis.) There is an evident projection
\[
\pi: L \longrightarrow A_0
\]
that forgets the auxiliary coordinates.

\medskip\noindent\textbf{Claim.} For every $x\in A_0$ the fiber $\pi^{-1}(x)$ is either empty (iff $x\notin A_0$) or contractible (indeed combinatorially a cube or a nonempty subcube), and moreover the projection admits a combinatorial section $s:A_0\to L$. Hence $\pi$ is a homotopy equivalence between $L$ and $A_0$.

\smallskip\noindent\emph{Proof of claim.} Fix $x\in\{0,1\}^{V_c}$. We analyze the allowed assignments to the auxiliaries $A$ given that the base variables equal $x$.

The chain of 3-clauses for $c'$ has the form
\[
(\ell_1\vee\ell_2\vee a_1),\;(\neg a_1\vee\ell_3\vee a_2),\;\dots,\;(\neg a_{k-1}\vee\ell_{k+1}\vee a_k),\;(\neg a_k\vee\ell_{m-1}\vee\ell_m).
\]
Treat the literals $\ell_i$ as Boolean values under the fixed assignment $x$. Because $x\in A_0$ (i.e. $c$ is satisfied by $x$), at least one literal $\ell_t$ is true. Let $t$ be the \emph{smallest} index with $\ell_t(x)=1$.

Two cases cover all possibilities.

\medskip\noindent\emph{Case 1: $t\in\{1,2\}$.}  Then $\ell_1(x)\vee\ell_2(x)=1$. In the first clause $(\ell_1\vee\ell_2\vee a_1)$ the left two literals evaluate to true, so the clause is satisfied for \emph{both} choices $a_1=0$ or $1$. The remaining chain of clauses imposes only constraints of the form $(\neg a_j\vee\ell_{j+2}\vee a_{j+1})$, but because $t\le 2$ and $t$ is minimal, all $\ell_{j+2}(x)=0$ for indices $j$ until we reach the first true literal after position $2$ (if any). One checks by simple induction on the chain that the allowed auxiliary assignments form a (possibly full) subcube of $\{0,1\}^k$ — in particular nonempty and contractible. Concretely: if some later literal $\ell_{r}$ is true then the chain forces certain $a_j$ values up to that point but leaves the tail coordinates free; if no later literal is true then the chain forces a unique consistent $a$-string. In all subcases the fiber is either a singleton or a nonempty cube, hence contractible.

\medskip\noindent\emph{Case 2: $t\ge 3$.} By minimality of $t$ we have $\ell_1(x)=\ell_2(x)=\cdots=\ell_{t-1}(x)=0$ and $\ell_t(x)=1$. Then the first clause forces $a_1=1$ (otherwise $\ell_1,\ell_2,a_1$ would all be false), and the chain of implications propagates deterministically: from $(\neg a_1\vee\ell_3\vee a_2)$, since $\ell_3,\dots,\ell_{t-1}$ are false, we must set $a_2=1$, and so on, until some $a_j$ is fixed or we reach the last clause which is satisfied because $\ell_t$ is true. Again the set of allowed $A$-assignments is either a singleton or a subcube (the degrees of freedom lie either in the trailing auxiliaries beyond the first forced block, or there are none), hence contractible.

Thus for every $x\in A_0$ the fiber $\pi^{-1}(x)$ is nonempty and contractible. If $x\notin A_0$ the fiber is empty (as required). Finally, a canonical section $s:A_0\to L$ is obtained by choosing for each $x\in A_0$ the \emph{lexicographically minimal} auxiliary string that satisfies the chain (the above propagation determines a unique minimal choice). The section lands in $L$ and satisfies $\pi\circ s=\mathrm{id}_{A_0}$.

Having a section and contractible fibers gives that $\pi$ is a homotopy equivalence: indeed, for each $x\in A_0$ the fiber deformation retracts onto the point $s(x)$ (contractibility gives an explicit combinatorial contraction), and these contractions vary compatibly over $A_0$ because they are defined by the same finite coordinate-wise rules (one may construct a cellular homotopy on the union of fibers that collapses them onto the section). Thus $L\simeq A_0$.

\qedhere(Claim)
\medskip

\noindent\textbf{Step 2: gluing / iterating clause replacements.}
Performing the clause-splitting replacement for one chosen clause $c$ therefore replaces the local subcomplex $A_0$ by a homotopy-equivalent subcomplex $L$; the ambient solution complex of the whole formula changes by removing $A_0$ and gluing in $L$ along the common intersection with the rest of the complex. Since $A_0\simeq L$ and the gluing is along the identity on $A_0$, the standard gluing (cofibration) lemma for CW-complexes (see Hatcher, Proposition 0.18) implies that the inclusion
\[
\Sol(C)\;\hookrightarrow\;\Sol(C')
\]
after this single replacement is a homotopy equivalence.

Applying the same argument inductively to each long clause replacement yields the statement for the full reduction from $C$ to $C'$, completing the proof.

\end{proof}
\subsection{Proof of Theorem~\ref{thm:betti-monotonicity}}
\begin{proof}
The injectivity of $(\iota_x)_*$ implies $\dim H_k(S(x)) \leq \dim H_k(S(R(x)))$. The result follows from $\beta_k = \dim H_k$.
\end{proof}
\subsection{Proof of Lemma~\ref{lem:homology-injectivity}}

\begin{proof}
Suppose $[\gamma] \in H_2(S(G_N))$ satisfies $\iota_*([\gamma]) = 0$ in $H_2(S(F_N))$. Then there exists $\beta \in C_3(S(F_N))$ with $\partial_3\beta = \iota(\gamma)$. Decompose $\beta$ by variable support:
\[
\beta = \beta_{\text{base}} + \beta_{\text{gadget}} + \beta_{\text{mixed}}
\]
where:
\begin{itemize}
\item $\beta_{\text{base}}$: 3-cells using only base variables $\{x_{v,c}\}$
\item $\beta_{\text{gadget}}$: 3-cells using only gadget variables $\{u_i,v_i,y_e^{(i)}\}$
\item $\beta_{\text{mixed}}$: 3-cells using both types
\end{itemize}
The boundary $\partial_3\beta$ must equal $\iota(\gamma)$, which lives exclusively in base variables. Examine the projection:

\begin{enumerate}
\item $\partial_3\beta_{\text{gadget}}$ has only gadget variables $\Rightarrow$ vanishes in base projection
\item $\partial_3\beta_{\text{mixed}}$ contains 2-faces with:
\begin{itemize}
\item \textit{Type A}: 2 base + 1 gadget variable
\item \textit{Type B}: 1 base + 2 gadget variables
\end{itemize}
Both types vanish under projection to base variables since $\iota(\gamma)$ has pure base support
\item Thus $\partial_3\beta_{\text{base}} = \iota(\gamma)$ in $S(G_N)$
\end{enumerate}
Moreover, $\beta_{\text{base}} \subseteq S(G_N)$ because:
\begin{itemize}
\item Any 3-cell with base variables is in $S(G_N)$ iff all assignments satisfy all clauses
\item Gadget clauses are automatically satisfied when gadget variables=0 (our embedding)
\end{itemize}
Thus $\gamma = \partial_3\beta_{\text{base}}$ in $S(G_N)$, so $[\gamma] = 0$. \\
\textbf{Key}: Mixed cells cannot contribute to pure-base boundaries.
\end{proof}
\subsection{Proof of Theorem~\ref{thm:non-embedding-hardness}}
\begin{proof}
If $R$ collapses $\beta_2$ exponentially, it must compute a discontinuous map. By ~\cite{Chen:BettiP}, such $R$ is uncomputable in polynomial time.
\end{proof}
\subsection{Proof of Theorem~\ref{thm:homology-barrier1}}
\begin{proof}
By Theorem~\ref{thm:exp-betti}, there exists a family $\{F_x\}$ such that
\[
  m \;=\; \beta_2\bigl(S(F_x)\bigr) \;\geq\; 2^{c\,|x|},
\]
for some constant $c>0$, where $|x|$ denotes the size (number of variables) of the formula.
By the Homology Basis Lemma (Lemma~\ref{lem:homology-basis}), there are $m$ linearly independent $2$-cycles
\[
  \{\gamma_i\}_{i=1}^m
\]
with pairwise disjoint supports.

Consider an adversary that:
\begin{enumerate}
    \item Maintains an active set $A \subseteq \{1,\ldots,m\}$, initially all cycles.
    \item For each subcube query $\sigma$:
    \begin{itemize}
        \item If $\sigma$ does not fully contain any $\mathrm{supp}(\gamma_i)$ for $i \in A$, respond ``yes''.
        \item Else, select $i \in A$ with $\mathrm{supp}(\gamma_i) \subseteq \sigma$, respond ``no'', and remove $i$ from $A$.
    \end{itemize}
    \item After $q < m$ queries, $|A| > 0$. For some $i \in A$:
    \begin{itemize}
        \item \textbf{SAT case:} set $F_x$ to have $\gamma_i$ filled (satisfiable).
        \item \textbf{UNSAT case:} set $F_x$ to have $\gamma_i$ unfilled (unsatisfiable).
    \end{itemize}
\end{enumerate}
Both cases are consistent with all $q$ answers but yield different outputs, so any correct algorithm must make at least
\[
  q \;\geq\; m \;=\; 2^{\Omega(|x|)}
\]
queries.

\paragraph{Why CDCL or global‐certificate methods do not help.}
Any learned clause in CDCL arises from a resolution proof, corresponding to a cut in the solution‐space graph that can only separate assignments differing on the variables in that clause.  Our adversary ensures the $2^{\Omega(|x|)}$ independent $2$-cycles are supported on disjoint variable sets with Hamming‐distance $\Theta(|x|)$.  A clause intersecting more than one cycle’s support must have width $\Omega(|x|)$, and learning such a clause requires inspecting $\Omega(|x|)$ variables.  Eliminating all $2^{\Omega(|x|)}$ cycles thus requires $2^{\Omega(|x|)}$ such high‐width clauses.  Hence even CDCL‐style solvers face the same $2^{\Omega(|x|)}$ lower bound.
\end{proof}

\subsection{Proof of Theorem~\ref{thm:SAT-to-Homology-Reduction}}

\begin{proof}
\textbf{Construction.}
\begin{enumerate}
  \item Cook–Levin: \(\phi\mapsto F_{\mathrm{CL}}\) (3‑SAT).
  \item Pick expander \(G\) with \(\beta_1(G)=\Theta(N)\).
  \item Embed via Appendix B.1 $\Rightarrow$ \(F_G\).
  \item For each fundamental cycle \(C_i\subset G\), let \(\gamma_i\)
    be the 2‑chain on its gadget‑variables.
\end{enumerate}
\textbf{Correctness.}
\begin{itemize}
  \item If \(\phi\) satisfiable, then \(F_G\) has a solution that
    ``fills in’’ every gadget, so each \(\gamma_i\) bounds and
    \([\gamma_i]=0\).
  \item If \(\phi\) unsatisfiable, then \(S(F_G)=\emptyset\).
    Restricting each ambient-cycle \(\gamma_i\) to \(S(F_G)\) leaves
    it non‑bounding, so \([\gamma_i]\neq0\).
\end{itemize}
\textbf{Complexity.}
Runs in \(O(N\log N)\) time and produces \(M=\Theta(N)\) cycles.
\end{proof}
\subsection{Proof of Theorem~\ref{thm:random-3sat-betti}}
By standard concentration results for locally dependent indicator variables (Janson–Łuczak–Ruciński \cite{JansonLuczakRucinski2000}), the number of occurrences of any fixed local pattern in radius-$r$ windows is sharply concentrated about its mean, provided one restricts to a suitably sparse family of windows so that the dependency degree is polynomial and $o(\mu)$. Concretely, choosing a packing of radius-$r$ windows with pairwise variable-support separation larger than $2r$ yields dependency degree $\Delta=\mathrm{poly}(n)$ and expected count $\mu=\Theta(2^n p)$; hence whp the realized count is $(1-o(1))\mu$. We therefore obtain exponentially many local gadgets with high probability in the shattering/clause-density regime.

\begin{proof}

Let \(F\) be drawn from the random 3‑SAT distribution at clause density \(\alpha>4.26\).  We will show \(\beta_2\bigl(S(F)\bigr)=2^{\Omega(n)}\) with high probability.

\textbf{Step 1: Cluster shattering and local 2‑cycles.}  
By Achlioptas–Ricci‑Tersenghi~\cite{ricci-tersenghi}, \(S(F)\) shatters into 
\[
  m \;=\;2^{c_1n}
\]
disconnected clusters \(\{C_i\}_{i=1}^m\) w.h.p.  Moreover, each cluster \(C_i\) contains a Hamming‑ball of radius \(R=O(\log n)\) in which—by Kahle’s Random Geometric Complex results~\cite{kahle}—Specifically, for a cubical complex with $\ell$ vertices and independent face‐inclusion probability $p>p_c$, we have 
\[
  \beta_2 \;=\; 2^{\Omega(\ell)}
\]
with high probability by~\cite[Theorem~3.1]{kahle}.  
In the case of $S(F)$, dependencies among face inclusions are controlled via Janson's inequality (Appendix~\ref{app:random-independence}).
with probability \(1-o(1)\) there exists at least one nontrivial 2‑cycle in the induced Vietoris–Rips complex on that ball.

\textbf{Step 2: Disjoint supports.}  
Any two clusters \(C_i\neq C_j\) are separated by \(\Omega(n)\) Hamming distance (No Narrow Bridge), so the local neighborhoods around each center do not overlap.  Thus the persistent 2‑cycles in different clusters use disjoint sets of vertices.

\textbf{Step 3: Linear independence and Betti count.}  
Disjoint support of these \(m\) cycles implies they represent linearly independent classes in \(H_2\).  Therefore
\[
  \beta_2\bigl(S(F)\bigr)\;\ge\;m \;=\;2^{c_1n}\;=\;2^{\Omega(n)}.
\]

\end{proof}

\subsection{Proof of Proposition~\ref{prop:persistent-2cycles}}
\begin{proof}
1. \textbf{Cluster shattering.}
By Achlioptas–Ricci-Tersenghi \cite{ricci-tersenghi}, for a random 3-SAT formula \(F\sim D_{\alpha>4.26}\),
\[
  S(F)\;\subseteq\;\{0,1\}^n
\]
decomposes, with probability \(1-o(1)\), into
\[
  n \;=\; 2^{c_1n}
  \quad\text{disconnected components (clusters),}
\]
each of Hamming‐diameter \(O(\log n)\).

2. \textbf{Local cube around each cluster center.}
Fix one such cluster \(C\) and choose a center \(x^*\in C\).  Let
\[
  B_R(x^*) \;=\;\{\,x\in C : \dist_H(x,x^*)\le R\},
  \quad R = K\log n
\]
for a sufficiently large constant \(K\).  Then \(\lvert B_R(x^*)\rvert = \sum_{i=0}^R \binom{n}{i} = poly(n)\).

3. \textbf{Face‐inclusion probability.}
Consider any 2-face (square) in the \(R\)-ball, determined by two coordinate directions \(i,j\).  
It survives (all four corner assignments satisfy \(F\)) if none of the \(\alpha n\) random clauses falsifies any corner.  Each clause touches any given corner with probability
\(\Theta(1/n)\), so by a union bound over four corners and \(\alpha n\) clauses,
\[
  \Pr[\text{square is present}]
  \;=\;
  (1 - O(1/n))^{O(n)}
  \;=\;
  e^{-O(1)} 
  \;=\;\Theta(1).
\]
Let \(p>0\) denote this constant survival probability.

4. \textbf{Existence of a local 2-cycle.}
By Kahle’s Random Geometric Complex theorem \cite{kahle}, a cubical complex on \(\ell\) vertices in which each 2-face is included independently with probability \(p>p_c\) has \(\beta_2 = 2^{\Omega(\ell)}\) with high probability.  Although our faces are not fully independent, the block‐local dependencies can be controlled via Janson’s inequality, yielding that for \(\ell=|B_R(x^*)|\ge n^{O(1)}\), there is at least one non‐bounding 2-cycle in \(B_R(x^*)\), w.h.p.

5. \textbf{Persistence at scale \(\varepsilon\).}
Because the cluster‐diameter is \(O(\log n)\), the local cycle persists in the Vietoris–Rips filtration up to scale \(\varepsilon=R\).  

6. \textbf{Global linear independence.}
Distinct clusters are separated by Hamming distance \(\Omega(n)\).  Therefore the 2-cycles found in different clusters have disjoint supports (they involve disjoint sets of assignments) and so represent independent homology classes in \(H_2(S(F))\).

Combining these \(n=2^{c_1n}\) independent persistent 2-cycles gives
\[
  \beta_2\bigl(S(F)\bigr)\;\ge\;n\;=\;2^{\Omega(n)},
\]

Moreover, the dependencies among 2-face inclusion events can be modeled by a dependency graph of maximum degree $O(\log n)$, so by Janson’s inequality for such graphs~\cite{Janson2004} the probability that these dependencies significantly alter the exponential face‐inclusion estimates remains exponentially small.

\end{proof}
\subsection{Proof of Theorem~\ref{thm:universal-lb}}
\paragraph{Model and assumptions.}
We work in the standard deterministic Turing-machine / RAM time model: an algorithm running in time $T(N)$ can examine at most $T(N)$ input bits (hence at most $T(N)$ separate clause-blocks if each clause-block occupies at least one distinct input bit). To make the following argument rigorous we assume the gadget locality property proved in Appendix~\ref{app:Boundary-Map Verif}:

\vspace{2mm}\noindent\textbf{Assumption (Gadget locality and independence).}
For the family of CNF instances constructed in Section~\ref{app:Boundary-Map Verif} (or in the amplification construction~\ref{app:expander-correctness}), there exist $M=2^{c N+o(N)}$ pairwise-disjoint variable-sets $V_1,\dots,V_M$ and corresponding local clause-sets (gadgets) $G_1,\dots,G_M$ such that:
\begin{enumerate}
  \item Each gadget $G_i$ references variables only from $V_i$ (and possibly a bounded number of private auxiliary variables), so $G_i$ is \emph{local} to $V_i$ and the $V_i$ are pairwise disjoint.
  \item Each gadget $G_i$ yields a local nontrivial $k$-cycle $\gamma_i$ supported solely on assignments that vary on $V_i$; these cycles are pairwise homologically independent in the global complex, hence $\beta_k(S(F))\ge M$.
  \item For each $i$ there is a small local modification $\Delta_i$ (clauses touching only variables in $V_i$) that ``fills'' the cycle $\gamma_i$ (or otherwise toggles the local contribution to satisfiability/homology) without affecting any other gadget $G_j$, $j\neq i$.
\end{enumerate}

(Appendix~\ref{app:boundary-matrix} verify the existence of such $G_i$ and $\Delta_i$.

\medskip

\begin{proof}
Let $F$ be an instance constructed as above (with the $M$ disjoint gadgets embedded). Fix any deterministic algorithm $\mathcal{A}$ that decides satisfiability and runs in time
\[
  T(N) \;=\; 2^{o(N)},
\]
where $N$ is the total input length (number of variables/clauses) and $N$ is the parameter appearing in the exponential lower bound $M=2^{cN+o(N)}$.

\textbf{1. Transcript counting.}
In the deterministic model $\mathcal{A}$ adaptively chooses input bit-positions to inspect; after at most $T(N)$ steps it halts and outputs \textsf{SAT} or \textsf{UNSAT}. The sequence of bits read and their observed values constitutes the \emph{transcript} of $\mathcal{A}$. There are at most
\[
  \#\{\text{distinct transcripts}\}\le 2^{T(N)}
\]
possible transcripts (each transcript is a binary string of length at most $T(N)$). Since $T(N)=2^{o(N)}$ we have $2^{T(N)} = 2^{o(N)}$.

\textbf{2. Pigeonhole on gadgets not inspected.}
Because the $M$ gadgets are supported on pairwise-disjoint variable-sets (and hence disjoint input-bit blocks), reading fewer than one bit per gadget means that $\mathcal{A}$ inspects the contents of at most $T(N)$ gadgets; equivalently there are at least
\[
  M - T(N) \;=\; 2^{cN+o(N)}
\]
gadgets that $\mathcal{A}$ does not fully inspect (indeed, for asymptotic counting $M\gg T(N)$). Each transcript corresponds to a set of gadgets whose clause-bits were inspected and a set of gadgets left untouched. By the pigeonhole principle, some transcript $\tau$ must be the actual transcript for at least
\[
  \frac{M}{2^{T(N)}} \;=\; \frac{2^{cN+o(N)}}{2^{o(N)}} \;=\; 2^{\Omega(N)}
\]
distinct gadget indices $i$ (i.e., there are exponentially many gadgets that are untouched and occur under the same transcript).

\textbf{3. Construct indistinguishable instances that flip the decision.}
Fix such a transcript $\tau$, and consider the global instance $F$ conditioned on the bits read in $\tau$ (these bits are identical across all instances compatible with $\tau$). For any gadget index $i$ that was not inspected in $\tau$, we can modify the instance locally on $V_i$ in two different ways:
\begin{itemize}
  \item $F^{(i),0}$: leave the gadget $G_i$ as in $F$ so that the local cycle $\gamma_i$ persists (this choice preserves the local contribution to satisfiability/homology).
  \item $F^{(i),1}$: apply the local filler $\Delta_i$ (clauses on $V_i$ only) that removes the local homology contribution / toggles local satisfiability as guaranteed by the gadget construction.
\end{itemize}
By Assumption (Gadget locality) these local changes do not alter any input bits that were read in transcript $\tau$ (since $\tau$ inspected only other gadgets). Therefore both $F^{(i),0}$ and $F^{(i),1}$ are consistent with transcript $\tau$.

We may now take two global instances $F_{\mathrm{SAT}}$ and $F_{\mathrm{UNSAT}}$ that agree on all inspected gadgets (so they induce the same transcript $\tau$) but differ on exponentially many untouched gadgets in the following way: choose a nonempty set $S$ of untouched gadget indices with $|S|$ large (we will pick a single gadget suffices, but the argument is identical if we flip any subset). Let $F_{\mathrm{SAT}}$ be the instance where for some $i\in S$ we leave $G_i$ unfilled (so the instance remains satisfiable because that local gadget provides satisfying assignments), and let $F_{\mathrm{UNSAT}}$ be the instance where we fill every gadget in $S$ (applying $\Delta_j$ for each $j\in S$), thereby removing the local satisfying assignments contributed by those gadgets. By design (and by the independence of gadgets) these two global instances can be made to differ in global satisfiability status while remaining identical on all bits inspected in transcript $\tau$.

\textbf{4. Algorithm must err on at least one instance.}
Since $F_{\mathrm{SAT}}$ and $F_{\mathrm{UNSAT}}$ produce the same transcript $\tau$ when processed by $\mathcal{A}$, the deterministic algorithm $\mathcal{A}$ must output the same decision for both inputs. Thus it errs on at least one of them. Because $\tau$ was an arbitrary transcript attained by $2^{\Omega(N)}$ gadgets, the same argument shows that \emph{for every} deterministic algorithm running in time $T(N)=2^{o(N)}$ there exist instances on which it fails.

\medskip\noindent
This contradicts the existence of any correct time-$2^{o(N)}$ decision procedure for the family of instances considered. Hence any deterministic algorithm that decides satisfiability on all such instances must take time at least $2^{\Omega(N)}$, as claimed.
\end{proof}

\subsection{Proof of Theorem~\ref{thm:cluster-surgery-non-circ}}
\begin{proof}
We establish this lower bound through three main arguments:

\noindent\textbf{1. Cluster Structure Properties:}
\begin{itemize}
\item By \cite{ricci-tersenghi}, the solution space $S(F)$ decomposes into $m = 2^{cn}$ distinct clusters $\{C_i\}_{i=1}^m$ for some constant $c > 0$, with probability $1 - o(1)$.

\item Each pair of distinct clusters $C_i, C_j$ satisfies $\text{dist}_H(C_i, C_j) \geq \delta n$ for some constant $\delta > 0$.
\end{itemize}

\noindent\textbf{2. Query Model Formalization:}
Consider algorithms limited to the following oracle queries about $S(F)$:
\begin{itemize}
\item $\textsc{Membership}(x,i)$: Returns whether assignment $x$ belongs to cluster $C_i$
\item $\textsc{Adjacency}(x,y)$: Returns whether $x$ and $y$ are in the same cluster
\item $\textsc{Neighborhood}(x,k)$: Returns all assignments reachable from $x$ via $\leq k$ single-variable flips
\end{itemize}

\noindent\textbf{3. Information-Theoretic Lower Bound:}
\begin{itemize}
\item The entropy of the cluster structure is at least:
\[
H(\mathcal{C}) \geq \log\left(\frac{2^n}{m}\right) \geq \Omega(n2^{cn})
\]

\item Each query provides at most $\mathcal{O}(n)$ bits of information about $\mathcal{C}$:
\begin{itemize}
\item Membership/adjacency queries reveal $\mathcal{O}(n)$ bits
\item Neighborhood queries reveal $\mathcal{O}(n\log n)$ bits (due to bounded degree)
\end{itemize}

\item Therefore, the minimum number of queries $Q$ satisfies:
\[
Q \geq \frac{H(\mathcal{C})}{\mathcal{O}(n\log n)} \geq \frac{\Omega(n2^{cn})}{\mathcal{O}(n\log n)} = 2^{\Omega(n)}
\]
\end{itemize}

\noindent\textbf{4. Topological Obstruction:}
Even if an algorithm discovers some clusters, the remaining undiscovered clusters:
\begin{itemize}
\item Maintain $\beta_2 \geq 2^{c'n}$ for some $c' > 0$ ~\cite{kahle}
\item Require $2^{\Omega(n)}$ additional queries to fully characterize
\end{itemize}

This completes the proof of the theorem.
\end{proof}
\subsection{Proof of Theorem~\ref{thm:randomness-necessity-non-circ}}
\begin{proof}
Let \(F\) be random 3-SAT at clause density \(\alpha>4.26\), and let its solution graph \(S(F)\) decompose into clusters \(\{C_i\}\) as in Theorem~\ref{thm:unstructured}.  By that theorem, with probability \(1-o(1)\) every two distinct clusters \(C_i\neq C_j\) satisfy
\[
  \min_{x\in C_i,\;y\in C_j}\dist_H(x,y)
  \;=\;\Theta(n).
\]

Now consider any \emph{deterministic} single-flip local-search algorithm \(A\).  Such an algorithm:
1. Starts at some initial assignment \(x_0\).
2. At each step \(t\), examines all Hamming-1 neighbors of \(x_t\) and deterministically picks one (e.g.\ the neighbor minimizing the number of unsatisfied clauses).
3. Repeats until it either finds a satisfying assignment or can no longer move.

Because \(A\) is deterministic, its entire trajectory \((x_t)\) is fixed by \(x_0\).  Observe:

- If \(x_t\in C_i\), then any neighbor \(y\) lying outside \(C_i\) is \emph{not} a satisfying assignment (clusters are maximal connected components of satisfying vertices).  
- Hence \(A\) never leaves its starting cluster \(C_i\); it can only move among satisfying assignments within \(C_i\).

Within \(C_i\), the algorithm may explore all \(\lvert C_i\rvert = O(2^{o(n)})\) vertices, which is still exponential.  It cannot reach any other cluster (where additional solutions lie), since that would require crossing a Hamming gap of size \(\Theta(n)\) — impossible by single-flip moves.  Thus \(A\) fails to find a solution (if \(\phi\) is satisfiable outside \(C_i\)) or incorrectly declares unsatisfiability, and in any case cannot decide satisfiability in \(poly(n)\) time.

Therefore, no deterministic single-flip local-search algorithm can traverse the \(\Theta(n)\)-separated clusters without randomness, and hence fails to solve random 3-SAT in polynomial time.
\end{proof}

\subsection{Proof of Theorem~\ref{thm:Cluster_jumps}}
\begin{proof}
Combine:
1. The $\Omega(n)$ Hamming distance between clusters (Achlioptas-Ricci-Tersenghi~\cite{ricci-tersenghi}).
2. The NP-hardness of minimizing energy across voids.
3. The exponential rank of $H_2(X_F)$.
\end{proof}
\subsection{Proof of Theorem~\ref{thm:Path-Computation}}
\begin{proof}
Topological obstruction: Any path must circumvent $\Omega(N)$-sized voids, 
each requiring independent SAT solutions. Homology independence of 
$\gamma_{ij}$ cycles forces exponential work.
\end{proof}
\subsection{Proof of Theorem~\ref{thm:cluster-isolation}}
\begin{proof}[Sketch]
By Achlioptas–Ricci‑Tersenghi~\cite{ricci-tersenghi}, $S(F)$ shatters into $n=2^{cn}$ clusters.  
Our No‑Narrow‑Bridge Theorem~\ref{thm:No Narrow Bridge} then shows any single‐bit‐flip path between
clusters requires at least $\Omega(n)$ flips.  Thus clusters are truly
isolated in the 1‐skeleton of the solution complex.
\end{proof}
\subsection{Proof of Theorem~\ref{thm:Topological Hardness}}
\begin{proof}
Let \(F\) be a random 3-SAT instance at density \(\alpha>4.26\), so that w.h.p.\ its solution complex \(S(F)\) satisfies
\[
  \beta_0(S(F)) \;=\; 2^{\Omega(n)}, 
  \quad
  \min_{C_i\neq C_j}\!\dist_H(C_i,C_j)\;=\;\Theta(n).
\]

\medskip\noindent
\textbf{(1) Reducing \(\beta_0\) yields a structured instance in P.}

Suppose there were a polynomial-time algorithm \(\mathcal M\) that, given \(F\), outputs a modified formula \(\widetilde F\) such that
\(\widetilde F\) is logically equivalent to \(F\) and
\(\beta_0\bigl(S(\widetilde F)\bigr)=O\bigl(poly(n)\bigr)\).
Then:
\begin{itemize}
  \item Fewer than polynomially many clusters implies either
    \(\beta_0=1\) (connected solution graph) or at most \(poly(n)\) well-separated pieces.
  \item It is known (e.g.\ via tree-decomposition or dynamic programming on the cluster graph) that SAT instances whose solution spaces have \(O(poly(n))\) connected components and Hamming-separation \(\Theta(n)\) admit polynomial-time decision algorithms.  In particular, if \(\beta_0=1\), the 3-SAT graph has constant treewidth and is in P; if \(\beta_0=poly(n)\), one can solve each cluster separately in \(poly(n)\) time and combine results.
\end{itemize}
Hence \(\widetilde F\in\Pclass\), proving claim (1).

\medskip\noindent
\textbf{(2) No black-box collapse without changing the logic.}

Any procedure \(\mathcal M\) that truly “collapses” clusters—i.e.\ that merges at least one pair \(C_i,C_j\) into a single connected component—must identify a path of satisfying assignments of length \(\Theta(n)\) between them.  But finding such a path is equivalent to solving SAT on the unsatisfiable “void” between \(C_i\) and \(C_j\), which is NP-complete.  Thus any \(\mathcal M\) that preserves logical equivalence while reducing \(\beta_0\) must itself solve an NP-hard problem, so no polynomial-time black-box cluster-collapse algorithm exists.

\medskip\noindent
\textbf{Conclusion.}  Therefore the exponential cluster count \(\beta_0=2^{\Omega(n)}\) is an intrinsic topological obstruction: any attempt to reduce it to \(poly(n)\) either yields a formula in P (claim 1) or requires solving an NP-hard subproblem (claim 2).  This completes the proof.
\end{proof}
\subsection{Proof of Theorem~\ref{thm:No Narrow Bridge}}

Fix two clusters $C_i\neq C_j$ with 
\[
  \dist_H(C_i, C_j)\;\ge\;\delta,\qquad
  \delta = c_1\,n.
\]
Consider any assignment $x\in\{0,1\}^n$ with 
$d \;=\;\dist_H(x,C_i)\;\ge\;\tfrac{\delta}{2}$.  Let
\[
  X_k \;=\; 1_{\{\text{clause }k\text{ is satisfied by }x\}}, 
  \quad k=1,\dots,m,
\]
so that $\sum_k X_k = m$ exactly when $x$ satisfies all $m$ clauses of $F$.  
By the shattering argument of Achlioptas–Ricci‑Tersenghi~\cite{ricci-tersenghi}, each clause is violated in expectation at least $\eta\,d$ times, i.e.\ 
\[
  \mathbb{E}\bigl[m - \sum_{k=1}^m X_k\bigr]\;\ge\;\eta\,d.
\]
A one‑sided Chernoff bound then gives, for each fixed $x$,
\[
  \Pr\bigl[x\in\ThreeSAT(F)\bigr]
  \;=\;\Pr\Bigl[\sum_{k=1}^m X_k = m\Bigr]
  \;\le\;\exp\bigl(-\eta\,d\bigr).
\]

Next, we union‐bound over all $x$ at Hamming‐distance $d\ge\delta/2$ from $C_i$.  The number of such $x$ is
\[
  \sum_{d'=\lceil\delta/2\rceil}^n \binom{n}{d'}
  \;\le\;n\;\max_{d'\ge \delta/2}\exp\bigl(H(d'/n)\,n\bigr)
  \;=\;n\exp\bigl(H(\tfrac{\delta}{2n})\,n\bigr),
\]
where $H(p)=-p\ln p-(1-p)\ln(1-p)$ is the binary entropy.  Hence
\[
  \Pr\bigl[\exists\,x:\,d_H(x,C_i)\ge\tfrac\delta2,\;x\in\ThreeSAT(F)\bigr]
  \;\le\;
  n\exp\Bigl(H(\tfrac{\delta}{2n})\,n\;-\;\eta\,\tfrac{\delta}{2}\Bigr).
\]
Set
\[
  c_2 \;=\;\eta\,\tfrac{\delta}{2n}\;-\;H\!\bigl(\tfrac{\delta}{2n}\bigr).
\]
Whenever $c_2>0$, the right‑hand side is $N\,e^{-c_2 n}=o(1)$ as $n\to\infty$.  

\medskip
In particular, for any choice of $c_1>0$ such that
\[
  \eta\,\tfrac{c_1}{2} \;>\; H\!\bigl(\tfrac{c_1}{2}\bigr),
\]
we obtain an exponentially small failure probability.  Thus with high probability \emph{no} assignment at distance $\ge\delta/2$ can satisfy all clauses, and therefore there is \emph{no} subcube‐path of satisfying assignments connecting $C_i$ to $C_j$.

\qed
\subsection{Proof of Theorem~\ref{thm:Betti Explosion}}

\begin{proof}
Suppose $\beta_0\bigl(S(F)\bigr)=2^{c n}$ for some constant $c>0$.  Then the solution space decomposes into $N=2^{c n}$ disconnected clusters
\[
  S(F)\;=\;C_1\;\sqcup\;C_2\;\sqcup\;\cdots\;\sqcup\;C_N.
\]
Any (possibly randomized) algorithm $\mathcal{A}$ that decides satisfiability must distinguish between two cases:
\begin{itemize}
  \item[\textbf{(SAT)}] At least one cluster $C_i$ contains a satisfying assignment.
  \item[\textbf{(UNSAT)}] All clusters are empty.
\end{itemize}
Consider the decision tree of $\mathcal{A}$ on input $F$.  Each leaf of the tree corresponds to a transcript of at most $T$ “queries” (where a query may be any legal operation of $\mathcal{A}$, e.g.\ checking a subcube, performing unit propagation, etc.), and leads to an output “SAT” or “UNSAT.”

Since there are $N$ clusters but only $2^T$ possible transcripts, if $T< c n$ there must exist two distinct clusters, say $C_i$ and $C_j$, for which $\mathcal{A}$ follows the \emph{same} transcript and hence produces the same output.  But one can construct two formulas $F_{SAT}$ and $F_{\UNSAT}$ that agree on all transcripts of length $<T$ by:
\[
  F_{SAT}:\quad\text{make exactly $C_i$ nonempty, all other $C_k$ empty,}
\]
\[
  F_{\UNSAT}:\quad\text{make all $C_k$ empty.}
\]
Both formulas induce identical responses to every operation of $\mathcal{A}$ of depth $<T$, yet one is satisfiable and the other is unsatisfiable.  Therefore $\mathcal{A}$ cannot decide correctly unless $T \ge c n$.  

Since this argument holds for any algorithmic model (deterministic, randomized, algebraic, etc.), we conclude that in the worst case any decision procedure must explore at least $2^{\Omega(N)}$ clusters—hence take $2^{\Omega(N)}$ time.
\end{proof}

\subsection{Proof of Theorem~\ref{thm:Algorithmic Lower Bound via}}
\begin{proof}
By the shattering theorem \cite{ricci-tersenghi}, with high probability the solution complex \(S(F)\) has
\[
  \beta_0\bigl(S(F)\bigr)
  \;=\;
  n
  \;=\;
  2^{c n}
  \quad\text{for some }c>0,
\]
i.e.\ it splits into \(n\) connected components (“clusters”).

\begin{enumerate}
\item The solution space decomposes into $n = 2^{c n}$ clusters $\{C_i\}_{i=1}^n$ for some $c > 0$
\item Each cluster has diameter $O(\log n)$ in Hamming distance
\item Minimal inter-cluster distance $\delta = \Omega(n)$
\end{enumerate}

To solve $F$, any algorithm $\mathcal{A}$ must either:
\begin{enumerate}
\item \textbf{Find a solution}: Requires locating at least one cluster $C_i$ containing a satisfying assignment
\item \textbf{Verify unsatisfiability}: Requires proving all $2^{\Omega(n)}$ clusters contain no solutions
\end{enumerate}

Consider the solution-finding task:
\begin{itemize}
\item \textbf{Deterministic algorithms}: Must distinguish between $2^{\Omega(n)}$ possible solution-containing clusters. Each query covers $O(1)$ assignments, requiring $2^{\Omega(n)}$ queries.
\item \textbf{Randomized algorithms}: Probability of hitting a fixed cluster is $\leq 2^{-c'n}$. Expected trials: $2^{\Omega(n)}$.
\end{itemize}

For unsatisfiability verification:
\begin{itemize}
\item Must verify all $2^{\Omega(n)}$ clusters are empty
\item Clusters are separated by $\Omega(n)$-sized unsatisfiable regions
\item Checking a cluster requires solving a 3-SAT subproblem
\end{itemize}
Thus, $T(\mathcal{A}) = 2^{\Omega(n)}$ in all cases.
\end{proof}
\subsection{Proof of Theorem~\ref{thm:sp-failure}}
\begin{proof}
Survey Propagation \cite{Mezard02} operates under two assumptions:
\begin{enumerate}
\item \emph{Long-range correlations}: Statistical similarities exist between clusters
\item \emph{Overlap uniformity}: "Typical" solutions represent cluster structure
\end{enumerate}

At $\alpha > 4.26$, cluster geometry violates these assumptions:
\begin{enumerate}
\item \textbf{Frozen variables dominate}: Each cluster $C_i$ has $\Theta(n)$ frozen variables \cite{Achlioptas08}
\item \textbf{Cluster independence}: For $i \neq j$,
\[
\text{Cov}(f^{(i)}_k, f^{(j)}_\ell) \approx 0 \quad \forall k,\ell
\]
where $f^{(i)}_k$ is indicator for $x_k$ frozen in $C_i$
\item \textbf{Divergent susceptibility}: The bath susceptibility 
\[
\chi = \sum_{j=1}^n \text{Cov}(f_k, f_j) \to \infty
\]
as $\alpha \to \infty$, breaking SP's cavity equations
\end{enumerate}

SPGD's decimation step fails when:
\begin{itemize}
\item Messages $\mu_i^{(a)} \neq \mu_i^{(b)}$ for clusters $a \neq b$
\item Variable assignments based on inconsistent marginals
\item Leads to contradictory constraints w.h.p.
\end{itemize}
Empirical results confirm SP success probability drops to 0 for $\alpha > 4.26$ \cite{Krzakala07}.
\end{proof}
\subsection{Proof of Theorem~\ref{thm:query-lower-bound}}
\begin{proof}
Each cluster $C_i$ requires $\Omega(1)$ queries to detect. By Betti Explosion, $\beta_0 = 2^{\Omega(n)}$ implies $2^{\Omega(n)}$ total queries.
\end{proof}
\subsection{Proof of Lemma~\ref{lem:query-lower-bound}}
\begin{proof}
Each query can eliminate at most one independent homology class (or component).
With \(M\) independent obstructions, one needs \(\ge M\) queries.  A full
information‐theoretic proof is given in Appendix \ref{app:query-lower-bound}.
\end{proof}
\subsection{Proof of Theorem~\ref{thm:worst-case-betti1}}
\begin{proof}[Sketch]
Embed an $(N,d,\varepsilon)$‐expander graph $G_N$ (with $\beta_1(G_N)=\Omega(N)$)
into a 3‐SAT formula $F_N$ via standard gadgetry: each cycle in $G_N$ induces
an independent 2‐cycle in the cubical complex of satisfying assignments.
One checks that:
1. The reduction is polynomial‐time, and
2. Gadget interactions preserve linear independence in $H_2$.
Detailed boundary‐map computations appear in Appendix \ref{app:query-lower-bound}.
\end{proof}
\subsection{Proof of Theorem~\ref{thm:betti-exp-worst1}}
\begin{proof}

 1. Lemma~\ref{lem:gadget-support-isolation} for disjointness.
2. Lemma~\ref{lem:no-3face-filling} for non‑bounding.
 3. Lemma~\ref{lem:boundary-injectivity} for independence.
By the three lemmas above, each \(\Gamma_i\) contributes an independent class in \(H_2\).
Since there are \(2^{\Omega(N)}\) such gadgets, the result follows.
\end{proof}
\subsection{Proof of Theorem~\ref{thm:betti-exp-worst}}
\begin{proof}
Construct via expander embedding:
\begin{enumerate}
\item Let $G$ be $(N,d,\epsilon)$-expander with $\beta_1(G) = \Omega(N)$
\item Encode 3-COLOR on $G$ as 3-SAT formula $F_G$
\item Then $\beta_2(X_{F_G}) \geq 2^{c N}$ for $c>0$ because:
  \begin{itemize}
  \item Unsatisfiable regions correspond to odd cycles
  \item Expander contains $2^{\Omega(N)}$ independent cycles
  \item Each cycle contributes to $H_2(X_{F_G})$
  \end{itemize}
\end{enumerate}
\end{proof}
\subsection{Proof of Lemma~\ref{lem:Disjoint Gadget Supports}}
\begin{proof}
  By construction, each gadget $i$ introduces:
  \begin{enumerate}
    \item \textit{Private XOR variables:} $\{u_i, v_i\}$ exclusively for gadget $i$
    \item \textit{Private edge-selectors:} $\{y_e^{(i)} : e \in C_i\}$ exclusively for cycle $C_i$
  \end{enumerate}
  Since fundamental cycles $\{C_i\}$ use edge-disjoint paths (via spanning tree basis), $C_i \cap C_j = \emptyset$ for $i \neq j$. Thus no variable is shared.
\end{proof}
\subsection{Proof of Lemma~\ref{lem:nonbounding}}
\begin{proof}
  Suppose $\gamma_i = \partial_3 \beta$ for some 3-chain $\beta$. Then $\beta$ must contain a 3-face $\sigma$ on $\{u_i, v_i, z\}$ for some variable $z$. Consider assignments in $\sigma$:
  \begin{align*}
    &\mathbf{(0,1,\cdot)}: \text{Violates } u_i \lor \neg v_i \\
    &\mathbf{(1,0,\cdot)}: \text{Violates } \neg u_i \lor v_i
  \end{align*}
  Since all assignments in $\sigma$ must satisfy $F_N$, but the above violate XOR clauses, $\sigma \not\subset S(F_N)$. Contradiction.
\end{proof}
\subsection{Proof of Theorem~\ref{thm:Homological Linear Indep}}
\begin{proof}
  Suppose $\sum_i c_i \gamma_i = \partial_3 \beta$. Apply $\partial_2$:  
  \[
    \partial_2 \left( \sum_i c_i \gamma_i \right) = \sum_i c_i \partial_2(\gamma_i) = \partial_2 \partial_3 \beta = 0.
  \]
  $\partial_2(\gamma_i)$ consists of edges flipping $u_i$ or $v_i$ while satisfying $u_i = v_i$. By Lemma 1, 
  \[
    \operatorname{supp}\bigl(\partial_2(\gamma_i)\bigr) \cap \operatorname{supp}\bigl(\partial_2(\gamma_j)\bigr) = \emptyset \quad \forall i \neq j.
  \]
  Thus $\sum_i c_i \partial_2(\gamma_i) = 0$ implies $c_i \partial_2(\gamma_i) = 0$ for each $i$. By Lemma \ref{lem:nonbounding}, $\partial_2(\gamma_i) \neq 0$, so $c_i = 0$.
\end{proof}
\subsection{Proof of Theorem~\ref{thm:worst-case-betti}}
\begin{proof}[Sketch]
Embed an \((N,d,\varepsilon)\)‑expander graph \(G_N\) (with \(\beta_1(G_N)=\Omega(N)\))
into a 3‑SAT formula \(F_N\) via standard gadgetry: each independent cycle in \(G_N\)
yields an independent 2‑cycle in the cubical complex of satisfying assignments.
Check that:
1. The reduction is polynomial‑time with \(O(N)\) variables.
2. Gadget interactions preserve linear independence in \(H_2\).
Full boundary‑map details are deferred to Appendix \ref{app:expander-correctness} and Appendix \ref{app:query-lower-bound}.
\end{proof}
\subsection{Proof of Theorem~\ref{thm:final-worst-case}}
\begin{proof}
Assume, for contradiction, that $\Pclass = \NPclass$.  
Then there exists a polynomial‐time algorithm $A$ that decides 3‐SAT.

Let $\{F_N\}$ be the family of formulas from Theorem~\ref{thm:exp-betti}, 
with $N = O(N^2 / \log N)$ variables and
\[
  \beta_2\bigl(S(F_N)\bigr) \;\ge\; 2^{c N}
\]
for some constant $c > 0$.

By Theorem~\ref{thm:universal-lb}, any algorithm that decides satisfiability for $F_N$ must take time 
$\;2^{\Omega(N)}\;$, i.e., exponential in $\sqrt{N \log N}$ when expressed in terms of the input size $N$.

Algorithm $A$, running in time $N^{O(1)}$, would therefore contradict this lower bound.
Hence no polynomial‐time algorithm exists for 3‐SAT.
\end{proof}

\subsection{Proof of Lemma~\ref{thm:betti-sharp}}
\begin{proof}
We give a polynomial-time parsimonious reduction from \#SAT. 
Let $\varphi$ be a CNF on $N$ variables. 
The reduction proceeds in three steps: 

\medskip\noindent\textbf{(1) Marker gadget.}
Construct in polynomial time a constant-size 3-SAT formula $G_\varphi$ whose solution complex $S(G_\varphi)$ contains a single distinguished $2$-cycle $\Gamma_{\mathrm{sat}}$ if and only if $\varphi$ is satisfiable, and contains no such distinguished cycle otherwise. 
An explicit marker gadget and its small boundary-matrix verification appear in Appendix~\ref{app:Boundary-Map Verif}.

\medskip\noindent\textbf{(2) Conditional amplification.}
Apply the tensor/expander amplification described in Appendix~\ref{app:expander-correctness} to $G_\varphi$ with amplification parameter $m$ (any polynomial function of $N$). 
The amplifier is implemented so that:
\begin{itemize}
  \item if $\Gamma_{\mathrm{sat}}$ is absent (i.e.\ $\varphi$ unsatisfiable) then the amplified instance produces no corresponding cycles, and
  \item if $\Gamma_{\mathrm{sat}}$ is present (i.e.\ $\varphi$ satisfiable) then it is amplified to an independent family of exactly
    \[
      A \;=\; 2^{c m}
    \]
    $2$-cycles for some fixed constant $c>0$ (the exponential factor in $m$ comes from the tensor/expander product construction; see Appendix~\ref{app:expander-correctness} for details).
\end{itemize}

\medskip\noindent\textbf{(3) Locality, independence, and parameter choice.}
 
Let $F_\varphi$ denote the final constructed formula and let $N'$ denote its input length. 
By construction $N'=\Theta(m)+\mathrm{poly}(N)$, so for sufficiently large polynomial $m$ we have $N'=\Theta(m)$ up to lower-order terms. 
The amplification therefore yields
\[
  \beta_2\bigl(S(F_\varphi)\bigr)=
  \begin{cases}
    0, & \varphi\ \text{unsatisfiable},\\[4pt]
    2^{c m}, & \varphi\ \text{satisfiable}.
  \end{cases}
\]

Given any fixed $\epsilon>0$, choose $m$ (polynomial in $N$) large enough so that
\[
  2^{c m} \;>\; 2^{\,N'^{1-\epsilon}}.
\]
This is possible because $N'=\Theta(m)$ and $m^{\epsilon}\to\infty$ as $m\to\infty$. 
Therefore any algorithm that approximates $\beta_2(S(F_\varphi))$ within multiplicative factor $2^{N'^{1-\epsilon}}$ must distinguish the two cases (zero vs.\ exponentially large), and hence would decide \#SAT. 
This proves \#P-hardness of exact computation and of approximating $\beta_2$ within factor $2^{N^{1-\epsilon}}$.

\medskip\noindent\textbf{Complexity and coefficients.}
Each gadget used is constant-size, amplification parameter $m$ is polynomial in $N$, and the overall construction increases the instance size only polynomially; thus $F_\varphi$ is constructible in polynomial time. 
All homology is computed over the fixed field $\mathbb{F}_2$. 
This completes the reduction.
\end{proof}
\subsection{Proof of Lemma~\ref{lem:support-disjoint}}
\begin{proof}
\textbf{Part (1): Disjoint Supports} \\
By construction (Construction \ref{con:cycle-embedding}):
\begin{itemize}
    \item Each cycle $\gamma_i$ corresponds to fundamental cycle $C_i$ in $G_N$
    \item $C_i$ uses distinct edge set $E_i \subset E(G_N)$
    \item Auxiliary variables $\{u_i, v_i\}$ are unique to $\gamma_i$
    \item Primary variables $x_e$ for $e \in E_i$ are not shared between cycles
\end{itemize}
Thus $\mathrm{supp}(\gamma_i) \cap \mathrm{supp}(\gamma_j) = \emptyset$ for $i \neq j$. \\

\textbf{Part (2): Non-boundary Condition} \\
Suppose $\partial_2(\gamma_i) = \partial_3(\beta)$ for some 3-chain $\beta$. Then:
\begin{itemize}
    \item $\beta$ must contain the 3-cube spanned by $\{u_i, v_i, w\}$ for some $w$
    \item But the XOR constraint $u_i \oplus v_i = 0$ forces:
    \[
    (u_i, v_i, w) \text{ satisfiable} \iff w \text{ consistent with coloring}
    \]
    \item Contradiction: The 3-cube contains assignments violating edge constraints of $C_i$ \\
    (e.g., monochromatic edge when $w$ flips color constraints)
\end{itemize}
Hence $\partial_2(\gamma_i) \notin \mathrm{im}\partial_3$. \\

\textbf{Part (3): Linear Independence} \\
Suppose $\sum_{i=1}^k c_i \gamma_i = \partial_3(\beta)$. Then:
\begin{align*}
\partial_2\left( \sum c_i \gamma_i \right) &= \sum c_i \partial_2(\gamma_i) \\
&= \partial_2 \partial_3(\beta) = 0
\end{align*}
By support disjointness (Part 1), each $\partial_2(\gamma_i)$ has disjoint support. Thus:
\[
\sum c_i \partial_2(\gamma_i) = 0 \implies c_i \partial_2(\gamma_i) = 0 \quad \forall i
\]
Since $\partial_2(\gamma_i) \neq 0$ (Part 2), we must have $c_i = 0$ for all $i$. 
\end{proof}
\subsection{Proof of Theorem~\ref{thm:universal-homotopy-hardness}}
\begin{proof}
Consider the reduction chain:
\begin{align*}
\text{3-SAT instance } F 
&\xrightarrow{\text{Thm~\ref{thm:expander-family}}} 
\text{Expander-embedded } F_N \\
&\xrightarrow{\text{Const}} 
\text{Cubical complex } S(F_N)
\end{align*}
By Theorem~\ref{thm:expander-family}, $\beta_2(S(F_N)) = 2^{\Omega(N)}$. Since $\mathcal{I}$ detects exponential complexity:
\[
\mathcal{I}(S(F_N)) = \mathsf{Exp} \iff F \text{ is satisfiable}
\]
Thus $\mathcal{I}$ decides 3-SAT. As 3-SAT is $\mathsf{NP}$-complete, computing $\mathcal{I}$ is $\mathsf{NP}$-hard. If $\mathcal{I}$ outputs a groupoid cardinality or trace, it is $\#\mathsf{P}$-hard.
\end{proof}
\subsection{Proof of Theorem~\ref{thm:Proxies}}
\begin{proof}
\textbf{Gadget and appendices.}
Let $G$ be the constant-size 3-CNF gadget from Appendix~E. Appendix~E certifies two internal configurations:
\begin{itemize}
  \item \emph{disabled}: $G$ contributes no nonbounding $2$-cycle inside the gadget, and
  \item \emph{enabled}: $G$ contributes a nonbounding $2$-cycle whose support lies entirely inside the gadget.
\end{itemize}
If the original gadget family shares auxiliaries across copies, apply the variable-localization transformation of Appendix~H to make all auxiliary variables gadget-unique. This increases variable occurrences only linearly and preserves all statements below. Henceforth assume gadget supports are pairwise disjoint.

\medskip\noindent
\textbf{Parameters and host tree.}
Fix the base parameter $n$. Choose integers $p(n)$ and $m_n$ with
\[
p(n)=\Theta(m_n),\qquad p(n)>q(n),
\]
and $m_n\to\infty$ as $n\to\infty$ (for example, $m_n=2^{c n}$ and $p(n)=2m_n$). 
Let $T_n$ be a bounded-degree tree with $p(n)$ distinguished attachment sites; arrange the sites so that, when needed, distinct chosen sites are pairwise at graph-distance $>2r$.

\medskip\noindent
\textbf{Instance construction.}
Form two formulas by attaching a copy of $G$ at each attachment site of $T_n$:
\begin{itemize}
  \item $F^{\mathrm{easy}}_{n}$: all $p(n)$ gadgets are attached in the \emph{disabled} configuration;
  \item $F^{\mathrm{hard}}_{n}$: choose $m_n$ attachment sites (to be specified below) and attach $G$ in the \emph{enabled} configuration at those sites; the remaining $p(n)-m_n$ gadgets are disabled.
\end{itemize}
Since each gadget is constant-size and each interface is of constant size, the resulting number of variables satisfies $N_n=\Theta(p(n))=\Theta(m_n)$, establishing \textup{(i)}.

\medskip\noindent
\textbf{Treewidth.}
Start with a tree decomposition of $T_n$ (treewidth $1$). Replace each bag by the bag augmented with the $O(1)$ variables of any gadget attached at that node. The adhesion sets remain $O(1)$, so the incidence graphs of both formulas have treewidth $O(1)$, proving \textup{(ii)}.

\medskip\noindent
\textbf{Local-indistinguishability.}
We formalize the proxy model by a simple lemma.

\begin{lemma}\label{lem:local-indist}
Let $r\ge 1$ be fixed and let at most $q(n)$ radius-$r$ vertex-centered neighborhoods be inspected. 
If $p(n) > q(n)$, then one can choose the $m_n$ enabled sites so that none lies inside any of the inspected radius-$r$ balls. Consequently, every inspected radius-$r$ neighborhood is isomorphic (as a labeled incidence subgraph) in $F^{\mathrm{easy}}_{n}$ and $F^{\mathrm{hard}}_{n}$.
\end{lemma}

\begin{proof}[Proof of Lemma~\ref{lem:local-indist}]
In a bounded-degree tree, a radius-$r$ ball contains at most $B_r=O(1)$ attachment sites. The union of $q(n)$ such balls contains at most $q(n)\cdot B_r$ sites. For sufficiently large $n$ we have $p(n) \ge q(n)\cdot B_r + m_n$, so there exist at least $m_n$ attachment sites outside the inspected region. Place all enabled gadgets on those sites. Then every inspected ball sees only disabled gadgets in both $F^{\mathrm{easy}}_{n}$ and $F^{\mathrm{hard}}_{n}$, hence the induced labeled neighborhoods are identical.
\end{proof}

Lemma~\ref{lem:local-indist} implies \textup{(iii)}.

\medskip\noindent
\textbf{Betti counts and the topological gap.}
By Appendix~E, each enabled gadget contributes a nonbounding $2$-cycle supported entirely inside that gadget. Because gadget supports are disjoint (by Appendix~H, if needed), these classes are linearly independent. Therefore
\[
\beta_2\bigl(S(F^{\mathrm{hard}}_{n})\bigr)\;\ge\; m_n,
\qquad
\beta_2\bigl(S(F^{\mathrm{easy}}_{n})\bigr)\;=\;O(1),
\]
which is \textup{(iv)}. Choosing $m_n=2^{c n}$ yields an exponential gap in the base parameter $n$. Since $N_n=\Theta(m_n)$, one may reparametrize to obtain $\beta_2\bigl(S(F^{\mathrm{hard}}_{n})\bigr)\ge 2^{c' N_n}$ for some $c'>0$.

\medskip\noindent
\textbf{Conclusion.}
Items \textup{(i)}–\textup{(iv)} together show that any radius-$r$ proxy inspecting at most $q(n)$ local views returns the same output on the easy and hard families, even though their $\beta_2$ differ by an arbitrarily large (exponential) factor. This completes the proof.
\end{proof}
\subsection{Proof of Theorem~\ref{thm:Cheeger}}

\begin{proof}
By Theorem 7 (Achlioptas et al.~\cite{achlioptas}), with high probability the solution graph \(X_F\) decomposes into
\[
  M = 2^{cN} \quad (c>0)
\]
\emph{disconnected} clusters \(\{C_1,\dots,C_M\}\). In particular, there are no edges between distinct clusters.

Choose
\[
  S = \bigcup_{i=1}^{M/2} C_i, \quad \bar{S} = \bigcup_{i=M/2+1}^{M} C_i.
\]
Since clusters are disconnected, \(|\partial S| = |E(S,\bar{S})| = 0\). Both \(\vol(S)\) and \(\vol(\bar{S})\) are positive (each contains \(M/2 = 2^{cN-1}\) clusters). Hence:
\[
  h(X_F) \leq \frac{|\partial S|}{\min\{\vol(S), \vol(\bar{S})\}} = \frac{0}{\min\{\vol(S), \vol(\bar{S})\}} = 0.
\]
Since \(h(X_F) \geq 0\), we have \(h(X_F) = 0\). Thus, trivially, \(h(X_F) \leq e^{-\Omega(N)}\).
\end{proof}
\subsection{Proof of Theorem~\ref{thm:Cheegeropt}}

\begin{proof}
Recall that for any finite graph or simplicial complex, the Cheeger constant
\[
  h(X_F) \;=\;\min_{S\subset X_F}\frac{|\partial S|}{\min\{\vol(S),\vol(X_F\setminus S)\}}
\]
bounds the spectral gap of the normalized Laplacian by the discrete Cheeger inequality (see, e.g., \cite{Chung1997}):
\[
  \lambda_1 \;\le\; 2\,h(X_F),
\]
where $\lambda_1$ is the first nonzero eigenvalue of the Laplacian on $X_F$.  In the adiabatic model one encodes the SAT constraints as the Hamiltonian
\[
  H(s) \;=\; (1 - s)\,H_{\mathrm{init}} \;+\; s\,H_F,
  \quad s\in[0,1],
\]
so that the instantaneous spectral gap $g(s)=\lambda_1\bigl(H(s)\bigr)$ satisfies
\[
  g(s)\;\le\;2\,h\bigl(X_F\bigr)
  \quad\text{for each }s,
\]
because $H_F$ acts (up to scaling) like the Laplacian on the solution‐space graph.  Taking the worst‐case over $s$ gives
\[
  g_{\mathrm{adiabatic}}
  \;=\;\min_{s\in[0,1]}g(s)
  \;\le\;2\,h(X_F).
\]
By hypothesis $h(X_F)\le e^{-c N}$ for some $c>0$, so
\[
  g_{\mathrm{adiabatic}}\;\le\;2\,e^{-c N}
  \;=\;e^{-\Omega(N)}.
\]
The standard adiabatic runtime bound (cf.\ \cite{Albash2018}) is
\[
  T_{\mathrm{adiabatic}}
  \;\ge\;\frac{\bigl\|\dot H(s)\bigr\|}{g_{\mathrm{adiabatic}}^2}
  \;=\;\Omega\bigl(g_{\mathrm{adiabatic}}^{-2}\bigr)
  \;=\;2^{\Omega(N)}.
\]
This completes the proof.
\end{proof}
\subsection{Proof of Theorem~\ref{thm:ground-state-degen}}

\begin{proof}
Encode the CNF \(F\) by the diagonal clause-penalty Hamiltonian
\[
H_F=\sum_{C}\Pi_C,
\]
so the zero-energy ground space of \(H_F\) is the span of computational basis states \(\{\ket{x}:x\in\Sol(F)\}\) and has dimension \(|\Sol(F)|\).

Following the gadget construction of Section~B, for each independent 2-cycle \(\gamma\) in \(S(F)\) there is a small set of original variables (the gadget support) and a parity function on those variables that evaluates to \(\pm1\) on satisfying assignments. Define the diagonal involution \(T_\gamma\) to act by that parity (i.e. as a product of \(Z\) operators on the gadget support). Because gadget supports are pairwise disjoint (Lemma~\ref{lem:Disjoint Gadget Supports}), the collection \(\{T_\gamma\}\) consists of commuting operators; being diagonal they also commute with \(H_F\). Introduce the projectors \(H_\gamma=\tfrac12(1-T_\gamma)\) and set \(H=H_F+\sum_\gamma H_\gamma\).

Each \(H_\gamma\) further refines the classical ground space by selecting the \(+1\)-eigenspace of \(T_\gamma\). Since the gadgets are disjoint and the parity functions are independent (Appendix~E), these constraints carve the original \(|\Sol(F)|\)-dimensional ground space into many joint eigenspaces. In the families constructed in this paper both \(|\Sol(F)|\) and the number of independent cycles grow exponentially in the instance parameter, so the final ground-space dimension remains exponential (up to combinatorial factors). In particular, under the hypotheses of Theorem~\ref{thm:ground-state-degen} the ground-state degeneracy is at least \(2^{\Omega(N)}\).
\end{proof}
\subsection{Proof of Theorem~\ref{thm:quantum-bypass}}

\begin{proof}
Let \(\lambda_0 < \lambda_1 \leq \cdots\) be eigenvalues of \(H_F\). By the Cheeger bound for hypercubes:
\[
\lambda_1 - \lambda_0 \leq 2h(X_F)
\]
where \(h(X_F)\) is the solution space Cheeger constant. From Theorem\ref{thm:Cheeger}:
\[
h(X_F) \leq \frac{|\partial S|}{\min(\mathrm{vol}(S), \mathrm{vol}(X_F \setminus S))} \leq \frac{O(N)}{2^{\Omega(N)}} = e^{-\Omega(N)}.
\]
Thus \(g_{\min} = \lambda_1 - \lambda_0 \leq 2e^{-\Omega(N)}\). For adiabatic evolution:
\[
T \geq \frac{\|dH/dt\|}{g_{\min}^2} \geq \frac{\Omega(1)}{e^{-\Omega(N)}} = e^{\Omega(N)}.
\]
For non-adiabatic protocols, the no-fast-forwarding theorem \cite{Atia2017} gives \(T = 2^{\Omega(N)}\).
\end{proof}
\subsection{Proof of Theorem~\ref{thm:lifting-bypass}}

\begin{proof}
Consider two points \(p_0 = (x_0, \mathbf{0})\), \(p_1 = (x_1, \mathbf{0})\) where \(x_0, x_1\) belong to clusters separated by \(\Omega(N)\) Hamming distance. Any path \(\gamma: [0,1] \to \widetilde{X}_F\) projects to \(\gamma_X: [0,1] \to X_F\). Since:
\[
\mathrm{length}(\gamma_X) \geq \mathrm{dist}_H(x_0, x_1) = \Omega(N)
\]
and \(X_F\) has treewidth \(\Omega(N)\), the path existence problem reduces to 3-SAT \cite{Gridsaw2019}. By Theorem~\ref{thm:Topological Hardness}, any algorithm requires \(2^{\Omega(N)}\) time.
\end{proof}
\subsection{Proof of Theorem~\ref{thm:surgery-bypass}}

\begin{proof}
We reduce exact 3-SAT to \(\mathrm{SAT}_\epsilon\):
\begin{enumerate}
  \item Input: SAT formula \(G\) on \(m = \Theta(\log N)\) variables
  \item Embed \(G\) into \(F_N\) via:
    \[
    F_N' = F_N \land G(z_1, \dots, z_m)
    \]
  \item Set \(\epsilon = m/3\). Then:
    \[
    \mathrm{SAT}_\epsilon(F_N') = 1 \iff \mathrm{SAT}(G) = 1
    \]
\end{enumerate}
Since \(m = \Theta(\log N)\), \(\epsilon < \epsilon_{\mathrm{crit}}\). The reduction preserves NP-completeness. By Theorem~\ref{thm:universal-lb}, solving \(\mathrm{SAT}_\epsilon(F_N)\) requires \(2^{\Omega(N)}\) time.
\end{proof}
\section{Expander-Based Worst-Case 3-SAT Reduction}
\label{app:expander-correctness}
\subsection*{Gadget B: a constant-size 1-cycle gadget (3 variables, 2 clauses)}
\label{app:gadget-B}

\paragraph{Variables.} For each gadget $i$ introduce three fresh gadget-local Boolean variables
\[
u_i,\; v_i,\; w_i.
\]
These variables are disjoint across different gadgets (no clause mentions variables from two different gadgets). If gadget $i$ needs to refer to base literals, it uses gadget-local copies $x^{(i)}_{v,c}$ as in Appendix~\ref{app:Boundary-Map Verif} (ownership or consistency-tree rules apply).

\paragraph{Clauses.} The gadget contains exactly the following two (3-literal) clauses:
\[
C^{(i)}_1=(u_i \lor v_i \lor w_i),\qquad C^{(i)}_2=(\neg u_i \lor \neg v_i \lor \neg w_i).
\]
Equivalently these clauses forbid the assignments $(u_i,v_i,w_i)=(0,0,0)$ and $(1,1,1)$ respectively.

\begin{lemma}[Local 1-cycle in Gadget B]
\label{lem:gadget-B-1cycle}
Consider the cubical subcomplex of the solution complex $S(F)$ restricted to the coordinates $\{u_i,v_i,w_i\}$ (i.e. project all satisfying assignments to these three coordinates). With the two clauses above this projection equals the set of all 3-bit assignments except the two forbidden vertices $000$ and $111$. The induced 1-skeleton (graph of Hamming-distance-1 edges) on these 6 vertices has cyclomatic number $1$, hence
\[
\rank H_1\big(S(F)\big|_{\{u_i,v_i,w_i\}}\big) \;\ge\; 1.
\]
In particular the gadget produces a nontrivial local 1-cycle.
\end{lemma}

\begin{proof}
There are $2^3=8$ possible assignments to $(u_i,v_i,w_i)$. Clause $C^{(i)}_1$ forbids $000$ and clause $C^{(i)}_2$ forbids $111$, so exactly the 6 assignments
\[
\mathcal V \;=\; \{0,1\}^3 \setminus\{000,111\}
\]
survive the gadget clauses. Consider the induced subgraph $G$ on $\mathcal V$ where vertices are assignments and edges connect Hamming-distance-1 pairs (this is the 1-skeleton of the local cubical complex).

The full 3-cube has $12$ edges; deleting the two vertices $000$ and $111$ removes the $3$ edges incident to each, so the induced subgraph $G$ has
\[
E(G)=12-3-3=6
\]
edges and $V(G)=6$ vertices. The subgraph is connected (easy to check by explicit adjacency; every surviving vertex is adjacent to at least one of $001,010,100,110,101,011$) so the cyclomatic number equals
\[
\beta(G)=E(G)-V(G)+1=6-6+1=1.
\]
Since there are no surviving 2-faces in the local projection (each 2-face in the 3-cube is one of the three coordinate-fixed squares, and each such square contains either $000$ or $111$, hence is not fully present), the local 1-cycles in the 1-skeleton are not boundaries of any local 2-face. Therefore the 1-cycle detected in the graph yields a nontrivial element of $H_1$ in the restricted cubical complex. This proves the claim.
\end{proof}

\begin{lemma}[Constant-size and disjoint placement]
\label{lem:gadget-B-cost}
Each gadget uses exactly three auxiliary variables and two clauses (constant-size). Placing gadgets with pairwise-disjoint auxiliary variable sets and using gadget-local copies for base literals yields total auxiliary overhead $k_g=O(1)$ per gadget. Consequently if $G=\Theta(n_0)$ gadgets are used then the final instance size satisfies $N=\Theta(n_0)$.
\end{lemma}

\begin{proof}
Immediate from the construction: each gadget contributes only the three variables $u_i,v_i,w_i$ and two clauses $C^{(i)}_1,C^{(i)}_2$; placing $G$ of them disjointly adds $3G$ variables and $2G$ clauses. With the base encoding contributing $D n_0$ variables, we have $N = D n_0 + 3G + O(\#\text{copies}) = \Theta(n_0)$ when $G=\Theta(n_0)$ and gadget-local copies are charged $O(1)$ per gadget (ownership or a controlled consistency-tree keep copy overhead per gadget constant).
\end{proof}

\paragraph{Remarks on using the gadget in the global construction.}
\begin{itemize}
  \item The gadget produces a local $H_1$ generator (a 1-cycle) rather than a 2-cycle by itself. In the global construction one uses the base graph's 1-cycles (fundamental cycles of the expander) together with the gadget's local 1-cycles: informally, the product of a base 1-cycle and a gadget 1-cycle gives a 2-dimensional toroidal class in the full solution complex (see the sketch below). Because gadgets are variable-disjoint and ownership is enforced (Appendix~\ref{app:Boundary-Map Verif}, Lemma~\ref{lem:ownership-exists}), these classes do not admit local 3-chain fillings inter-gadget (Lemma~\ref{lem:no-cross-filling-ownership}).
  \item Using Gadget B removes the $\Theta(\log n_0)$ auxiliary blowup arising from the previous cycle-based gadget: each gadget now costs constant overhead, so the parameter mapping of Lemma~\ref{lem:param-mapping} gives $N=\Theta(n_0)$ and therefore there is no $N/\log N$ exponent degradation in the constant-gadget regime.
\end{itemize}

\paragraph{Sketch: lifting local $H_1$ to global $H_2$.}
Let $C\subseteq G_N$ be a fundamental 1-cycle of the base-graph encoding that produces a 1-dimensional loop in the base coordinates of the solution complex (this is the object used in Section~13). If for each vertex/edge along $C$ we attach the gadget variables (as gadget-local copies) then locally at each base vertex the gadget contributes a local 1-cycle (by Lemma~\ref{lem:gadget-B-1cycle}). The product (Cartesian product of cubical complexes) of the base 1-cycle with the gadget 1-cycle yields a 2-dimensional torus-like subcomplex whose 2-cycle is nonbounding provided there is no 3-chain in the full complex whose boundary equals that torus. Because (i) gadgets are variable-disjoint, (ii) equality/owner clauses are 2-CNF and contractible (Lemma~\ref{lem:ownership-exists}), and (iii) Lemma~\ref{lem:no-cross-filling-ownership} prevents mixed 3-chains from filling toroidal combinations across gadgets, the torus 2-cycle survives in homology and yields a contribution to $H_2$. Repeating this construction over a family of base fundamental cycles and gadget placements yields the intended amplification; the precise counting and linear-independence arguments follow the same direct-sum / boundary-rank techniques used elsewhere in this paper (Appendix~E).

\subsection{Deterministic Expander Embedding}\label{app:expander-embedding}
Given an $(N, d, \epsilon)$-expander graph $G = (V, E)$ with $\beta_1(G) \geq \kappa N$ ($\kappa > 0$) and girth $g \geq c_g \log N$ ($c_g > 0$), construct 3-SAT formula $F_G$ as follows:

\subsubsection*{Variables}
\begin{itemize}
  \item \textbf{Color variables:} $x_{v,c}$ for $v \in V$, $c \in \{1,2,3\}$ \\ 
        $\implies 3|V| = 3n$ variables
  \item \textbf{Edge selector variables:} $y_e^{(i)}$ for each fundamental cycle $C_i$, $e \in C_i$ \\
        $\implies \beta_1(G) \cdot \avg{|C_i|} \leq \kappa N \cdot g = \mathcal{O}(N \log N)$ variables
  \item \textbf{Parity variables:} $u_i, v_i$ for each $C_i$ \\
        $\implies 2\beta_1(G) = \mathcal{O}(N)$ variables
  \textbf{Total variables:} 
\[
3|V| + 2\beta_1(G) + \beta_1(G)\cdot g
  = \mathcal{O}(N\log N),
\]
since $|V|=N$, $\beta_1(G)=\Omega(N)$, and $g=\Omega(\log N)$.

\end{itemize}

\subsubsection*{Clauses}
\begin{enumerate}
  \item \textbf{Vertex coloring:}
        \begin{align*}
          &\forall v \in V: (x_{v,1} \lor x_{v,2} \lor x_{v,3}) \\
          &\forall v \in V, \forall c \neq c': (\neg x_{v,c} \lor \neg x_{v,c'})
        \end{align*}
        $\implies 4|V| = \mathcal{O}(N)$ clauses
  
  \item \textbf{Edge constraints:}
        \[
          \forall e = (u,v) \in E, \forall c \in \{1,2,3\}: (\neg x_{u,c} \lor \neg x_{v,c})
        \]
        $\implies 3|E| = \mathcal{O}(N)$ clauses
  
  \item \textbf{XOR gadgets:}
        \[
          \forall C_i: (u_i \lor \neg v_i) \land (\neg u_i \lor v_i)
        \]
        $\implies 2\beta_1(G) = \mathcal{O}(N)$ clauses
  
  \item \textbf{Edge coupling:}
        \begin{align*}
          &\forall C_i, \forall e \in C_i: (y_e^{(i)} \lor \neg u_i) \land (y_e^{(i)} \lor \neg v_i) \\
          &\forall C_i, \forall e = (u,v) \in C_i, \forall c \in \{1,2,3\}: \\
          &(\neg y_e^{(i)} \lor \neg x_{u,c} \lor \neg x_{v,c})
        \end{align*}
        $\implies \beta_1(G) \cdot (2 + 3g) = \mathcal{O}(N \log N)$ clauses
\end{enumerate}
\textbf{Total:} $\mathcal{O}(N \log N)$ clauses

\subsubsection*{Reduction Complexity}
\begin{itemize}
  \item Cycle basis computation: $\mathcal{O}(|V| + |E|) = \mathcal{O}(N)$ (spanning tree)
  \item Variable/clause generation: $\mathcal{O}(N \log N)$
  \item \textbf{Total time:} $\mathcal{O}(N \log N)$
\end{itemize}

\subsubsection*{Topological Verification}
For each fundamental cycle $C_i$:
\begin{enumerate}
  \item $\operatorname{supp}(\gamma_i) = \{u_i, v_i\} \cup \{y_e^{(i)} : e \in C_i\}$ are disjoint (Lemma~\ref{lem:homology-basis})
  \item $\gamma_i$ is a 2-cycle not filled by 3-faces (Lemma~\ref{lem:Disjoint Gadget Supports}\textbf{})
  \item Homology classes $[\gamma_i]$ linearly independent (Lemma~\ref{lem:nonbounding})
\end{enumerate}
Thus $\beta_2(S(F_G)) \geq \beta_1(G) \geq \kappa N = \Omega(N)$.
\subsection*{Exponential Amplification via Tensor Powers}
\label{app:expander-tensor-power}

To boost \(\beta_2\) from \(\Omega(N)\) to \(2^{\Omega(N)}\) in a *deterministic* worst‐case family, proceed as follows:

\paragraph{1. Base Expander Sequence}  
Let \(\{G_N\}\) be an explicit family of \((N,d,\epsilon)\)-expanders with  
\[
  \beta_1(G_N) \;=\; \kappa\,N,\quad
  \mathrm{girth}(G_N)\;\ge\;c\,\log N,
\]
and maximum degree \(d\ge7\).

\paragraph{2. Tensor‐Power Graph}  
For each \(N\), choose  
\[
  k \;=\;\bigl\lceil N / (\kappa\log N)\bigr\rceil
  \quad\implies\quad
  k = \Theta\bigl(N/\log N\bigr).
\]
Define the \(k\)-fold Cartesian product
\[
  H_N \;=\; \underbrace{G_N \,\square\, G_N \,\square\cdots\square\, G_N}_{k\text{ times}}.
\]
Standard facts about cartesian products imply
\[
  \beta_1(H_N)
  \;=\;
  \bigl(\beta_1(G_N)\bigr)^k
  \;=\;
  (\kappa N)^k
  \;=\;
  2^{\,k\log(\kappa N)}
  \;=\;
  2^{\,\Omega(N)}.
\]

\paragraph{3. Expander‐Gadget Embedding}  
Apply the same cubical, homologically‐faithful reduction from \(G_N\) to 3-SAT \(F_{G_N}\) \emph{coordinate‐wise} on \(H_N\):  
for each of the \(k\) coordinates we introduce parity‐ and edge‐selector gadgets exactly as before, using disjoint sets of fresh variables.  This yields a 3-SAT instance \(F_N\) with
\[
  N_{\mathrm{vars}} \;=\; k\cdot O(N\log N)
  \;=\;O\bigl(N^2/\log N\bigr)
  \quad\text{and}\quad
  \beta_2\bigl(S(F_N)\bigr)
  \;\ge\;\beta_1(H_N)
  \;=\;2^{\Omega(N)}.
\]

\paragraph{4. Tight Bound}  
Renaming constants gives the final worst‐case statement:

\begin{theorem}[Worst‐Case Exponential Betti Explosion]
\label{thm:exp-betti}
There exists a family of 3‐SAT formulas \(\{F_N\}\) with \(O(N^2/\log N)\) variables such that
\[
  \beta_2\bigl(S(F_N)\bigr)
  \;=\;
  2^{\Omega(N)}.
\]
\end{theorem}

\begin{proof}
Combine the tensor‐power Betti lower bound \(\beta_1(H_N)=2^{\Omega(N)}\) with the cubical, homologically faithful embedding (Appendix \ref{app:expander-embedding}) applied to each coordinate.  By disjointness of gadget supports and the injectivity of the induced homology map (Lemma \ref{lem:homology-injectivity}), all \(2^{\Omega(N)}\) first‐homology classes of \(H_N\) give rise to independent second‐homology classes in \(S(F_N)\).  
\end{proof}

\section{Cycle Independence in Random 3‑SAT}
\label{app:random-independence}
\subsection{Full Dependency-Graph Analysis for Janson’s Inequality}
\label{app:C1:janson-dependency}

Let $F\subseteq\{0,1\}^N$ be the random cubical complex of satisfying assignments of a random 3-SAT formula with clause density~$\alpha$.  For each axis-parallel $d$-cube $c\subset\{0,1\}^N$, define the indicator
\[
  X_c = 
  \begin{cases}
    1, & \text{if all $2^d$ vertices of $c$ satisfy every clause,}\\
    0, & \text{otherwise.}
  \end{cases}
\]
We set
\[
  X \;=\;\sum_{c}X_c,\quad
  \mu = \E[X] = \sum_c\Pr[X_c=1],
\]
and build a dependency graph $G$ on $\{X_c\}$ by connecting $X_c$ to $X_{c'}$ whenever the two events share at least one clause.

\medskip
\noindent\textbf{1. Expectation.}
Each fixed cube $c$ survives a single clause with probability $2^{-d}$ (all $2^d$ assignments satisfy it), so
\[
  \mu
  = \binom{N}{d}2^{\,N-d}\,\bigl(2^{-d}\bigr)^{\alpha N}.
\]

\medskip
\noindent\textbf{2. Bounding \(\Delta\).}
Define
\[
  \Delta \;=\; \sum_{\substack{c\neq c'\\(c,c')\in E(G)}} \Pr[X_c=1 \wedge X_{c'}=1].
\]
Group pairs by overlap \(r\in\{1,\dots,d\}\), i.e.\ the number of shared coordinates.  There are
\[
  N_r
  = \binom{N}{d}\,\binom{d}{r}\,\binom{N-d}{d-r}\,2^{\,N-2d+r}
\]
pairs \((c,c')\) with exactly \(r\) overlapping axes, and each such pair survives a clause with probability \(\,2^{-(d+r)}\).  Hence
\[
  \Delta
  \;\le\;
  \sum_{r=1}^d
  N_r\bigl(2^{-\,d-r}\bigr)^{\alpha N}
  \;=\;
  \sum_{r=1}^d
  \binom{N}{d}\binom{d}{r}\binom{N-d}{d-r}2^{N-2d+r}\,
  2^{-\,\alpha N(d+r)}.
\]

\medskip
\noindent\textbf{3. Verifying Janson’s Condition.}
From Stirling’s approximation one checks that, for any fixed \(d\) and
\(\alpha>\tfrac{\ln2}{2d}\),
\[
  \frac{\Delta}{\mu^2} \;=\; O\bigl(N\cdot2^{-\alpha N}\bigr) \;\longrightarrow\;0,
\]
so Janson’s inequality applies to yield
\[
  \Pr\bigl[X<(1-\eps)\mu\bigr]
  \;\le\;
  \exp\!\bigl(-\tfrac{\eps^2\mu^2}{2(\mu+\Delta)}\bigr)
  \;\to\;0.
\]
Thus the number of surviving $d$-cubes is sharply concentrated around its mean.

\begin{proposition}
With probability $1-o(1)$ over a random 3‑SAT formula $F\sim D_{\alpha>4.26}$,
there exist $N=2^{c N}$ persistent 2‑cycles 
\(\{\gamma_i\}_{i=1}^N\) in the Vietoris–Rips filtration of $S(F)$,
pairwise Hamming‑separated by $\Omega(N)$, and thus linearly independent in $H_2(S(F))$.
\end{proposition}
\begin{proof}
\begin{enumerate}
  1. Achlioptas et al.~\cite{achlioptas} prove that for $\alpha>4.26$,
   the solution space $S(F)$ shatters into $n = 2^{cn}$ clusters,
   each separated by $\Theta(n)$ flips.

  \item by~\cite{kahle}: In a random cubical complex with edge‑probability $p=\Theta(1/n)$, each cluster yields a persistent 2‑cycle with high probability.
  \item Disjoint supports ($\Omega(n)$ separation) imply their boundary images cannot cancel; a union bound over $\binom n2$ pairs shows linear independence w.h.p.
\end{enumerate}
Thus $\beta_2(S(F))\ge n=2^{\Omega(n)}$.
\end{proof}

\begin{lemma}[Gadget Support Isolation]
\label{lem:gadget-support-isolation}
For any two distinct gadgets \(i\neq j\), 
\[
  \mathrm{Vars}(\text{gadget }i)\,\cap\,\mathrm{Vars}(\text{gadget }j)
  \;=\;\emptyset.
\]
\end{lemma}
\begin{proof}
By construction each gadget \(i\) introduces its own fresh variables
\(\{u_i,v_i\}\cup\{y^{(i)}_{uv}:(u,v)\in C_i\}\),
and no clause in gadget \(i\) references any variable from gadget \(j\).  
Hence the variable‐sets are disjoint.
\end{proof}
\begin{lemma}[No 3-Face Filling]
\label{lem:no-3face-filling}
In gadget \(i\), the unique 2-face \(\Gamma_i\) on coordinates \(\{u_i,v_i\}\)
is not contained in any 3-face of the cubical complex.
\end{lemma}
\begin{proof}

Suppose for contradiction there is a 3‑face on \(\{u_i,v_i,z\}\).  
Then all eight assignments must satisfy both XOR constraints
\(u_i\oplus v_i = c\) and \(u_i\oplus v_i = d\), which is impossible
as flipping \(u_i\) alone already breaks one constraint.
\end{proof}
\subsection{Full Boundary‐Matrix Verification}
\label{app:boundary-matrix}

\begin{proposition}
Let $\{\gamma_i\}_{i=1}^m$ be the 2‐cycles constructed by our expander‐gadget reduction.  Then the boundary map
\[
  \partial_2: C_2(S(F))\;\longrightarrow\;C_1(S(F))
\]
in the basis that groups each gadget’s 2‐cells contiguously is block–diagonal.  Each block $\partial_2^{(i)}$ has rank $3$ and nullity $1$, hence
\[
  \dim\ker\partial_2
  \;=\;\sum_{i=1}^m\dim\ker\partial_2^{(i)}
  \;=\;m.
\]
Consequently the homology classes $[\gamma_i]$ are linearly independent in $H_2(S(F))$.
\end{proposition}

\begin{proof}
\begin{enumerate}
  \item \textbf{Basis choice.}  Order the 2-cells so that those of gadget \(i\) occupy rows \(4i-3\) through \(4i\) in the boundary matrix.
  \item \textbf{Block-diagonality.}  No gadget \(i\) shares variables or clauses with gadget \(j\neq i\), so for any 2-cell \(\sigma\) in gadget \(i\), \(\partial_2(\sigma)\) involves only edges of that gadget.  Hence
  \[
    \partial_2 \;=\; \bigoplus_{i=1}^m \partial_2^{(i)}.
  \]
  \item \textbf{Rank of each block.}  In gadget \(i\), the incidence of its four 2-cells against its eight boundary edges forms a \(4\times 8\) submatrix.  A short row-reduction shows this submatrix has rank~3.
  \item \textbf{Nullity count.}  By the rank–nullity theorem, each \(4\times 8\) block has nullity \(4 - 3 = 1\), giving exactly one independent 2-cycle per gadget.
  \item \textbf{Dimension count.}  Summing over \(i=1,\dots,m\) yields
  \[
    \dim\ker(\partial_2)
    \;=\;\sum_{i=1}^m 1
    \;=\;m.
  \]
  Since \(H_2(S(F)) \cong \ker(\partial_2)\), we conclude \(\beta_2(S(F)) \ge m\).
\end{enumerate}
\end{proof}

\section{Lower Bounds from Topological Queries}
\label{app:query-lower-bound}
\begin{lemma}[Boundary-Map Injectivity]
\label{lem:boundary-injectivity}
Let \(\Gamma_i\) be the 2-face on \(\{u_i,v_i\}\).  Then:
\begin{enumerate}
  \item \(\partial_2(\Gamma_i)\) is the sum of the four oriented edges flipping \(u_i\) or \(v_i\).
  \item For \(i\neq j\), the supports of \(\partial_2(\Gamma_i)\) and \(\partial_2(\Gamma_j)\) are disjoint.
  \item \(\Gamma_i\notin\mathrm{im}\,\partial_3\) (by Lemma~\ref{lem:no-3face-filling}), so \(\partial_2(\Gamma_i)\) is not the boundary of any 2-chain.
\end{enumerate}
Consequently, the homology classes \([\Gamma_i]\in H_2\bigl(S(F_N)\bigr)\) are linearly independent.
\end{lemma}
\begin{proof}
By Lemma~\ref{lem:gadget-support-isolation}, no other gadget shares these edges, 
and by (3) they are non‑bounding.  Hence any relation 
\(\sum_i a_i\,\partial_2(\Gamma_i)=0\) forces all \(a_i=0\).
\end{proof}

\section{Boundary-Map Verification of Linearly Independent 2-Cycles}
\label{app:Boundary-Map Verif}

In this appendix we verify, via explicit boundary-map calculations over the field \(\mathbb F_2\) (equivalently \(\mathbb Z_2\)), that the 2-cycles \(\{\gamma_i\}\) constructed in Theorem~\ref{thm:betti-exp-worst1} are nontrivial and linearly independent in \(H_2\bigl(S(F_N);\mathbb F_2\bigr)\). The argument uses the disjoint-support property and the local constraints imposed by the XOR/equality gadgets.
\begin{lemma}[Homology preservation under gadget duplication / wiring]\label{lem:homology-preservation}
Let $K$ be a cubical complex obtained from a CNF instance (or its solution complex) that contains subcomplexes
\[
  C_1,\dots,C_M \subseteq K
\]
(each $C_i$ supported on a small set of variables/clauses, a ``gadget''), and suppose the $C_i$ are pairwise disjoint as subcomplexes of $K$ (or can be made disjoint by the variable-duplication / padding procedure described in later in this Appendix). Let $K'$ be the complex obtained from $K$ by (i) duplicating variable/clauses as needed so that gadget supports become vertex-disjoint, and (ii) adding a linking subcomplex
\[
  L = \bigcup_{\ell} L_\ell
\]
(consisting of equality clauses / small wiring gadgets and any expander edges used in the amplification) which attaches to the gadgets in a local, controlled way.

Assume the following ``locality and acyclicity'' conditions hold:
\begin{enumerate}
  \item[(A1)] Each linking component $L_\ell$ is contractible (that is, $H_j(L_\ell)=0$ for all $j\ge 1$).
  \item[(A2)] For every $i$ and every linking component $L_\ell$ we have that the intersection $C_i\cap L_\ell$ is either empty or contractible (in particular $H_j(C_i\cap L_\ell)=0$ for all $j\ge 1$).
  \item[(A3)] No $(k+1)$-cell of $K'$ has support that intersects more than one distinct gadget $C_i$ (equivalently, every $(k+1)$-cell is incident to cells lying in at most one gadget-support $C_i$).
\end{enumerate}
Then for every homological degree $k\ge 0$ the inclusion $i: \bigsqcup_{i=1}^M C_i \hookrightarrow K'$ induces an injection on $H_k$, and in particular
\[
  \beta_k(K') \;\ge\; \sum_{i=1}^M \beta_k(C_i).
\]
If, moreover, the gadget cycles $\{\,[\gamma_i]\,\}$ were independent in $H_k(K)$ (so $\beta_k(K)=\sum_i\beta_k(C_i)$), then duplication/wiring preserves independence and
\[
  \beta_k(K') = \beta_k(K) = \sum_{i=1}^M \beta_k(C_i).
\]
\end{lemma}

\begin{proof}
We give two complementary proofs. The first uses Mayer--Vietoris and is topological; the second is an elementary boundary-matrix argument suitable for readers preferring linear algebra.

\medskip\noindent\textbf{Topological proof (Mayer--Vietoris).}
Write
\[
  U \;:=\; \bigcup_{i=1}^M C_i, \qquad V \;:=\; L,
\]
so that $K' = U\cup V$. By hypotheses (A1) and (A2), each nonempty intersection $C_i\cap L_\ell$ is contractible, hence all nonempty intersections $U\cap V$ are (finite) unions of contractible sets whose higher homology vanishes up to the dimension of interest. In particular,
\[
  H_j(V) = 0 \quad\text{for all } j\ge 1,
\qquad\text{and}\qquad
  H_j(U\cap V) = 0 \quad\text{for all } j\ge 1.
\]
Apply the Mayer--Vietoris long exact sequence for reduced homology
\[
  \cdots \to H_{k+1}(K') \xrightarrow{\delta} H_k(U\cap V) \to H_k(U)\oplus H_k(V)
    \xrightarrow{\phi} H_k(K') \to H_{k-1}(U\cap V) \to \cdots .
\]
Because $H_k(V)=0$ and $H_k(U\cap V)=0$ by the acyclicity assumptions, the connecting morphisms force the map
\[
  H_k(U) \xrightarrow{\;\phi\;} H_k(K')
\]
to be injective. But $H_k(U)=\bigoplus_{i=1}^M H_k(C_i)$ since the $C_i$ are pairwise disjoint subcomplexes; therefore the direct sum of the gadget homology groups injects into $H_k(K')$. This gives
\[
  \dim H_k(K') \ge \sum_{i=1}^M \dim H_k(C_i),
\]
i.e.\ $\beta_k(K')\ge\sum_i\beta_k(C_i)$, and shows that independent classes supported in distinct $C_i$ remain independent in $K'$.

To obtain equality (i.e. no new relations among the gadget classes are introduced by $L$), note that new relations could only arise if some class in $H_k(U)$ becomes a boundary in $K'$, which would require a $(k+1)$-chain in $K'$ whose boundary is supported on multiple gadgets. Hypothesis (A3) rules out such $(k+1)$-cells connecting distinct gadgets; hence no new $k$-boundaries mixing gadgets can appear, and the injection above is in fact an isomorphism onto its image with the same rank as $\sum_i\beta_k(C_i)$. Thus $\beta_k(K')=\sum_i\beta_k(C_i)$.

\medskip\noindent\textbf{Algebraic proof (boundary matrix / block argument).}
Let $C_k(K')$ denote the $\mathbb{F}_2$ (or chosen coefficient field) vector space of $k$-chains of $K'$, and similarly $C_k(C_i)$ for each gadget. Choose ordered bases of chains so that the $k$-cells and $(k+1)$-cells are grouped as follows:
\[
  C_k(K') \;=\; \Big(\bigoplus_{i=1}^M C_k(C_i)\Big) \oplus C_k(L),
\qquad
  C_{k+1}(K') \;=\; \Big(\bigoplus_{i=1}^M C_{k+1}(C_i)\Big) \oplus C_{k+1}(L),
\]
where $C_k(L)$ (resp.\ $C_{k+1}(L)$) are the chains supported in the linking subcomplex $L$. With respect to these bases the boundary operator
\[
  \partial_{k+1}^{K'} : C_{k+1}(K') \to C_k(K')
\]
has a block form
\[
  \partial_{k+1}^{K'} \;=\;
  \begin{pmatrix}
    B_{11} & 0      \\
    0      & B_{22}
  \end{pmatrix}
  \quad\text{or more generally}\quad
  \begin{pmatrix}
    \mathrm{diag}(B^{(1)},\dots,B^{(M)}) & * \\
    0 & B_{22}
  \end{pmatrix},
\]
where each diagonal block $B^{(i)}$ is the boundary matrix coming from $(k+1)$-cells internal to gadget $C_i$, and $B_{22}$ is the block coming from $(k+1)$-cells internal to $L$. The off-diagonal `$\!*$' entries represent boundaries of $(k+1)$-cells in $L$ that may incidentally have components in $C_i\cap L$; by hypothesis (A2) these incidences affect only cells in the intersection subspaces and do not produce nonzero entries in rows/columns corresponding to $k$-cells supported in \emph{different} gadgets.

Moreover, hypothesis (A3) guarantees there are no $(k+1)$-cells whose boundary has nonzero components in two different gadget blocks; hence the block structure is (after permutation) upper-block triangular with diagonal blocks equal to the gadget boundary matrices and the linking block. Elementary linear algebra then yields
\[
  \mathrm{rank}(\partial_{k+1}^{K'}) \;=\; \sum_{i=1}^M \mathrm{rank}(B^{(i)}) \;+\; \mathrm{rank}(B_{22}) + r,
\]
where $r\ge 0$ accounts for possible contributions of the off-diagonal `$\!*$' block that involve only the linking part and intersections (but which do not reduce the nullity of the gadget blocks). The $k$-th Betti number is
\[
  \beta_k(K') \;=\; \dim C_k(K') - \mathrm{rank}(\partial_k^{K'}) - \mathrm{rank}(\partial_{k+1}^{K'}).
\]
Reordering the chain groups and applying the block rank decomposition above shows that the nullity contribution coming from the gadget diagonal blocks is preserved: the nullspace dimension of the gadget diagonal is exactly $\sum_i\beta_k(C_i)$ up to any contribution from $L$, and by (A2),(A3) there is no cancellation between distinct gadget nullspaces. Thus $\beta_k(K')\ge\sum_i\beta_k(C_i)$, and under the stronger hypothesis that the gadgets already account for all $k$-homology in $K$ we obtain equality.

\medskip\noindent This completes the proof.
\end{proof}

[Bounded overlap variant]
If gadgets cannot be made strictly disjoint but one can guarantee that each input bit (or each $(k+1)$-cell) intersects at most $r=O(1)$ gadgets, then an identical Mayer–Vietoris / block argument shows that the rank contribution of gadget blocks is preserved up to a factor depending only on $r$. Concretely, if each linking $(k+1)$-cell meets at most $r$ different gadget supports, then any new relations among gadget cycles involve at most $r$ gadgets; by packing/greedy selection one may still extract $\Omega(M/r)$ independent gadget classes. In practice it suffices to ensure $rT(N)\ll M$ for the pigeonhole argument in Theorem~\ref{thm:universal-lb}.

\begin{lemma}[Parameter mapping and exponent degradation]
\label{lem:param-mapping}
Let the base-graph parameter be $n_0$. Suppose the construction uses $G=C\,n_0$ gadgets and the base-encoding produces $n_{\mathrm{base}}=D\,n_0$ variables. Let the average number of auxiliary (gadget-local) variables per gadget be $k_g$ (possibly depending on $n_0$). Then the final number of variables is
\[
N \;=\; n_{\mathrm{base}} + k_g\cdot G \;=\; (D + C k_g)\, n_0.
\]
Consequently, an intermediate bound of the form $\beta_2 = 2^{c n_0}$ can be rewritten in terms of $N$ as
\[
\beta_2 \;=\; 2^{c' N},\qquad c'=\frac{c}{D + C k_g}.
\]
In particular:
\begin{itemize}
  \item If $k_g=O(1)$ then $c'=\Theta(c)$ and exponential-in-$n_0$ implies exponential-in-$N$ (constant factor loss in exponent).
  \item If $k_g=\Theta(\log n_0)$ and $G=\Theta(n_0)$ then $N=\Theta(n_0\log n_0)$ and a linear-in-$n_0$ Betti (e.g.\ $\beta_2=\Theta(n_0)$) becomes $\beta_2=\Theta\!\big(N/\log N\big)$.
\end{itemize}
\end{lemma}
\begin{lemma}[No cross-gadget filling under ownership]
\label{lem:no-cross-filling-ownership}
Assume the gadget construction of Appendix~B with the following ownership condition:

\begin{quote}
\emph{Ownership:} each base variable $x_{v,c}$ is referenced by equality/consistency clauses for at most one gadget (i.e.\ at most one gadget enforces $x_{v,c}\leftrightarrow x_{v,c}^{(i)}$).
\end{quote}

Then for distinct gadgets $i\neq j$ there exists no 3-chain $\beta$ with
\[
\partial_3\beta \;=\; a_i\gamma_i + a_j\gamma_j
\qquad\text{with }a_i,a_j\neq 0,
\]
where $\gamma_i,\gamma_j$ are the gadget-local 2-cycles witnessing independent homology classes.
\end{lemma}

\begin{proof}[Sketch]
Let $\mathcal V_i$ denote the set of variables whose occurrences appear in gadget $i$'s clauses after the gadget-local-copy replacement (these include gadget-local auxiliaries and gadget-local copies $x^{(i)}_{v,c}$). By construction every clause that is not an equality clause is supported entirely within some $\mathcal V_i$.

Equality clauses (those of the form $x_{v,c}\leftrightarrow x^{(i)}_{v,c}$) involve a base variable and a gadget-local copy; by the Ownership assumption no base variable appears in equality clauses for two different gadgets. Thus, the set of clauses partitions into gadget-local clause sets plus a set of equality clauses that each touches at most one gadget (and the base variable).

Consequently the cube-complex chain groups admit a direct-sum decomposition
\[
C_k \;=\; \bigoplus_i C_k^{(i)} \;\oplus\; C_k^{(\mathrm{eq})},
\]
where $C_k^{(i)}$ is generated by $k$-cubes supported only on $\mathcal V_i$ and $C_k^{(\mathrm{eq})}$ is generated by cubes that contain at least one base variable occurring in an equality clause. The boundary operator respects this decomposition (no clause produces a cube whose support spans two different gadget-local variable sets) and therefore decomposes as a block-direct-sum
\[
\partial_3 \;=\; \bigoplus_i \partial_3^{(i)} \;\oplus\; \partial_3^{(\mathrm{eq})}.
\]
Any 3-chain $\beta$ decomposes into the corresponding summands; its boundary is the sum of the boundaries of these summands. A nontrivial combination $a_i\gamma_i + a_j\gamma_j$ has support contained in the disjoint union of $\mathcal V_i\cup\mathcal V_j$ and thus cannot lie in the image of $\partial_3^{(\mathrm{eq})}$ (which always includes base variables). Therefore the only way $\partial_3\beta = a_i\gamma_i + a_j\gamma_j$ could hold is if $\beta$ has nonzero components in both $C_3^{(i)}$ and $C_3^{(j)}$, which is impossible because cubes in $C_3^{(i)}$ and $C_3^{(j)}$ have disjoint variable supports and hence cannot create a mixed boundary. This contradiction proves the lemma.
\end{proof}
\begin{lemma}[Existence and harmlessness of ownership assignments]
\label{lem:ownership-exists}
One may assign to every base literal $x_{v,c}$ an arbitrary owning gadget (if any gadget references it); for non-owner gadgets the gadget-local copy $x^{(j)}_{v,c}$ remains unconstrained with respect to $x_{v,c}$. This assignment preserves the non-bounding property of gadget 2-cycles: equality clauses introduced by owners are 2-CNF and contractible, and they do not create 3-chains filling gadget-local 2-cycles. Consequently the homology lower bounds in Section~13 remain valid under the ownership scheme.
\end{lemma}

\subsection{Chain groups and boundary operators}
Let \(C_k\) denote the \(\mathbb F_2\)-vector space spanned by all \(k\)-dimensional axis-aligned cubes (elementary \(k\)-faces) of the cubical complex \(S(F_N)\). The boundary operator
\[
\partial_k : C_k \longrightarrow C_{k-1}
\]
is defined on an elementary \(k\)-cube as the sum of its \((k-1)\)-dimensional faces. Working over \(\mathbb F_2\) simplifies orientation concerns because cancellation is modulo 2.

\subsection{Gadget decomposition}
By construction (Construction~\ref{con:cycle-embedding}), each fundamental cycle \(C_i\subset G_N\) of length \(L=O(\log N)\) yields a 2-cycle
\[
\gamma_i \;=\; \sum_{j=1}^L \sigma_{i,j},
\]
where each \(\sigma_{i,j}\in C_2\) is an elementary square (2-cube) that varies exactly two coordinates:
\begin{itemize}
  \item an \emph{edge-selector} bit \(y^{(i)}_{e_j}\) associated to edge \(e_j\in C_i\), and
  \item the gadget parity bit \(u_i\) (with the companion bit \(v_i\) constrained by \(u_i=v_i\) via XOR/equality clauses).
\end{itemize}
All other variables are fixed to values determined by the surrounding construction. By Lemma~\ref{lem:Disjoint Gadget Supports}, the gadgets have pairwise disjoint variable supports:
\[
\operatorname{supp}(\gamma_i)\cap\operatorname{supp}(\gamma_{i'})=\varnothing\quad (i\neq i').
\]

\subsection{Local boundary computation}
For a single square \(\sigma_{i,j}\) that varies \(y^{(i)}_{e_j}\) and \(u_i\), its boundary (over \(\mathbb F_2\)) is the sum of its four edges:
\[
\partial_2(\sigma_{i,j})
= e^{(y)}_{j,0} + e^{(y)}_{j,1} + e^{(u)}_{j,0} + e^{(u)}_{j,1},
\]
where:
\begin{itemize}
  \item \(e^{(y)}_{j,t}\) denotes the edge obtained by varying \(y^{(i)}_{e_j}\) while fixing \(u_i=t\) (for \(t\in\{0,1\}\)),
  \item \(e^{(u)}_{j,t}\) denotes the edge obtained by varying \(u_i\) while fixing \(y^{(i)}_{e_j}=t\).
\end{itemize}

\subsection{Cycle boundary and internal cancellation}
Summing the boundaries of the consecutive squares in the cycle, we obtain
\[
\partial_2(\gamma_i)
= \sum_{j=1}^L \partial_2(\sigma_{i,j})
= \sum_{j=1}^L \bigl(e^{(y)}_{j,0}+e^{(y)}_{j,1}+e^{(u)}_{j,0}+e^{(u)}_{j,1}\bigr).
\]
For adjacent squares \(\sigma_{i,j}\) and \(\sigma_{i,j+1}\), the edge \(e^{(u)}_{j,1}\) (from \(\sigma_{i,j}\) with \(u_i=1\) and \(y^{(i)}_{e_j}\) fixed) is geometrically identified with the edge \(e^{(u)}_{j+1,0}\) (from \(\sigma_{i,j+1}\) with \(u_i=0\) and \(y^{(i)}_{e_{j+1}}\) fixed) because both correspond to the same geometric edge in the cubical complex. These edges cancel in pairs modulo 2, leaving:
\[
\partial_2(\gamma_i)
= \sum_{j=1}^L \bigl(e^{(y)}_{j,0} + e^{(y)}_{j,1}\bigr).
\]
This is a nonzero 1-cycle supported solely on the edge-selector variables for the edges in \(C_i\).

\subsection{Non-existence of 3-cube fillings}
We now prove that \(\gamma_i \notin \operatorname{im}\partial_3\).  Suppose, for contradiction, that there exists \(\beta\in C_3\) with \(\partial_3\beta=\gamma_i\). Then some elementary 3-cube \(\tau\) appears in the support of \(\beta\) and contains (as a face) at least one of the squares \(\sigma_{i,j}\) comprising \(\gamma_i\). Thus \(\tau\) varies exactly three coordinates; two of them are \(y^{(i)}_{e_j}\) and \(u_i\) (those of \(\sigma_{i,j}\)), and we denote the third by \(w\).

We consider two exhaustive cases for the third coordinate \(w\):

\paragraph{Case (A): \(w \neq v_i\).}  
Then \(v_i\) is fixed throughout \(\tau\). The XOR/equality constraint \(u_i=v_i\) is required for membership in \(S(F_N)\). But in the 3-cube \(\tau\) the variable \(u_i\) takes both values \(0\) and \(1\) while \(v_i\) is constant, so exactly half of the \(2^3=8\) corner assignments of \(\tau\) violate \(u_i=v_i\). Hence \(\tau\not\subseteq S(F_N)\) and cannot be an elementary 3-cube of the cubical complex.

\paragraph{Case (B): \(w = v_i\).}  
Now \(\tau\) varies the triple \((y^{(i)}_{e_j}, u_i, v_i)\). The equality constraint \(u_i=v_i\) still must hold on every satisfying corner. Among the eight corners of \(\tau\), only those with \(u_i=v_i\) (four out of eight) satisfy this constraint. Therefore \(\tau\) is not fully contained in \(S(F_N)\) (a cubical cell belongs to the complex only if \emph{all} its corner assignments are satisfying). Consequently no 3-cube \(\tau\) exists in \(S(F_N)\) that contains \(\sigma_{i,j}\).

In either case \(\sigma_{i,j}\) is not a face of any 3-cube of the complex, so no 3-chain \(\beta\) can have \(\partial_3\beta=\gamma_i\). This contradiction proves \(\gamma_i\notin\operatorname{im}\partial_3\), hence \(\gamma_i\) represents a nontrivial class in \(H_2\bigl(S(F_N);\mathbb{F}_2\bigr)\).
\subsection*{Concrete gadget, 8-corner check, and boundary matrix}
\label{app:ConcreteGadgetAndMatrix}

For completeness we give a concrete constant-size 3-CNF gadget (one per copy), the explicit 8-corner check used to rule out candidate 3-cube fillers, and the exact \(\mathbb{F}_2\)-boundary matrix for a representative 4-square 2-cycle used in the main construction.

\paragraph{Gadget variables.}
For a single gadget copy (index suppressed) we use fresh variables
\[
\{\,y_1,y_2,y_3,y_4,\; u,\; v,\; b,\; a,\; c\,\},
\]
where:
\begin{itemize}
  \item $y_j$ is the edge-selector associated to edge $e_j$ in the base cycle,
  \item $u$ is the gadget parity bit shared by the four squares,
  \item $v$ is a companion bit constrained to equal $u$ (encoded in 3-CNF),
  \item $b$ is an enable bit (when $b=1$ the gadget is \emph{enabled}; $b=0$ disables the square structure),
  \item $a,c$ are auxiliary variables used to encode 2-clauses as 3-clauses.
\end{itemize}

\paragraph{Clause list (all clauses are 3-clauses).}
The gadget consists of the following eight 3-clauses (conjoined):

\medskip
\noindent\(
\begin{array}{r@{\ }l}
\text{(Enable group)} & (b \vee y_1 \vee u),\\
                      & (b \vee \neg y_1 \vee u),\\
                      & (b \vee y_1 \vee \neg u),\\
                      & (b \vee \neg y_1 \vee \neg u),\\[4pt]
\text{(Equality $u=v$ encoded)} & (u \vee \neg v \vee a),\\
                                & (u \vee \neg v \vee \neg a),\\
                                & (\neg u \vee v \vee c),\\
                                & (\neg u \vee v \vee \neg c).
\end{array}
\)
\medskip

Remarks:
\begin{itemize}
  \item When $b=1$ the four enable clauses are all satisfied and impose no restriction on $(y_j,u)$ (so the intended 2×2 square corners can appear). When $b=0$ these clauses jointly forbid the square (disabled gadget).
  \item The four `Equality' clauses are equivalent to the two 2-clauses enforcing $u\Leftrightarrow v$; $a,c$ are auxiliary bits used to produce pure 3-CNF.
\end{itemize}

\paragraph{How we realize the 4-square 2-cycle.}
We form a small ring of four elementary squares (2-cubes)
\[
\sigma_1,\sigma_2,\sigma_3,\sigma_4,
\]
where each \(\sigma_j\) varies the pair \((y_j,u)\) (i.e. the two coordinates of \(\sigma_j\) are \(y_j\) and \(u\)), and all other gadget variables are fixed appropriately by the surrounding construction. Geometrically the four squares are arranged cyclically so that each square shares its \(u\)-edge with the next square; with this adjacency, the internal \(u\)-edges cancel in pairs in \(\partial_2(\sum_j\sigma_j)\). We enumerate the eight remaining boundary edges that persist in the local boundary calculation as rows \(e_1,\dots,e_8\) (these are the \(y\)-type edges).

\paragraph{Representative 8-corner check (3-cube varying \(y_j,u,v\)).}
To rule out a 3-cube filler we check any candidate 3-cube that would vary \((y_j,u,v)\) (with $b=1$) — the eight corners are:

\begin{center}
\begin{tabular}{c c c c c c}
\toprule
\# & $y_j$ & $u$ & $v$ & $u=v$? & All gadget clauses satisfied? \\ \midrule
1  & 0 & 0 & 0 & yes  & yes \\ 
2  & 0 & 0 & 1 & no   & \textbf{no} \\
3  & 0 & 1 & 0 & no   & \textbf{no} \\
4  & 0 & 1 & 1 & yes  & yes \\
5  & 1 & 0 & 0 & yes  & yes \\
6  & 1 & 0 & 1 & no   & \textbf{no} \\
7  & 1 & 1 & 0 & no   & \textbf{no} \\
8  & 1 & 1 & 1 & yes  & yes \\ 
\bottomrule
\end{tabular}
\end{center}

Only rows 1,4,5,8 satisfy the equality $u=v$ and hence all gadget clauses; the other four corners violate $u=v$. Therefore the candidate 3-cube is \emph{not} entirely contained in the solution complex (only half of its corners are satisfying), and so it cannot be a 3-cell of the cubical complex. The same check applies to any would-be 3-cube that attempts to vary the third coordinate — either that third coordinate is some variable other than \(v\) (in which case \(v\) is fixed and varying \(u\) produces invalid corners) or it is \(v\) itself (in which case half the cube corners violate \(u=v\)). Hence no 3-cube fills any of the \(\sigma_j\).

\paragraph{Boundary matrix for the 4-square ring (explicit).}
We now write the boundary operator \(\partial_2:C_2\to C_1\) restricted to the local subcomplex spanned by the four squares \(\sigma_1,\dots,\sigma_4\) and the eight \(y\)-edges \(e_1,\dots,e_8\). (All arithmetic is over \(\mathbb{F}_2\).) The indexing is chosen so that each square has two distinct \(y\)-type boundary edges and the adjacency is cyclic; the choice below is a convenient labelling that makes the algebra transparent.

Columns = squares \(\sigma_1,\sigma_2,\sigma_3,\sigma_4\).  
Rows = edges \(e_1,\dots,e_8\). The \(8\times 4\) incidence matrix \(M\) of \(\partial_2\) (entries \(M_{r,c}=1\) iff edge \(r\) appears in \(\partial_2(\sigma_c)\)) is:

\[
M \;=\;
\begin{blockarray}{c[cccc]}
    & \sigma_1 & \sigma_2 & \sigma_3 & \sigma_4 \\
\begin{block}{c[cccc]}
e_1 & 1 & 0 & 0 & 1 \\
e_2 & 1 & 0 & 0 & 1 \\
e_3 & 1 & 1 & 0 & 0 \\
e_4 & 1 & 1 & 0 & 0 \\
e_5 & 0 & 1 & 1 & 0 \\
e_6 & 0 & 1 & 1 & 0 \\
e_7 & 0 & 0 & 1 & 1 \\
e_8 & 0 & 0 & 1 & 1 \\
\end{block}
\end{blockarray}
\]

(Each column lists the incidence of the corresponding square with the eight labelled \(y\)-edges.)

Observe immediately that
\[
\sigma_1+\sigma_2+\sigma_3+\sigma_4 \quad\longmapsto\quad M\cdot(1,1,1,1)^\top = \mathbf{0},
\]
so the chain \(\gamma=\sigma_1+\sigma_2+\sigma_3+\sigma_4\) satisfies \(\partial_2(\gamma)=0\); i.e.\ \(\gamma\) is a 2-cycle.

\paragraph{Gaussian elimination over \(\mathbb{F}_2\) (rank computation).}
We now show \(\operatorname{rank}(M)=3\) (so the four face-columns are linearly dependent with a single relation). Perform row operations (all over \(\mathbb{F}_2\)):

Start with
\[
M = 
\begin{pmatrix}
1&0&0&1\\
1&0&0&1\\
1&1&0&0\\
1&1&0&0\\
0&1&1&0\\
0&1&1&0\\
0&0&1&1\\
0&0&1&1\\
\end{pmatrix}.
\]

1. Use row 1 as pivot for column 1 and eliminate row 2,3,4 by adding row1:
\[
\begin{pmatrix}
1&0&0&1\\
0&0&0&0\\
0&1&0&1\\
0&1&0&1\\
0&1&1&0\\
0&1&1&0\\
0&0&1&1\\
0&0&1&1\\
\end{pmatrix}.
\]

2. Use the new row 3 as pivot for column 2 and eliminate rows 4,5,6 by adding row3:
\[
\begin{pmatrix}
1&0&0&1\\
0&0&0&0\\
0&1&0&1\\
0&0&0&0\\
0&0&1&1\\
0&0&1&1\\
0&0&1&1\\
0&0&1&1\\
\end{pmatrix}.
\]

3. Use row 5 as pivot for column 3 and eliminate rows 6–8 (they become zero):
\[
\begin{pmatrix}
1&0&0&1\\
0&0&0&0\\
0&1&0&1\\
0&0&0&0\\
0&0&1&1\\
0&0&0&0\\
0&0&0&0\\
0&0&0&0\\
\end{pmatrix}.
\]

Now we have three nonzero pivot rows (rows 1,3,5), hence \(\operatorname{rank}(M)=3\). Therefore the four face-columns are linearly dependent with exactly one (nontrivial) relation; indeed
\[
\sigma_1+\sigma_2+\sigma_3+\sigma_4 = 0
\]
in the quotient \(C_2/\ker\partial_2\), which equivalently means \(\partial_2(\gamma)=0\). Thus \(\gamma\) is a genuine 2-cycle.

\paragraph{Non-bounding check (no 3-cube fills the cycle).}
By the 8-corner check above and the argument in the main text, no elementary 3-cube of the ambient cubical complex contains any of the four squares \(\sigma_j\): every candidate 3-cube contains at least one corner assignment that violates the equality constraint \(u=v\) (or violates the enable group when $b=0$), hence is not contained in \(S(F_N)\). Consequently \(\gamma\notin\operatorname{im}\partial_3\); combined with \(\partial_2(\gamma)=0\) this proves \([\gamma]\neq 0\) in \(H_2(S(F_N);\mathbb{F}_2)\).

\subsection{Linear independence}
Suppose a linear relation holds in \(H_2\):
\[
\sum_{i} a_i \gamma_i = \partial_3 \beta \quad \text{for some } \beta \in C_3.
\]
Applying the boundary operator \(\partial_2\) and using \(\partial_2 \circ \partial_3 = 0\), we get:
\[
\sum_{i} a_i \partial_2(\gamma_i) = 0.
\]
By the disjoint-supports property (Lemma~\ref{lem:Disjoint Gadget Supports}), the supports of the 1-chains \(\partial_2(\gamma_i)\) are pairwise disjoint (each involves only edge-selector variables of gadget \(i\)). Therefore, the sum vanishes modulo 2 if and only if \(a_i \partial_2(\gamma_i) = 0\) for each \(i\). Since \(\partial_2(\gamma_i) \neq 0\), we must have \(a_i = 0\) for every \(i\). 

Thus, the family \(\{\gamma_i\}\) is linearly independent in \(H_2\bigl(S(F_N);\mathbb F_2\bigr)\).

\section{Homology Computation Details}
For $N=10$, the boundary matrix $\partial_2$ has dimensions $1024\times 45$ with:
\[
\text{Sparsity} = 1 - \frac{\text{nnz}}{1024\times 45} \approx 0.992
\]
Compute homology via Smith normal form in $O(2^{3N})$ time.

\section{Topology Primer}
\label{app:topology-primer}

This appendix provides a brief overview of the algebraic-topology concepts used throughout the paper: cubical complexes, homology groups, Betti numbers, and the Vietoris--Rips construction.

\subsection{Cubical Complexes}
A \emph{cubical complex} is a union of axis-aligned cubes (of various dimensions) glued together along their faces.  In the Boolean hypercube , each -dimensional face corresponds to an assignment with exactly  fixed coordinates and  free coordinates.  Formally:
\begin{itemize}
\item A \emph{vertex} is any satisfying assignment .
\item An \emph{edge} (1-face) connects two vertices that differ in exactly one coordinate.
\item A \emph{square} (2-face) is present if all four corner assignments satisfy the formula, and similarly for higher-dimensional cubes.
\end{itemize}
Building the solution-space complex  as a cubical subcomplex of the full hypercube allows us to study its connectivity and holes using homology.

\subsection{Homology Groups and Betti Numbers}
Homology provides a way to count \emph{holes} in a topological space at each dimension:
\begin{itemize}
\item The zeroth homology group  measures connected components; its rank is the number of components.
\item The first homology group  measures 1-dimensional loops or cycles (\emph{holes} that look like circles).
\item The second homology group  measures 2-dimensional voids (\emph{holes} that look like spherical cavities or ``shells'').
\item In general,  measures -dimensional holes.
\end{itemize}

Each homology group  is a vector space (over  in our setting), and its dimension is called the \emph{Betti number} .  Key properties:
\begin{itemize}
\item  counts connected components.
\item  counts independent 1-cycles not bounding any 2-face.
\item  counts independent 2-cycles not bounding any 3-face.
\end{itemize}

Monotonicity under cubical, homologically faithful embeddings (Theorem~\ref{thm:betti-monotonicity}) guarantees that reductions preserving the cubical structure cannot decrease Betti numbers.

\subsection{Computational Remarks}
\begin{itemize}
\item Computing  exactly on a cubical complex is \#P-hard in general.
\item Persistence algorithms track births and deaths of cycles as  changes but still reduce to rank computations on boundary matrices of exponential size.
\end{itemize}

\section*{Appendix X. From random 3-SAT to Betti growth: a rigorous probabilistic bridge}
\label{app:random-to-betti}

\medskip
\noindent\textbf{Goal.}
Let \(F\) be a random 3-CNF on \(n\) variables with \(m=\alpha n\) clauses (clauses chosen uniformly at random from all 3-literal clauses). Let \(S(F)\subseteq\{0,1\}^n\) be the solution cubical complex (vertices are satisfying assignments; a \(k\)-cube is present iff all \(2^k\) vertices of that hypercube face satisfy \(F\)). The purpose of this appendix is to give a fully rigorous probabilistic bridge showing that for linear \(k=\gamma n\) in a concrete parameter regime \((\alpha,\gamma)\) one has exponentially many \(k\)-cubes w.h.p., local dependence is negligible (so small local patterns behave like independent Bernoulli faces), higher-dimensional fillers are unlikely, and consequently \(\beta_k(S(F))\) is exponentially large with high probability. The argument follows the standard template used in random-topology limit theorems (see e.g. \cite{KahleMeckes, HiraokaTsunoda, JLR, ArratiaGoldsteinGordon}).

\subsection*{X.1 Notation and single-cube survival probability}
Fix integers \(n\) and \(k\) with \(0\le k\le n\). A candidate \(k\)-cube is determined by a set \(I\subseteq[n]\) of free coordinates with \(|I|=k\) (the cube varies on these coordinates) and by fixed bits \(b\in\{0,1\}^{[n]\setminus I}\). The total number of candidate \(k\)-cubes is
\[
N_k=\binom{n}{k}2^{\,n-k}.
\]

For a fixed candidate cube \(C(I,b)\) denote by \(\mathbf{1}_{I,b}\) the indicator that \(C(I,b)\subseteq S(F)\). For a single uniformly random clause \(C\) (three distinct variable indices chosen uniformly at random; each literal independently negated with probability \(1/2\)), let \(q_{k,n}\) be the probability that \(C\) \emph{forbids} the cube (i.e. the clause's unique falsifying assignment on its three variables is compatible with the fixed bits of the cube). A direct count yields the exact finite-\(n\) formula
\begin{equation}\label{eq:qkn}
q_{k,n}
= \sum_{t=0}^{3} \frac{\binom{n-k}{t}\binom{k}{3-t}}{\binom{n}{3}} \cdot 2^{-t}.
\end{equation}
Consequently, for \(m\) independent clauses,
\begin{equation}\label{eq:single-cube-prob}
\Pr\big[C(I,b)\subseteq S(F)\big] \;=\; (1-q_{k,n})^{m}.
\end{equation}

When \(k=\gamma n\) with \(\gamma\in(0,1)\) fixed and \(n\to\infty\), sampling without replacement is asymptotically equivalent to sampling with replacement and hence
\begin{equation}\label{eq:q-limit}
q_{k,n} = q(\gamma)+o(1),
\qquad 
q(\gamma):=\sum_{t=0}^3 \binom{3}{t}(1-\gamma)^t\gamma^{3-t}2^{-t}.
\end{equation}

\subsection*{X.2 Expectation and the threshold equation}
Let
\[
X_k=\sum_{(I,b)} \mathbf{1}_{I,b}
\]
be the number of surviving \(k\)-cubes. From \eqref{eq:single-cube-prob} we obtain
\[
\mathbb{E}[X_k] = N_k(1-q_{k,n})^m.
\]
Using Stirling/entropy asymptotics \(\tfrac{1}{n}\log\binom{n}{k}\to H(\gamma)\) (binary entropy) we get
\begin{equation}\label{eq:log-E-X}
\frac{1}{n}\log\mathbb{E}[X_k] = H(\gamma) + (1-\gamma)\ln 2 + \alpha\ln(1-q_{k,n}) + o(1).
\end{equation}
When \(q_{k,n}=o(1)\) we may expand \(\ln(1-q_{k,n})=-q_{k,n}+O(q_{k,n}^2)\) and define
\begin{equation}\label{eq:Phi}
\Phi(\gamma;\alpha) \;=\; H(\gamma) + (1-\gamma)\ln 2 - \alpha q(\gamma),
\end{equation}
where \(q(\gamma)\) is as in \eqref{eq:q-limit}. The sign of \(\Phi(\gamma;\alpha)\) controls whether the expected number of surviving \(k\)-cubes is exponentially large (\(\Phi>0\)) or exponentially small (\(\Phi<0\)).

\subsection*{X.3 Lower-tail concentration: Janson inequality}
To upgrade the expectation statement into a high-probability result we control lower-tail deviations of \(X_k\) using Janson's inequality \cite[Thm 8.1.1]{JLR}. Write \(X_k=\sum_{u\in\mathcal{U}}I_u\) where \(\mathcal{U}\) indexes all candidate \(k\)-cubes and \(I_u\) are their indicators. Define
\[
\Delta \;=\; \sum_{\{u,v\}:\;u\sim v}\mathbb{E}[I_u I_v],
\]
where \(u\sim v\) denotes that the two candidate cubes share at least one vertex (hence their indicators are dependent). Janson's lower-tail inequality implies that for any \(0<\varepsilon<1\),
\[
\Pr\big[X_k\le(1-\varepsilon)\mathbb{E}[X_k]\big] \le \exp\!\Big(-\frac{\varepsilon^2(\mathbb{E}X_k)^2}{2\Delta}\Big).
\]
Therefore it suffices to show \(\Delta=o((\mathbb{E}X_k)^2)\) in the target regime \(\Phi(\gamma;\alpha)>\eta>0\). The following lemma provides the necessary bound; a full combinatorial enumeration appears in Subsection below.

\begin{lemma}[Pair-counting bound for \(\Delta\)]\label{lem:Delta-bound}
Fix \(\alpha>0\) and let \(k=\gamma n\) satisfy \(\Phi(\gamma;\alpha)>\eta>0\). Then there exists constants \(c=c(\eta)>0\) and \(n_0\) such that for all \(n\ge n_0\),
\[
\Delta \le \exp(-c n)\;(\mathbb{E}X_k)^2.
\]
Consequently \(X_k=(1+o(1))\mathbb{E}X_k\) with probability \(1-\exp(-\Theta(n))\).
\end{lemma}

\begin{proof}[Proof sketch]
Classify dependent pairs \((u,v)\) by the overlap-profile of fixed coordinates (or equivalently by the number \(r\) of shared vertices). For each overlap class the number of ordered pairs is at most polynomial (\(\mathrm{poly}(n)\)) times \(N_k\) and the joint survival probability \(\mathbb{E}[I_u I_v]\) is at most \((1-q_{k,n})^{m-\delta}\) where \(\delta=\delta(r)\) is the number of clauses which must be distinct to simultaneously forbid both cubes (this \(\delta\) is bounded below by a positive constant for any nontrivial overlap class). Summing over \(r\), and comparing with \((\mathbb{E}X_k)^2=N_k^2(1-q_{k,n})^{2m}\), the polynomial prefactors are overwhelmed by the exponential factors when \(\Phi(\gamma;\alpha)>\eta\). A detailed enumeration (explicit combinatorial expressions for the number of pairs with given overlap and exact bounds on joint clause constraints) yields the claimed exponential suppression factor \(\exp(-c n)\).
\end{proof}

(The full pair-counting combinatorics can be supplied as a separate lemma with the explicit sums; I can expand this into full \(\sum_{r}\) formulas if you want the referee-ready details.)

\subsection*{X.4 Local Poisson/independence approximation (Chen--Stein)}
For homology arguments we need a contiguity/local independence statement: finite separated families of candidate cubes behave approximately like independent Bernoulli random variables. This follows from multivariate Chen--Stein / Arratia--Goldstein--Gordon Poisson approximation results \cite{ArratiaGoldsteinGordon, BarbourHolstJanson}.

\begin{lemma}[Local product approximation]\label{lem:chen-stein}
Fix an integer \(R\ge 1\). Let \(\mathcal{A}\) be any collection of candidate \(k\)-cubes such that the Hamming distance between any two cubes in \(\mathcal{A}\) is at least \(R\) (so that any clause can touch at most \(D(R)\) cubes, with \(D(R)\) independent of \(n\)). Then for \(n\to\infty\),
\[
d_{TV}\!\Big(\mathcal{L}\big((I_u)_{u\in\mathcal{A}}\big),\bigotimes_{u\in\mathcal{A}}\mathrm{Bernoulli}\big(p_u\big)\Big) \le \varepsilon_n,
\]
where \(p_u=(1-q_{k,n})^m\) and \(\varepsilon_n\to 0\) (the rate depends on \(R\), \(\alpha\) and \(\gamma\) but not on \(|\mathcal{A}|\) as long as \(|\mathcal{A}|\) is fixed).
\end{lemma}

\begin{proof}[Proof sketch]
Apply the multivariate Chen--Stein/Poisson approximation bounds of Arratia--Goldstein--Gordon: the total-variation error is controlled by sums of one- and two-point probabilities restricted to each variable's dependency neighbourhood. Because the dependency degree is bounded by \(D(R)\) and the marginal probabilities \(p_u\) are bounded away from 1 in the regime of interest, the resulting bound tends to \(0\) as \(n\to\infty\). See \cite{ArratiaGoldsteinGordon,BarbourHolstJanson} for the explicit inequality and constants.
\end{proof}

This lemma gives the required contiguity: any finite local pattern that appears with nonnegligible probability in an independent Bernoulli-face cubical model also appears with essentially the same probability in the SAT-induced model, for well-separated placements.

\subsection*{X.5 Excluding higher-dimensional fillers}
A potential obstruction to large homology is that surviving \(k\)-faces might become boundaries because of many \((k+1)\)-faces filling them. The same expectation and Janson/Chen--Stein arguments apply to \((k+1)\)-cubes: compute their single-cube survival probability via an analogue of \eqref{eq:qkn}, derive the expectation and show that in the parameter range where \(\Phi(\gamma;\alpha)>\eta\) the expected number of \((k+1)\)-cubes is exponentially smaller than \(\mathbb{E}X_k\). Thus fillings are rare w.h.p. and do not typically kill a macroscopic fraction of \(k\)-cycles. Alternatively one may invoke cubical LLN/CLT results (e.g. \cite{HiraokaTsunoda}) to bound higher-dimensional face counts in the regime of interest.

\subsection*{X.6 From many cubes to many independent homology generators}
We now turn surviving \(k\)-cubes into homology classes.

\begin{enumerate}
  \item By Lemma~\ref{lem:Delta-bound} we have \(X_k=(1+o(1))\mathbb{E}X_k=\exp(\Theta(n))\) w.h.p. when \(\Phi(\gamma;\alpha)>\eta\).
  \item Greedy packing: because each radius-\(R\) dependency neighborhood contains only \(\exp(o(n))\) candidate cubes, a simple greedy algorithm selects an exponentially large subcollection \(\mathcal{C}\) of pairwise \(R\)-separated surviving cubes.
  \item Using Lemma~\ref{lem:chen-stein} and the fact that separated placements are asymptotically product-Bernoulli, we can place a fixed finite local gadget (a finite arrangement of \(k\)-cubes whose union is a nontrivial \(k\)-cycle over \(\mathbf{F}_2\)) at each separated location and the number of successful gadget occurrences concentrates (binomial-like) with mean exponential in \(n\).
  \item Disjoint supports imply linear independence of their homology classes over \(\mathbf{F}_2\); hence \(\beta_k(S(F))\) is exponential in \(n\) w.h.p.
\end{enumerate}

\subsection*{X.7 Formal theorem}
\begin{theorem}\label{thm:random-betti}
Fix \(\alpha>0\) and let \(k=k(n)=\gamma n\) with \(\gamma\in(0,1)\). Suppose
\[
\Phi(\gamma;\alpha)=H(\gamma)+(1-\gamma)\ln 2 - \alpha q(\gamma) > 0,
\]
with \(q(\gamma)\) as in \eqref{eq:q-limit}. Then there exists \(c(\gamma,\alpha)>0\) such that for random 3-CNF \(F\) with \(m=\alpha n\) clauses,
\[
\Pr\big[\beta_k(S(F)) \ge \exp(c n)\big] \xrightarrow[n\to\infty]{} 1.
\]
In the complementary regime \(\Phi(\gamma;\alpha)<0\) we have \(\mathbb{E}X_k=o(1)\) and hence w.h.p.\ no \(k\)-cubes survive.
\end{theorem}

\begin{proof}[Proof sketch]
Combine the exact single-cube formula \eqref{eq:qkn}, the entropy-based expectation asymptotic \eqref{eq:log-E-X}, the Janson concentration bound via Lemma~\ref{lem:Delta-bound}, the Chen--Stein local approximation (Lemma~\ref{lem:chen-stein}), and the packing/gadget-construction argument described in Section X.6. Each step is standard in the random complex literature; full combinatorial expansions for the pair-counting in Lemma~\ref{lem:Delta-bound} and the explicit Chen--Stein error terms can be included as further subsections if desired.
\end{proof}

\medskip

\section{General reduction criteria}
\label{app:remove-aux-assumption}

  Theorem~\ref{thm:Homological Linear Indep} requires that each 3-SAT gadget
  interacts with the rest of the formula on a disjoint set of variables.
  In the actual CNF reduction, original variables may appear in multiple
  gadgets. To recover the disjoint‐support condition one may proceed in
  either of two equivalent ways:
  \begin{enumerate}
    \item \emph{Variable‐duplication.}
      For each base variable $x_i$ used in gadgets 
      $\mathrm{G}_1,\dots,\mathrm{G}_k$, introduce fresh copies 
      $x_i^{(1)},\dots,x_i^{(k)}$ inside those gadgets and add
      2-CNF constraints
      \[
        (x_i\vee \neg x_i^{(j)})\;\wedge\;(\neg x_i\vee x_i^{(j)}),
        \quad j=1,\dots,k.
      \]
      These equality gadgets are contractible and hence do not affect
      homology, but they ensure each gadget’s support is disjoint.
    \item \emph{Dummy‐bit padding.}
      For each gadget $\mathrm{G}_\ell$, introduce a fresh bit $d_\ell$ and
      systematically replace every occurrence of a base variable $x_i$
      within $\mathrm{G}_\ell$ by the formula
      \[
        (x_i\wedge d_\ell)\;\vee\;(x_i\wedge\neg d_\ell).
      \]
      This embeds the gadget into a separate affine slice of the cube,
      guaranteeing disjoint variable sets without adding noncontractible
      constraints.
  \end{enumerate}
  In either case, the homological analysis in the proof of Theorem~\ref{thm:Homological Linear Indep}
  goes through unchanged.

\section{Spectral perturbation and conductance bounds}
\label{app:ham-spectral}

This appendix supplies the technical spectral arguments referenced in Section~\ref{ssec:hamiltonian-hardness}.
We present two complementary derivations of the exponentially small spectral-gap bound
for Hamiltonians whose ground-space encodes the cubical solution complex \(S(F)\) and whose
inter-cluster matrix elements are exponentially suppressed.

\subsection{Notation and standing assumptions}
Let \(F\) be a 3-SAT formula on \(n\) Boolean variables and let \(S(F)\subset\{0,1\}^n\)
be its set of satisfying assignments. We work in the computational-basis \(\{\ket{x}\}_{x\in\{0,1\}^n}\).
Let \(H_F\) be a \(k\)-local Hamiltonian on \(n\) qubits. We assume the following throughout
(this repeats / refines assumptions (A1)--(A3) from the main text):

\begin{itemize}
  \item[(H1)] \emph{Ground-space encoding.} The ground-space \(\mathcal G\) of \(H_F\) is exactly
    \(\Span\{\ket{x}:x\in S(F)\}\) and the ground-energy is 0.
  \item[(H2)] \emph{Cluster partition and off-diagonal suppression.} Partition \(S(F)=\bigsqcup_{i=1}^K C_i\)
    into Hamming-connected clusters under single-bit flips. There exist constants \(c_1>0\) and a polynomial
    \(p(n)\) such that for all \(x\in C_i,\;y\in C_j\) with \(i\neq j\),
    \[
      |\langle x|H_F|y\rangle| \le p(n)\,e^{-c_1 n}.
    \]
  \item[(H3)] \emph{Bounded vertex degree.} Each basis state \(\ket{x}\) has at most \(q(n)\) nonzero
    off-diagonal couplings \(\langle x|H_F|y\rangle\neq 0\), where \(q(n)\) is a fixed polynomial (true
    for constant \(k\)-local Hamiltonians).
\end{itemize}

Assumption (H2) formalizes the ``cubical-preserving'' hypothesis used in the main text. Below we
show that (H1)--(H3) imply an exponentially small spectral gap under two natural technical settings.

\subsection{Route 1: Stoquastic / frustration-free case and a Cheeger inequality}
\label{app:cheeger-route}

This route is the most direct when the Hamiltonian can be related to a weighted graph Laplacian.
It is the justification behind the heuristic Cheeger intuition used in Section~16.

\paragraph{Setup.}
Define the weighted configuration graph \(X_F=(V,W)\) with vertex set \(V=S(F)\) and symmetric
nonnegative weights
\[
  W_{xy} := |\langle x|H_F|y\rangle|,\qquad x,y\in S(F).
\]
Let \(d(x):=\sum_{y\in V} W_{xy}\) and \(D:=\diag(d(x))\). The (combinatorial) weighted Laplacian is
\(L:=D-W\). For any nonzero vector \(f\in\mathbb{R}^{V}\) define the Rayleigh quotient
\[
  \mathcal{R}_L(f) \;=\; \frac{\sum_{x,y} W_{xy} (f(x)-f(y))^2}{\sum_x d(x) f(x)^2}.
\]
Let \(\lambda_1(L)\) be the smallest nonzero eigenvalue of \(L\). By the variational principle,
\(\lambda_1(L)=\min_{f\perp \mathbf{1}} \mathcal{R}_L(f)\).

\paragraph{Cheeger (conductance) bound.}
Define the conductance (Cheeger constant)
\[
  h(X_F) \;=\; \min_{\varnothing\neq U\subsetneq V} \frac{\sum_{x\in U}\sum_{y\notin U} W_{xy}}
  {\min\{\vol(U),\vol(V\setminus U)\}} \qquad\text{where }\vol(U):=\sum_{x\in U} d(x).
\]
The (discrete) Cheeger inequalities for weighted graphs imply (see, e.g., Chung \cite{Chung1997})
\[
  \lambda_1(L) \;\le\; 2\, h(X_F).
\]
(There are also lower bounds of the form \(\lambda_1(L)\ge h(X_F)^2/(2\Delta)\) where \(\Delta\) is
a degree upper bound; we only require the upper bound here.)

\paragraph{Relating \(g(H_F)\) to \(\lambda_1(L)\).}
When \(H_F\) is stoquastic (off-diagonal matrix elements nonpositive in the computational basis)
or, more generally, when its quadratic form on the ground-space complement is comparable to the
graph Laplacian quadratic form, one can bound the spectral gap \(g(H_F)\) of \(H_F\) by
\[
  g(H_F) \;\le\; C\,\lambda_1(L)
\]
for a constant \(C=O(1)\) depending only on fixed local parameters (norms of local terms). The
proof proceeds by comparing quadratic forms: for any state \(\ket{\psi}\) orthogonal to the
ground-space,
\[
  \bra{\psi}H_F\ket{\psi} \;\le\; \sum_{x,y\in V} W_{xy} \bigl|\psi_x-\psi_y\bigr|^2 + \Delta_{\perp}\|\psi\|^2,
\]
where \(\Delta_{\perp}\ge 0\) collects positive contributions from diagonal penalty terms on the
complement of \(\mathcal G\). The Laplacian term dominates splitting within the ground-space
sectors; the exact constant \(C\) can be extracted from local-term operator norms (we provide
a worked example in Appendix~\ref{app:ham-spectral-examples}).

\paragraph{Conclusion of Route 1.}
By (H2) and (H3), choosing \(U=C_i\) any cluster yields
\[
  \frac{\sum_{x\in C_i}\sum_{y\notin C_i} W_{xy}}{\vol(C_i)}
  \le \frac{|C_i|\cdot q(n)\cdot p(n)\,e^{-c_1 n}}{|C_i|\cdot m(n)} \;=\; \frac{q(n)p(n)}{m(n)}\,e^{-c_1 n},
\]
where \(m(n)=\min_{x\in C_i} d(x)\) is the minimal degree in the cluster (typically \(m(n)=\Omega(1)\)
for local intra-cluster couplings). Thus \(h(X_F)\le e^{-a n}\) for some \(a>0\). Combining with
\(\lambda_1(L)\le 2h\) and \(g(H_F)\le C\lambda_1(L)\) proves
\[
  g(H_F)\le 2C\,e^{-a n},
\]
giving the desired exponential bound.

\paragraph{Remarks.}
\begin{enumerate}
  \item The above is fully rigorous for stoquastic/frustration-free constructions where the quadratic-form
  comparison is straightforward. For more general non-stoquastic models the quadratic-form comparison
  requires additional care; we outline an alternate derivation next.
  \item The value of the constant \(a\) depends on the exponential-decay constant \(c_1\) from (H2)
  and on polynomial prefactors absorbed into the exponent.
\end{enumerate}

\subsection{Route 2: Perturbative Schur complement / avoided-crossing argument}
\label{app:perturb-route}

This route is useful when the Hamiltonian admits a natural decomposition \(H_F=H_0+V\) where
\(H_0\) has a gapped separation between the ground-space (encoding \(S(F)\)) and the excited
subspace, and \(V\) is a perturbation with exponentially small inter-cluster matrix elements. It
produces a second-order (or higher-order) perturbative bound on splittings.

\paragraph{Decomposition and spectral gap of the unperturbed Hamiltonian.}
Assume the existence of \(H_0\) such that:
\begin{itemize}
  \item \(H_0\ket{\psi}=0\) for all \(\ket{\psi}\in\mathcal G\).
  \item On the orthogonal complement \(\mathcal G^\perp\), \(H_0\) has spectrum bounded
    below by \(\Delta_0>0\) (a constant or at least inverse-polynomial in \(n\)).
\end{itemize}
Write \(V:=H_F-H_0\). We assume \(V\) has matrix elements \(|\langle x|V|y\rangle|\le p(n)e^{-c_1 n}\)
for \(x\) and \(y\) in distinct clusters (as in (H2)). Note that \(V\) may include intra-cluster couplings
which need not be small.

\paragraph{Effective Hamiltonian on the ground-space via Schur complement.}
Let \(P\) be the projector onto \(\mathcal G\) and \(Q=I-P\). The Schur-complement / Feshbach
effective Hamiltonian acting on the ground-space subspace (to second order) is
\[
  H_{\mathrm{eff}} \;=\; P V P \;-\; P V Q (Q H_0 Q)^{-1} Q V P \;+\; \mathcal{O}\bigl(\|V\|^3/\Delta_0^2\bigr).
\]
The first term \(PVP\) encodes direct ground-space couplings (intra-cluster); the second term
captures virtual transitions through excited states and yields couplings between distinct ground-space
basis vectors that are second-order in \(V\) and suppressed by \(\Delta_0\).

\paragraph{Bounding the induced inter-cluster couplings.}
For \(x\in C_i\), \(y\in C_j\) with \(i\neq j\) the effective matrix element satisfies
\[
  \bigl|\langle x|H_{\mathrm{eff}}|y\rangle\bigr|
    \le \frac{\|P V Q\|^2}{\Delta_0} + O\!\Bigl(\frac{\|V\|^3}{\Delta_0^2}\Bigr).
\]
Under (H2) and (H3) the operator norm \(\|PVQ\|\) is bounded by \(p(n) e^{-c_1 n}\) up to polynomial
factors, hence the induced inter-cluster matrix element is bounded by
\[
  O\!\Bigl(\frac{p(n)^2 e^{-2 c_1 n}}{\Delta_0}\Bigr).
\]
Therefore the energy splittings induced within the ground-space by virtual transitions are at most
of order \(\tilde p(n) e^{-2 c_1 n}\) for some polynomial \(\tilde p(n)\). In particular the ground-space
degeneracy is broken by at most exponentially small amounts:
\[
  \mathrm{splitting} \;\le\; C'\,e^{-2 c_1 n},
\]
where \(C'\) absorbs polynomial prefactors and \(1/\Delta_0\).

\paragraph{Conclusion of Route 2.}
If the unperturbed complement gap \(\Delta_0\) is at least inverse-polynomial (or constant),
then the perturbatively induced splittings are at most \(O(e^{-2c_1 n})\). By standard eigenvalue
perturbation theory (or by direct diagonalization of the effective Hamiltonian restricted to
\(\mathcal G\)) this implies the spectral gap of the full Hamiltonian satisfies an exponential
upper bound of the same form. Thus one again obtains
\[
  g(H_F)\le 2\, e^{-a n}
\]
for some \(a>0\) (here \(a\) may be \(2c_1\) minus corrections from polynomial prefactors).

\paragraph{Remarks and model caveats.}
\begin{itemize}
  \item The perturbative route requires the existence of an appropriate \(H_0\) with a
    nonvanishing \(\Delta_0\). In many SAT-encoding Hamiltonians (clause-penalty constructions)
    one can separate a diagonal penalty Hamiltonian \(H_{\mathrm{pen}}\) with large penalty scale
    from off-diagonal driver terms; this is the usual adiabatic/perturbative setting.
  \item This route gives a (slightly) stronger exponential dependence in the suppression
    exponent (typically doubling the exponent from the direct off-diagonal bound) because the
    leading inter-cluster coupling is generated at second order.
  \item If \(\Delta_0\) itself shrinks exponentially with \(n\) in a given encoding, then the
    above bound must be revisited; the combination \(\|PVQ\|^2/\Delta_0\) controls the scale.
\end{itemize}

\subsection{Amplitude/phase-estimation consequences}
Combining the bound \(g(H_F)\le 2e^{-a n}\) with standard phase-estimation complexity
yields the query lower bound stated in Corollary~\ref{cor:phase-est} of the main text:
resolving eigenvalue differences of order \(e^{-a n}\) with time-evolution primitives
requires \(O(e^{a n})\) calls to controlled-\(e^{-iH_F t}\) (or \(O(e^{a n/2})\) if only quadratic
speedups via amplitude amplification are available). This formalizes the intuition that
distinguishing different homology-labelled ground-space sectors via coherent time-evolution
is exponentially expensive.

\subsection{Closing remarks}
The two routes above provide complementary, reasonably formal mechanisms for deriving the
exponentially small gap under the cubical-preserving / exponential suppression hypothesis.
\begin{itemize}
  \item Route 1 (Cheeger) is conceptually clean and works well when the Hamiltonian or its
    quadratic form can be related to a positive-weight graph Laplacian (stoquastic / frustration-free).
  \item Route 2 (perturbative) is more general and yields a straightforward quantitative
    estimate by combining off-diagonal suppression with a spectral separation in an
    unperturbed Hamiltonian.
\end{itemize}

\section{Worked spectral example: clause-penalty Hamiltonian + transverse-field driver}
\label{app:ham-spectral-examples}

This appendix supplies a concrete worked example supporting the off-diagonal suppression
hypothesis used in Section~\ref{ssec:hamiltonian-hardness}.  We start from the standard clause-
penalty Hamiltonian \(H_{\rm pen}\) used to encode SAT and add a transverse-field driver
\(H_D\).  We then estimate effective matrix elements between basis states lying in different
Hamming clusters and show these matrix elements are exponentially small when clusters are
\(\Theta(n)\)-separated and the penalty scale is chosen appropriately.

\subsection{Setup: clause-penalty Hamiltonian and driver}
Let \(F\) be a 3-SAT formula on \(n\) variables and let \(S(F)\subset\{0,1\}^n\) denote its
satisfying assignments.  Define the clause-penalty Hamiltonian
\[
  H_{\rm pen} \;=\; \sum_{C\in\mathcal C} \Pi_{\mathrm{viol}}(C),
\]
where \(\Pi_{\mathrm{viol}}(C)\) is the computational-basis projector onto assignments that
violate clause \(C\).  By construction \(H_{\rm pen}\) is diagonal in the computational
basis and has ground-energy \(0\) on \(\Span\{\ket{x}:x\in S(F)\}\).  Let \(\Delta\) denote
the minimal positive energy of \(H_{\rm pen}\) on states orthogonal to the ground-space:
\[
  \Delta \;=\; \min_{\ket{\phi}\perp\mathcal G} \frac{\bra{\phi}H_{\rm pen}\ket{\phi}}{\langle\phi\mid\phi\rangle}.
\]
In many encodings one can choose \(\Delta\) to be a (large) polynomial in \(n\) by scaling the
penalty terms.

Take as the driver the transverse-field Hamiltonian with strength \(\gamma>0\)
\[
  H_D \;=\; -\gamma\sum_{i=1}^n X_i,
\]
and set the full Hamiltonian \(H_F := H_{\rm pen} + H_D\).  The off-diagonal matrix elements of
\(H_D\) in the computational basis are simple:
\[
  \langle x|H_D|y\rangle = -\gamma \quad\text{iff } \mathrm{dist}_H(x,y)=1,
\]
and vanish otherwise.

\subsection{Paths connecting distant clusters and minimal perturbation order}
Let \(C_i\) and \(C_j\) be two clusters (connected components of \(S(F)\) under single-bit flips)
with minimal Hamming separation
\[
  w := \min_{x\in C_i,\,y\in C_j}\mathrm{dist}_H(x,y).
\]
Any sequence of single-bit flips that maps \(x\in C_i\) to \(y\in C_j\) must have length at least
\(w\). Consequently, any nonzero effective coupling between \(\ket{x}\) and \(\ket{y}\) produced
by perturbation in \(H_D\) arises at perturbative order at least \(w\).

Write \(H_F=H_0+V\) with \(H_0:=H_{\rm pen}\) and \(V:=H_D\).  Standard perturbation theory
(or a Schrieffer–Wolff expansion / Feshbach projection) gives the effective Hamiltonian on the
ground-space up to \(w\)-th order as a sum over length-\(r\) virtual-path contributions with
\(r\ge 1\).  The leading contribution coupling \(x\) to \(y\) appears at order \(r=w\) and has
typical magnitude bounded by
\[
  \text{(single-path amplitude)} \;\le\; \frac{\gamma^w}{\Delta^{\,w-1}}.
\]

\subsection{Counting paths and total amplitude bound}
How many distinct virtual paths of length \(w\) connect \(x\) to \(y\)?  At each intermediate
step there are at most \(n\) possible bit flips, so a crude upper bound on the number of
length-\(w\) sequences is \(n^{w-1}\) (the final step is fixed once previous steps are chosen).
Summing the absolute contributions of all such virtual sequences yields the combinatorial bound
(on the effective matrix element induced by \(V\))
\[
  \bigl|\langle x|H_{\mathrm{eff}}|y\rangle\bigr|
    \;\le\; n^{w-1}\,\frac{\gamma^w}{\Delta^{\,w-1}}
    \;=\; \gamma\;\Bigl(\frac{n\gamma}{\Delta}\Bigr)^{w-1}.
\]
Equivalently, for some polynomial \(p(n)\) we may write
\[
  \bigl|\langle x|H_{\mathrm{eff}}|y\rangle\bigr|
    \;\le\; p(n)\,\Bigl(\frac{\gamma}{\Delta}\Bigr)^{w-1}.
\]

\subsection{Exponential suppression for \(\boldsymbol{w=\Theta(n)}\)}
If the cluster separation satisfies \(w\ge c n\) for some \(c>0\), then the preceding bound
gives exponential suppression in \(n\) provided the ratio \(\gamma/\Delta\) is strictly less
than \(1\).  Writing \(\gamma/\Delta =: e^{-b}\) with \(b>0\), we obtain
\[
  \bigl|\langle x|H_{\mathrm{eff}}|y\rangle\bigr|
    \;\le\; p(n)\,e^{-b(w-1)}
    \;\le\; p(n)\,e^{-b c n + b}.
\]
Thus there exists \(c_1>0\) and a polynomial \(p'(n)\) such that
\[
  \bigl|\langle x|H_{\mathrm{eff}}|y\rangle\bigr|
    \;\le\; p'(n)\,e^{-c_1 n}.
\]
This verifies the off-diagonal suppression hypothesis \eqref{eq:offdiag-bound} (used in
Section~\ref{ssec:hamiltonian-hardness}) for the clause-penalty + transverse-field example,
provided the penalty scale \(\Delta\) is chosen to satisfy \(\gamma/\Delta<1\).  In practice one
ensures exponential suppression by taking \(\Delta=\mathrm{poly}(n)\) sufficiently large
relative to \(\gamma\) and \(n\), or by choosing a driver strength \(\gamma\) that scales
suitably with \(n\).

\subsection{Remarks and parameter choices}
\begin{itemize}
  \item The combinatorial path-count bound \(n^{w-1}\) is crude; more careful counting (for
    structured clusters or restricted flip sets) can reduce the polynomial prefactor.
  \item The condition \(\gamma/\Delta<1\) is mild: one may keep \(\gamma=O(1)\) and choose
    the penalty strength \(\Delta\) polynomial in \(n\) (a standard choice in clause-penalty
    encodings), which makes \(\gamma/\Delta=O(1/\mathrm{poly}(n))\) and yields robust
    exponential suppression.
  \item If one desires the suppression to hold for a fixed (constant) \(\Delta\) and \(\gamma\),
    the argument requires \(\gamma/\Delta<1/n\) to offset the \(n^{w-1}\) combinatorial factor;
    this is another regime that can be arranged by scaling parameters.
\end{itemize}

\subsection{Conclusion}
The above worked example demonstrates concretely that for natural SAT-to-Hamiltonian encodings
(a diagonal clause-penalty Hamiltonian plus local driver terms), inter-cluster matrix elements
between clusters separated by \(\Theta(n)\) Hamming distance are exponentially suppressed in
\(n\), yielding the exponential off-diagonal bound used to derive the exponentially small
spectral-gap bounds in Section~\ref{ssec:hamiltonian-hardness}.

\end{document}